\documentclass[10pt,sigconf]{acmart}
\usepackage{graphicx}
\usepackage{multirow}

\usepackage{amsmath}
\usepackage{amsfonts}
\usepackage{amssymb}
\usepackage{graphicx} %
\usepackage{color}
\usepackage{url}
\usepackage{enumitem} %
\usepackage{bm}                     	%
\usepackage{upgreek}                %
\usepackage{listings} %
\usepackage{bbding} %

\usepackage{tikz} %
\usetikzlibrary{matrix}
\usetikzlibrary{calc}
\usetikzlibrary{math}

\usepackage{tabularx} %
\usepackage{booktabs} %
\usepackage{colortbl}
\usepackage[]{xcolor}

\definecolor{dg}{cmyk}{0.60,0,0.88,0.27}

\definecolor{Blue1}{rgb}{0.85, 0.95, 1.0}
\definecolor{Red1}{rgb}{1.0, 0.85, 0.85}

\newcolumntype{u}{>{\columncolor{Blue1}}c}
\newcolumntype{v}{>{\columncolor{Blue1}}r}

\usepackage{rotating}		%

\usepackage{subcaption}		%
\usepackage[ruled,noend,linesnumbered]{algorithm2e} %

\DontPrintSemicolon     %
\SetNlSty{}{}{}                %
\SetAlgoInsideSkip{smallskip}   %
\SetAlCapSkip{1.5mm}             %
\SetInd{1.5mm}{1.5mm}             %

\SetAlFnt{\small}			%
\SetAlCapFnt{\small}		%
\SetAlCapNameFnt{\small}

\newtheorem{definition}{Definition}
\newtheorem{example}{Example}

\newtheorem{lemma}{Lemma}

\newtheorem{questionW}{Question}
\newtheorem{resultW}{Result}

{\end{itshape}
}
\newcommand{\hide}[1]{}

\newcommand{\later}[1]{}

\usepackage[capitalise,nameinlink]{cleveref}

\crefname{algocf}{alg.}{algs.}
\Crefname{algocf}{Algorithm}{Algorithms}
\crefalias{AlgoLine}{line}%
\crefname{proposition}{prop.}{prop.}
\crefname{figure}{Figure}{Figures}
\Crefname{example}{Example}{Example}

\crefformat{equation}{(#2#1#3)}						%
\crefrangeformat{equation}{(#3#1#4) to~(#5#2#6)}
\crefmultiformat{equation}{(#2#1#3)}%
{ and~(#2#1#3)}{, (#2#1#3)}{ and~(#2#1#3)}

\AtEndPreamble{%
    \hypersetup{colorlinks,
      linkcolor=purple,
      citecolor=blue,
      urlcolor=ACMDarkBlue,
      filecolor=ACMDarkBlue}}

\usepackage{tcolorbox}		%
\tcbuselibrary{breakable,skins}		%

\tcbset{examplestyle/.style={
		enhanced jigsaw,	%
		colback=blue!10,	%
		colframe=blue!10,	%
		arc=0mm,
		boxrule=0pt,		%
		left=1mm,
		right=1mm,
		left skip=-1mm,  %
		right skip=-1mm, %
		top=0pt,		%
		bottom=0pt,		%
		breakable,		%
		parbox = false,		%
		before={\par\pagebreak[0]\vspace{1mm}\parindent=0pt},		
		after={\par\pagebreak[0]\vspace{1mm}\parindent=0pt},				
		bottomrule = 0mm,
		boxsep = 0mm,					%
		topsep at break=0pt,			%
		bottomsep at break=0pt,			%
		pad at break=0mm,
		pad before break=0mm,		
		pad after break=1mm,		
		bottomrule at break=0mm,
		toprule at break=0mm,		
		}}

\tcolorboxenvironment{example}{examplestyle}

\tcbset{qrboxstyle/.style={
		enhanced jigsaw,	%
		colback=gray!20,
		colframe=gray!40,
		arc=0mm,
		boxrule=1pt,		%
		left=1pt,
		right=1pt,
		topsep at break=1mm,			%
		top=1pt,		%
		bottom=0mm,		%
		breakable,		%
		parbox = false		%
		}}

\newcommand{\resultbox}[1]{
\begin{tcolorbox}[
	enhanced jigsaw,		%
	colback=red!5,
	colframe=red!75!black,	
	arc=0mm,
	left skip=-1mm,
	right skip=-1mm,	
	left=0mm,
	topsep at break=1mm,			%
	right=0mm,
	top=0mm,
	bottom=0mm,		%
	breakable,		%
	parbox = false		%
]
\emph{#1}
\end{tcolorbox}
}

\newcounter{resultboxenv}

\newsavebox{\coloredbgbox}

\usepackage[normalem]{ulem} %
\newcommand\hl{\bgroup\markoverwith
  {\textcolor{yellow}{\rule[-.5ex]{2pt}{2.5ex}}}\ULon}

\usepackage{marginnote}

\let\footnotesize\scriptsize

\renewcommand\dbltopfraction{1} %
\newcommand{\Var}{\ensuremath{\mathbb{V}}}

\newcommand{\bigO}{\ensuremath{\mathcal{O}}}

\newcommand{\method}{\textsc{RecPart}\xspace}
\newcommand{\methodBase}{\textsc{RecPart-b}\xspace}
\newcommand{\methodS}{\textsc{RecPart-S}\xspace}

\newcommand{\methodVit}{\textsc{CS$_{\text{IO}}$}\xspace}
\newcommand{\methodGrid}{\textsc{Grid-$\varepsilon$}\xspace}
\newcommand{\methodOneB}{\textsc{1-Bucket}\xspace}
\newcommand{\methodGridOpt}{\textsc{Grid*}\xspace}
\newcommand{\parTime}{optimization time}
\newcommand{\ParTime}{optimization time}

\newcommand{\Imbal}{Imbalance\xspace}
\renewcommand{\epsilon}{\varepsilon}    %
\newcommand{\introparagraph}[1]{\textbf{#1.}}        %

\renewcommand{\epsilon}{\varepsilon}    %

\usepackage{scalerel}
\usepackage{tikz}
\usetikzlibrary{svg.path}

\definecolor{orcidlogocol}{HTML}{A6CE39}
\tikzset{
  orcidlogo/.pic={
    \fill[orcidlogocol] svg{M256,128c0,70.7-57.3,128-128,128C57.3,256,0,198.7,0,128C0,57.3,57.3,0,128,0C198.7,0,256,57.3,256,128z};
    \fill[white] svg{M86.3,186.2H70.9V79.1h15.4v48.4V186.2z}
                 svg{M108.9,79.1h41.6c39.6,0,57,28.3,57,53.6c0,27.5-21.5,53.6-56.8,53.6h-41.8V79.1z M124.3,172.4h24.5c34.9,0,42.9-26.5,42.9-39.7c0-21.5-13.7-39.7-43.7-39.7h-23.7V172.4z}
                 svg{M88.7,56.8c0,5.5-4.5,10.1-10.1,10.1c-5.6,0-10.1-4.6-10.1-10.1c0-5.6,4.5-10.1,10.1-10.1C84.2,46.7,88.7,51.3,88.7,56.8z};
  }
}

\RequirePackage{etoolbox}
\DeclareRobustCommand\orcidicon[1]{\href{https://orcid.org/#1}{\mbox{\scalerel*{
\begin{tikzpicture}[yscale=-1,transform shape]
\pic{orcidlogo};
\end{tikzpicture}
}{|}}}}

\renewcommand{\arraystretch}{0.95}

\copyrightyear{2020}
\acmYear{2020}
\setcopyright{acmlicensed}\acmConference[SIGMOD'20]{Proceedings of the 2020 ACM SIGMOD International Conference on Management of Data}{June 14--19, 2020}{Portland, OR, USA}
\acmBooktitle{Proceedings of the 2020 ACM SIGMOD International Conference on Management of Data (SIGMOD'20), June 14--19, 2020, Portland, OR, USA}
\acmPrice{15.00}
\acmDOI{10.1145/3318464.3389750}
\acmISBN{978-1-4503-6735-6/20/06}

\renewcommand\footnotetextcopyrightpermission[1]{} %

\begin{document}
\fancyhead{}					%

\title[Near-Optimal Distributed Band-Joins through Recursive Partitioning]{Near-Optimal Distributed Band-Joins\\through Recursive Partitioning}

\author{Rundong Li}
\authornote{Work performed while PhD student at Northeastern University.}
\orcid{0000-0003-3142-2559}
\affiliation{\orcidicon{0000-0003-3142-2559} \institution{Google, USA}}
\email{eadon.lee@gmail.com}

\author{Wolfgang Gatterbauer}
\orcid{0000-0002-9614-0504}
\affiliation{\orcidicon{0000-0002-9614-0504} Northeastern University, USA}
\email{w.gatterbauer@northeastern.edu}

\author{Mirek Riedewald}
\orcid{0000-0002-6102-7472}
\affiliation{\orcidicon{0000-0002-6102-7472} Northeastern University, USA}
\email{m.riedewald@northeastern.edu}

\begin{abstract}
We consider running-time optimization for band-joins in a distributed system, e.g., the cloud. 
To balance load across worker machines, input has to be partitioned, which causes duplication. 
We explore how to resolve this tension between \emph{maximum load per worker}
and \emph{input duplication} for band-joins between two relations. 
Previous work suffered from high optimization cost or considered partitionings
that were too restricted (resulting in suboptimal join performance).
Our main insight is that \emph{recursive partitioning of the join-attribute space}
with the appropriate split scoring measure can achieve
both low optimization cost and low join cost.
It is the first approach that is not only effective for one-dimensional band-joins
but also for joins on multiple attributes.
Experiments indicate that our method
is able to find partitionings that are within 10\% of the \emph{lower bound} for both
maximum load per worker and input duplication for a broad range of settings,
significantly improving over previous work.
\end{abstract}




\maketitle

\section{Introduction}
\label{sec:intro}

Given two relations $S$ and $T$, the band-join $S \bowtie_B T$ returns all pairs $(s \in S, t \in T)$
that are ``close'' to each other. Closeness is determined based on band-width constraints
on the join attributes which we also call \emph{dimensions}.
This is related to (but in some aspects more general than, and in others a special case of)
\emph{similarity joins} (see \Cref{sec:related}).
The Oracle Database SQL Language Reference guide~\cite{OracleSQLlanguage19c} presents a
one-dimensional (1D) example of a query finding employees whose salaries differ by at most
\$100. Their discussion of band-join specific optimizations highlights the operator's importance.
Zhao et al~\cite{zhaoRDW16:arraySimJoin} describe an astronomy application where celestial objects are matched
using band conditions on time and coordinates \emph{ra} (right ascension) and \emph{dec} (declination).
This type of approximate matching based on space and time is very common in practice and leads to
three-dimensional (3D) band-joins like this:
\begin{example} \label{ex:eBird}
Consider bird-observation table \texttt{B} with columns 
\texttt{longitude}, 
\texttt{latitude}, 
\texttt{time}, \texttt{species}, 
\texttt{count}, and weather table \texttt{W},
reporting precipitation and temperature for location-time combinations.
A scientist studying how weather affects bird sightings wants to join these tables 
on attributes \texttt{longitude}, \texttt{latitude}, and \texttt{time}.
Since weather reports do not cover the exact time and location of the bird sighting,
she uses a band-join to link each bird report with
weather data for ``nearby'' time and location, e.g.,
$|\mathtt{B.longitude} - \mathtt{W.longitude}| \le 0.5$ \texttt{AND}
$|\mathtt{B.latitude} - \mathtt{W.latitude}| \le 0.5$ \texttt{AND}
$|\mathtt{B.time} - \mathtt{W.time}| \le 10$.
\end{example}

\begin{figure}[tb]
\includegraphics[width=\linewidth]{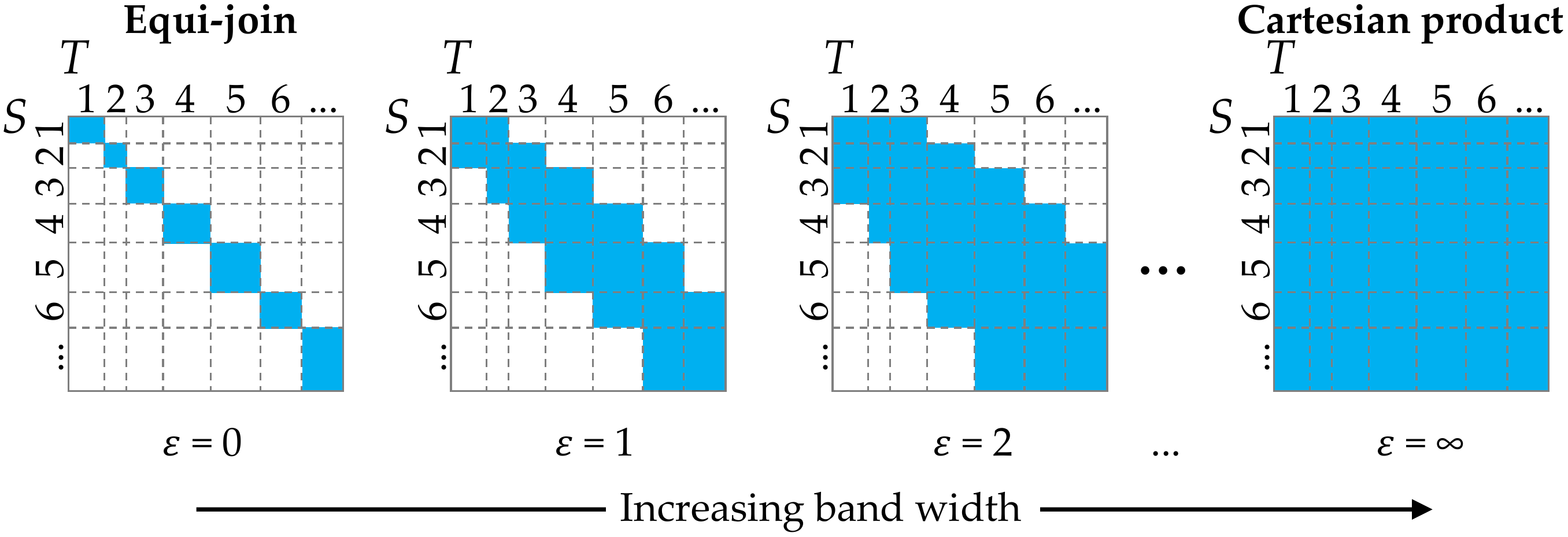}
\caption{Join matrices for one-dimensional (1D) band-join $|S.A - T.A| \le \varepsilon$
for increasing band width $\varepsilon$ from equi-join ($\varepsilon=0$) to Cartesian product ($\varepsilon=\infty$).
Numbers on matrix rows and columns indicate distinct $A$-values of input tuples. Cell $(i, j)$ corresponds to attribute
pair $(s_i, t_j)$ and is shaded iff the pair fulfills the join condition and is in the output.}
\label{fig:bandJoinSpectrum}
\end{figure}

We are interested in minimizing end-to-end running time of distributed band-joins,
which is the sum of (1) \emph{optimization time} (for finding a good execution strategy) and 
(2) \emph{join time} (for the join execution).
Join time depends on the \emph{data partitioning}
used to assign input records to worker machines.
As seen in \Cref{fig:bandJoinSpectrum}, band-joins generalize equi-join and Cartesian product.
Partitioning algorithms with optimality guarantees exist only for these two
extremes~\cite{beame2014skew,li2018submodularity,Okcan:theta-join}.

\begin{figure}[t]
\centering
    \begin{subfigure}[t]{0.47\linewidth}
	\centering
	\includegraphics[scale=0.28]{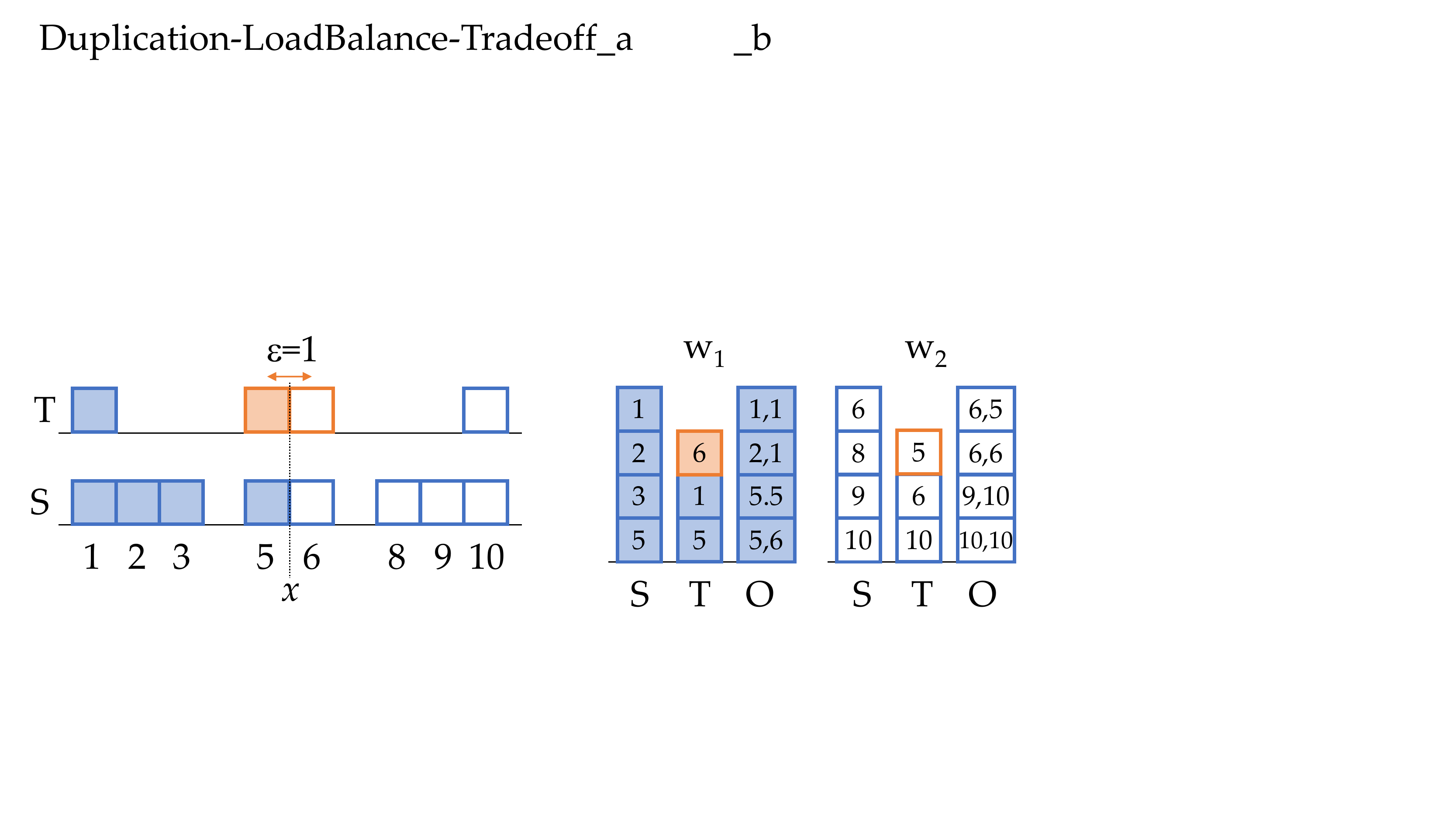}
	\caption{Input distribution}
	\label{fig:Duplication-LoadBalance-Tradeoff_a}
	\end{subfigure}
	\hspace{2mm}
    \begin{subfigure}[t]{0.47\linewidth}
	\centering
	\includegraphics[scale=0.28]{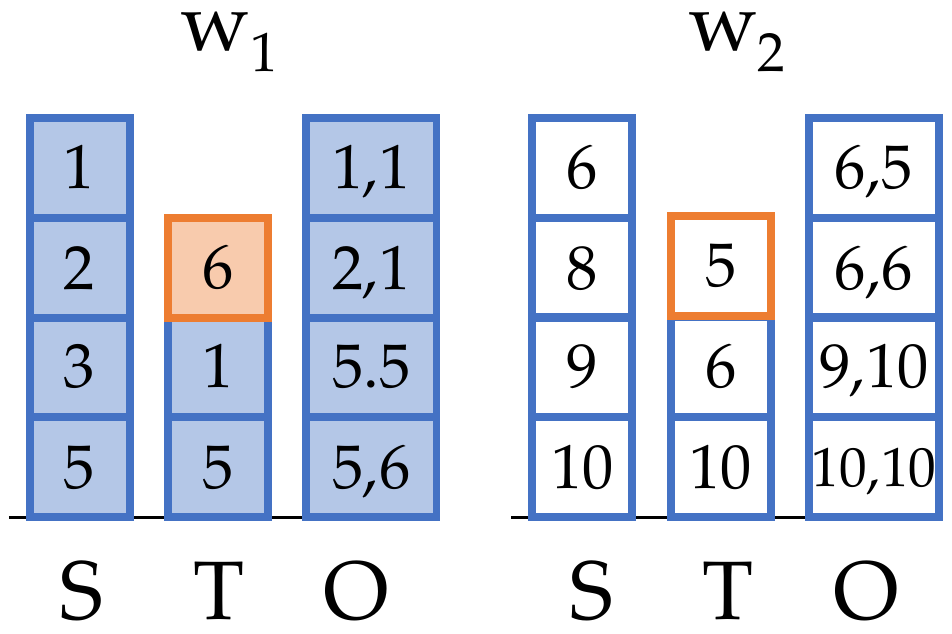}
	\caption{Load distribution}
	\label{fig:Duplication-LoadBalance-Tradeoff_b}
	\end{subfigure}
    \begin{subfigure}[t]{0.47\linewidth}
	\centering
	\includegraphics[scale=0.28]{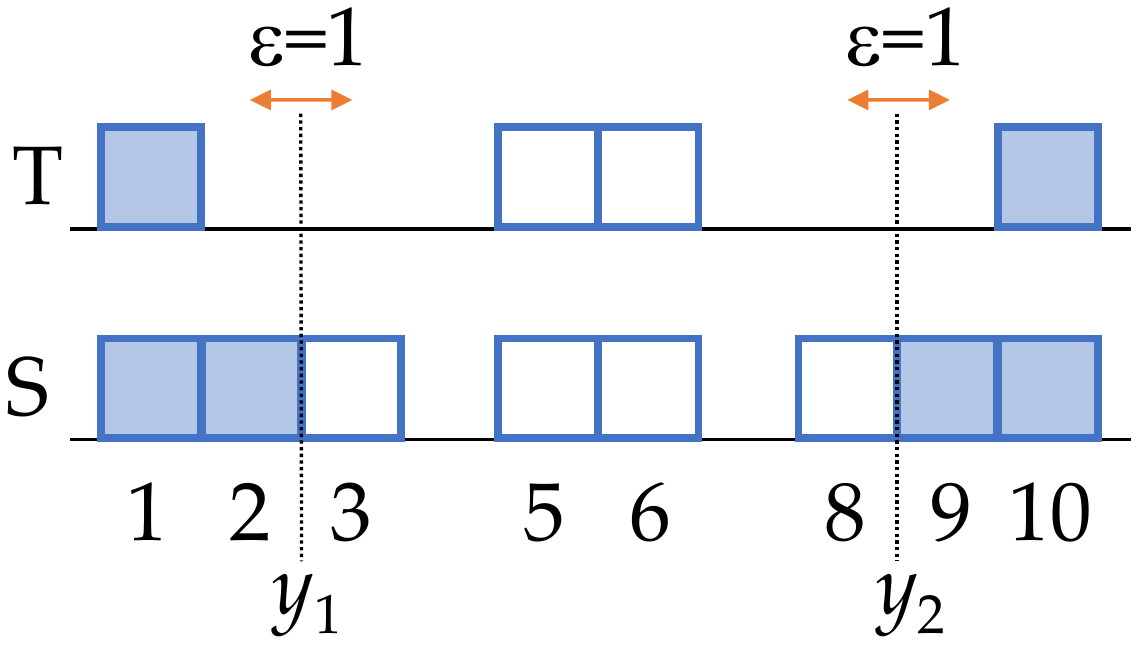}
	\caption{Input distribution}
	\label{fig:Duplication-LoadBalance-Tradeoff_e}
	\end{subfigure}
	\hspace{2mm}
    \begin{subfigure}[t]{0.47\linewidth}
	\centering
	\includegraphics[scale=0.28]{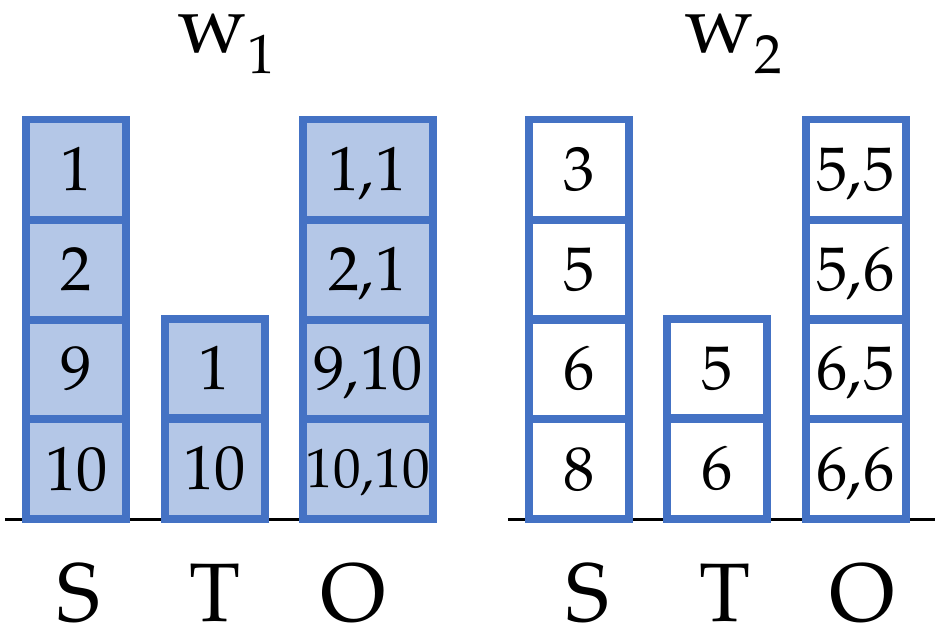}
	\caption{Load distribution}
	\label{fig:Duplication-LoadBalance-Tradeoff_f}
	\end{subfigure}
\caption{Input ($S$, $T$) and output ($O$) on workers $w_1$ and $w_2$ when
splitting on $x$ (top row). The $T$-tuples shown in orange are duplicated because they are within
the band width of the split point. When splitting on $y_1$ and $y_2$, no tuple is duplicated
and load is perfectly balanced (bottom row).}
\label{fig:dupVsBalanceX}
\end{figure}

\begin{example}\label{ex:motivating_splittingPoints}
To see why distributed band-joins are difficult,
consider a 1D join with band width $\epsilon = 1$ of $S=\{1,2,3,5,6,8,9,10\}$ and $T=\{1,5,6,10\}$
on $w=2$ workers. For balancing load, we may split $S$ on value $x$
and send the left half to worker $w_1$ and the right half to $w_2$
(see \cref{fig:Duplication-LoadBalance-Tradeoff_a}).
To not miss results near the split point, all $T$-tuples within band width $\varepsilon = 1$
of $x$ have to be copied across the boundary.
\Cref{fig:Duplication-LoadBalance-Tradeoff_b} shows the resulting input and output tuples
on each worker, with duplicates in orange.
By splitting in sparse regions of $T$, e.g., on $y_1$ and $y_2$
(\cref{fig:Duplication-LoadBalance-Tradeoff_e,fig:Duplication-LoadBalance-Tradeoff_f}), 
perfect load balance can be achieved without input duplication.
\end{example}

\emph{The main contribution of this work is a novel algorithm \method (\textbf{Rec}ursive \textbf{Part}itioning)
that quickly and efficiently finds split points such as $y_1$ and $y_2$.} To do so, it has to carefully
navigate the tradeoff between load balance and input duplication. For instance, $y_1$ by itself appears like a poor
choice from a load-balance point of view. It takes the additional split on $y_2$ to unlock $y_1$'s true potential.

\begin{figure}[tb]
\centering
\includegraphics[width=0.9\linewidth]{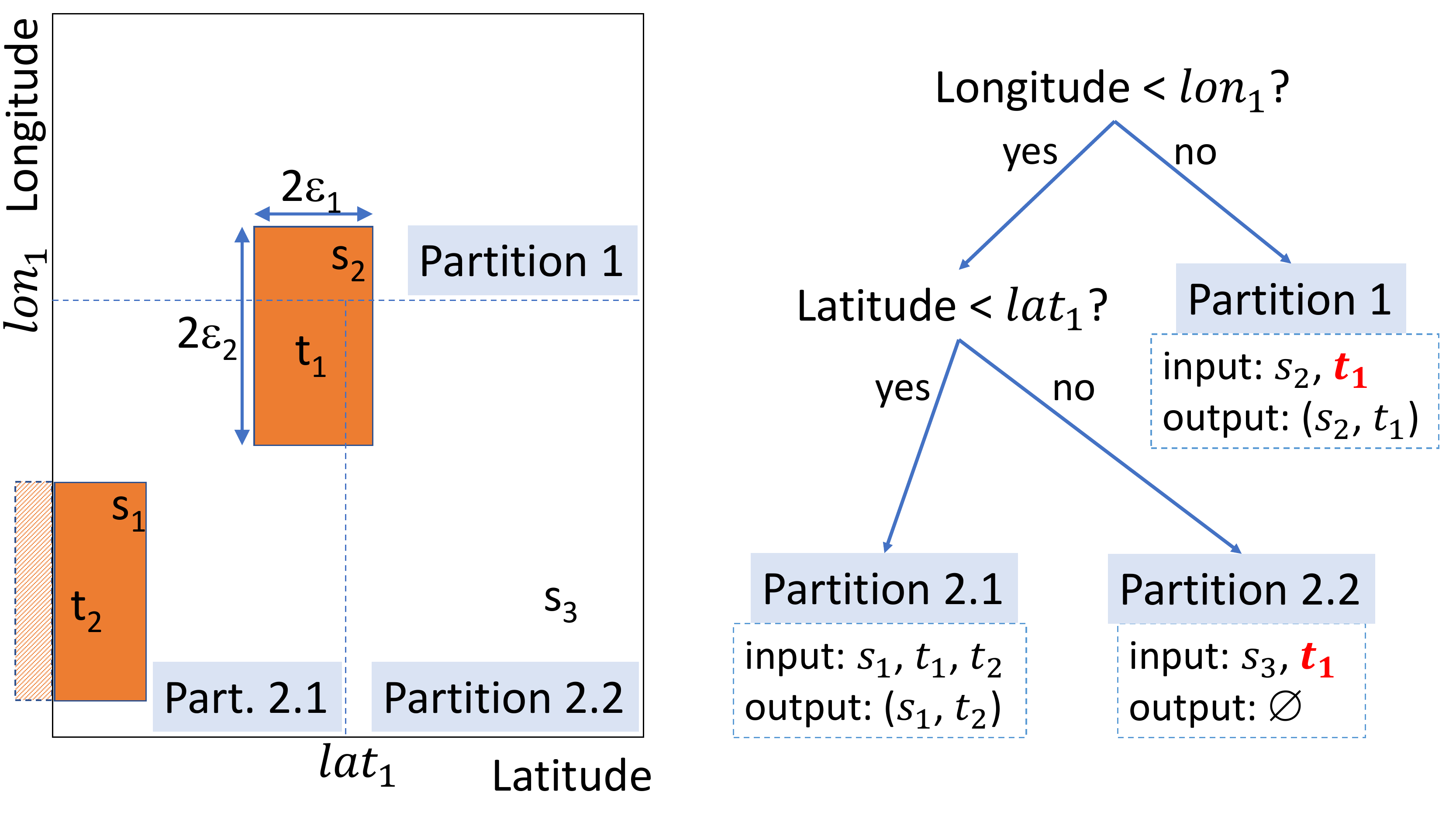}
\caption{Recursive \emph{split tree} for a 2D band-join on latitude and longitude. All splits are $T$-splits, i.e.,
$T$-tuples within band width of the split boundary are sent to both children. For instance, at the root
the $\varepsilon$-range (orange box) of $t_1$ crosses the $\mathrm{lon}_1$ line and therefore the tuple is copied
to the right sub-tree (partition 1). The same happens again for the split on $\mathrm{lat}_1$.
This ensures that no match is missed (e.g., $(s_2, t_1)$) and no output tuple is produced twice.
}
\label{fig:latLonPartitions}
\end{figure}

\introparagraph{Overview of the approach}
We propose recursive (i.e., hierarchical) partitioning of the join-attribute space, because it
offers a broad variety of partitioning options that can be explored efficiently.
As illustrated by the \emph{split tree} in \Cref{fig:latLonPartitions}, each path from the root node to a leaf defines a partition of the
join-attribute space as a conjunction of the split predicates in the nodes along the path.
Like \emph{decision trees} in machine learning~\cite{han2011data}, \method's split tree is grown
from the root, each time recursively splitting some leaf node. This step-wise expansion is a perfect
solution for the problem of \emph{navigating two optimization goals}: minimizing max worker load
and minimizing input duplication. As \method grows the tree, input duplication is monotonically
increasing, because more tuples may have to be copied across a newly added split boundary. At the same time,
large partitions are broken up and hence load balance may improve.

To find good partitionings, it is important to (1) use an appropriate scoring function to pick
a good split point (e.g., choose $y_1$ or $y_2$ over $x$ in \Cref{fig:dupVsBalanceX})
and to (2) choose the best leaf to be split next. We propose the \emph{ratio between
load balance improvement and additional input duplication} for both decisions.
In \Cref{ex:motivating_splittingPoints}, this would favor $y_1$ and $y_2$
over $x$, because they add zero duplication. Similarly, a leaf node with a zero-duplicate split
option would be preferred over a leaf whose split would cause duplication.

When a leaf becomes so small that virtually all tuples in the corresponding partition join with
each other, then it is not split any further. However, if the load induced by that partition
is high, then the leaf ``internally'' uses a grid-style partitioning inspired by
 \methodOneB~\cite{Okcan:theta-join} to create more fine-grained partitions. This is motivated by the
observation that the band-join computation in a sufficiently small partition behaves like
a Cartesian product---for which \methodOneB was shown to be near-optimal.

\introparagraph{Main contributions}
(1) We demonstrate analytically and empirically that previous work falls short either due to high optimization time
(to find a partitioning) or due to high join time (caused by an inferior partitioning), especially for band-joins
in more than one dimension.

(2) To address those shortcomings, we propose \emph{recursive partitioning} of the multidimensional
join-attribute space. Given a fixed-size input and output sample, our
algorithm \method finds a partitioning in $\bigO(w \log w + w d)$, where $w$ is the number of workers
and $d$ is the number of join attributes, i.e., the dimensionality of the band-join.
\method is inspired by decision trees~\cite{han2011data}, which had not been explored in the context
of optimizing running time of distributed band-joins. To make them work for our problem, we identify
a \emph{new scoring measure} to determine the best split points: \emph{ratio of load variance reduction
to input duplication increase}. It is informed by our observation that a good split
should improve load balance with minimal additional input duplication.
We also identify a \emph{new stopping condition} for tree growth.

\begin{figure}[tbp]
\centering
\includegraphics[width=0.7\linewidth]{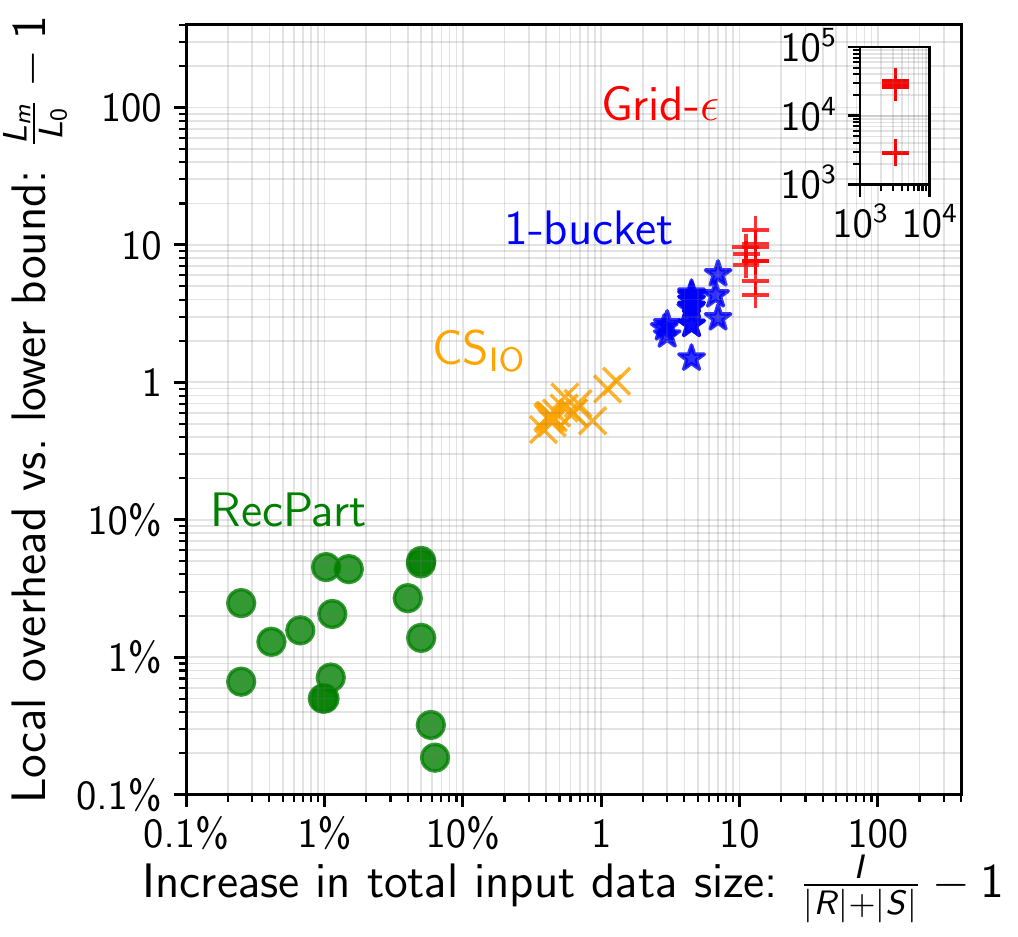}
\caption{Total input duplication (x-axis) and maximum overhead across workers (y-axis) 
for a variety of data points
for our method \method vs.\ 3 competitors
(see \cref{sec:exp} for details).
\method is always within 10\% of the lower bounds 
(0\% duplication and 0\% overhead).}
\label{fig:comp-overhead}
\end{figure}

(3) While we could not prove near-optimality of \method's partitioning, our experiments
provide strong empirical evidence. Across a variety of datasets,
cluster sizes, and join conditions, \method always found partitions for which both 
total input duplication 
and
max worker load 
were within 10\% of the corresponding \emph{lower bounds}, beating
all competitors (even those with significantly higher optimization cost) by a wide margin.
\Cref{fig:comp-overhead} shows this for a variety of problems
(notice the log scale).
The definition of the axes, the algorithms, and the detailed experiments are presented in \Cref{sec:problemDef},
\Cref{sec:related}, and \Cref{sec:exp}, respectively.

(4) We prove new \Cref{lem:gridBalanceBad,lem:gridBalanceGood} that characterize
when grid partitioning will be effective for distributed band-joins.

\section{Problem Definition}
\label{sec:problemDef}

Without loss of generality, let $S$ and $T$ be two relations with the same schema $(A_1, A_2,\ldots, A_d)$.
Given a band width $\varepsilon_i \ge 0$ for each attribute $A_i$, the band-join of $S$ and $T$ is defined as
\[
S \bowtie_B T = \{(s, t): s \in S \land t \in T \land \forall_{1 \le i \le d} |s.A_i - t.A_i| \le \varepsilon_i\}.
\]
We call $d$ the \emph{dimensionality} of the join and $A_i$ the $i$-th \emph{dimension}.
We refer to the $d$-dimensional hyper-rectangle centered around a tuple $a$ with side-length
$2\varepsilon_i$ in dimension $i$, formally
$\{(x_1,\ldots, x_d):\; \forall_{1 \le i \le d}\, a.A_i - \varepsilon_i \le x_i \le a.A_i + \varepsilon_i\}$,
as the \emph{$\varepsilon$-range} around $a$
(depicted as orange box in \Cref{fig:latLonPartitions}). 
Note that $(s, t)$ is in the output iff $s$ falls into the $\varepsilon$-range around $t$ (and vice versa). 
It is straightforward to generalize all results in this paper to asymmetric band conditions
($a.A_i - \varepsilon_{i_L} \le x_i \le a.A_i + \varepsilon_{i_R}$) and to
relations with attributes that do not appear in the join condition.

\begin{definition}[Join Partitioning]\label{def:joinPartitioning}
	Given input relations $S$ and $T$ with $Q = S \bowtie_B T$ and $w$ worker machines.
	A \emph{join partitioning} is an assignment $h: (S \cup T) \rightarrow 2^{\{1,\ldots, w\}}\setminus \emptyset$
	of each input tuple to one
	or more workers so that each join result $q \in Q$ can be recovered by exactly one local join.
\end{definition}

A \emph{local join} is the band-join executed by a worker on the input subset it receives.
The definition ensures that \emph{each output tuple is produced by exactly one worker},
which avoids an expensive post-processing phase for duplicate elimination and is in line
with previous work~\cite{beame2014skew,khayyat2017fast,Okcan:theta-join,vitorovic2016load}.

\begin{figure}[tb]
\centering
\includegraphics[width=\linewidth]{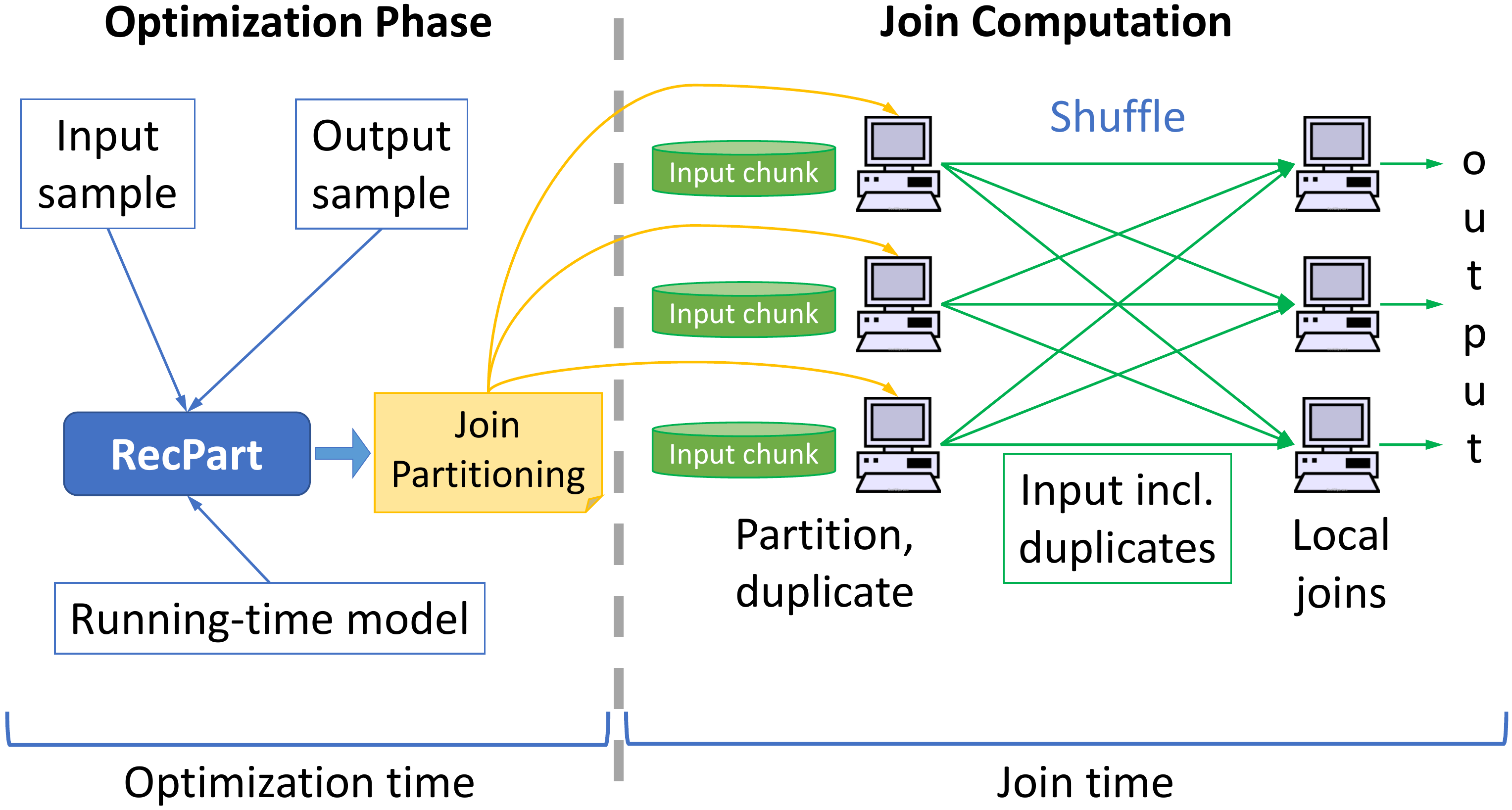}
\caption{Overview of the proposed approach.}
\label{fig:systemOverview}
\end{figure}

Given $S$, $T$, and a band-join condition,
our goal is to minimize the time to compute $S \bowtie_B T$.
This time is the sum of \emph{optimization time} (i.e., the time to find a join partitioning) 
and \emph{join time} (i.e., the time to compute
the result based on this partitioning), as illustrated in \Cref{fig:systemOverview}.

We follow common practice and define the load
$L_i$ on a worker $w_i$ as the weighted sum
$L_i = \beta_2 I_i + \beta_3 O_i$, $0 \le \beta_2, 0 \le \beta_3$
of input $I_i$ and output $O_i$ assigned to
it~\cite{beame2014skew,li2018submodularity,Okcan:theta-join,vitorovic2016load}.
\emph{Max worker load} $L_m = \max_i L_i$ is the maximum load
assigned to any worker. In addition, we also evaluate a partitioning based on its
\emph{total amount of input} $I$. It accounts for given inputs $S$ and $T$ and all duplicates
created by the partitioning, i.e., $I = \sum_{x \in S \cup T} |h(x)|$.
Recall from \cref{def:joinPartitioning} that $h$ assigns input tuples to a subset of the workers, i.e., $|h(x)|$ is the number of workers
that receive tuple $x$.

\begin{lemma}[Lower Bounds]\label{lem:lowerBounds}
$|S|+|T|$ is a lower bound for total input $I$.
And $L_0 = (\beta_2 (|S|+|T|) + \beta_3 |S \bowtie_B T|) / w$
is a lower bound for max worker load $L_m$.
\end{lemma}

The lower bound for total input $I$ follows from \Cref{def:joinPartitioning}, because each input tuple has to be examined by
at least one worker. For max worker load, note that any partitioning has to distribute a total input of at least $|S|+|T|$
and a total output of $|S \bowtie_B T|$ over the $w$ workers, for a total load of at least
$\beta_2 (|S|+|T|) + \beta_3 |S \bowtie_B T|$.

\introparagraph{System Model and Measures of Success}
We consider the standard Hadoop MapReduce and Spark environment where inputs $S$ and $T$ are stored in
files. These files may have been chunked up before the computation starts, with chunks distributed
over the workers. Or they may reside in a separate cloud storage service such as Amazon's S3. The files
are not pre-partitioned on the join attributes and the workers do not have advance knowledge which chunk
contains which join-attribute values. Hence any desired join partitioning requires
that---in MapReduce terminology---(1) the entire input
is read by map tasks and (2) a full shuffle is performed to group the data according to the partitioning and
have each group be processed by a reduce task.

In this setting, the shuffle time is determined by total input $I$, 
and the duration of the reduce phase (each worker performing local joins)
is determined by max worker load. Hence we are interested in evaluating
how close a partitioning comes to the lower bounds for total input and max worker load (\Cref{lem:lowerBounds}).
We use $\frac{I-(|S|+|T|)}{|S|+|T|}$ and $\frac{L_m-L_0}{L_0}$, respectively, which measure by how much
a value exceeds the lower bound, relative to the lower bound. For instance, for $L_m=11$ and $L_0=10$ we
obtain 0.1, meaning that the max worker load of the partitioning is 10\% higher than the lower bound.

For systems with very fast networks, an emerging
trend~\cite{binnigCGKZ16:endOfSlowNetworks,rodiger2016flow}, data transfer time is
negligible compared to local join time, therefore the goal is to minimize max worker load, i.e.,
the success measure is $\frac{L_m-L_0}{L_0}$.
In applications where the input is already \emph{pre-partitioned on the join attributes}
(e.g., for dimension-dimension array joins~\cite{duggan2015skew}) the optimization goal
concentrates on reducing data movement~\cite{zhaoRDW16:arraySimJoin}. There our approach can be used
to \emph{find the best pre-partitioning}, i.e., to chunk up the array on the dimensions.

In addition to comparing to the lower bounds, we also measure end-to-end running time for a MapReduce/Spark
implementation of the band-join in the cloud. For join-time estimation we rely on the model
by Li et al.~\cite{Li2018DAPD}, which was shown to be sufficiently accurate to
optimize running time of various algorithms, including equi-joins. Similar to the equi-join model, our band-join
model $\mathcal{M}$ takes as input triple $(I, I_m, O_m)$ and estimates join time as
a piecewise linear model
$\mathcal{M}(I, I_m, O_m) = \beta_0 + \beta_1 I + \beta_2 I_m + \beta_3 O_m$.
The $\beta$-coefficients
are determined using linear regression on a small benchmark of training queries and inputs.

\section{Related Work}
\label{sec:related}

\subsection{Direct Competitors}
\label{sec:mainCompetitors}

\begin{figure}
\centering
\setlength{\tabcolsep}{1mm}
\renewcommand{\arraystretch}{1.1}
\small

\begin{tabular}{| c |c|c|c|c|} 
\hline
	\multicolumn{2}{|c|}{{Join-Matrix Covering}} 
	& \multicolumn{2}{c|}{{Attribute-Space Partitioning}} 
	& \multirow{2}{*}{$d$} 	
\\\cline{1-4}
	\methodOneB \cite{Okcan:theta-join}
	& \methodVit~\cite{vitorovic2016load} 
 	& \methodGrid \cite{dewitt1991evaluation,soloviev1993truncating} 
	& \method 
	&
\\\hline
    \multirow{3}{*}{\raisebox{-12mm}[0mm][0mm]{ \includegraphics[width=0.2\linewidth]{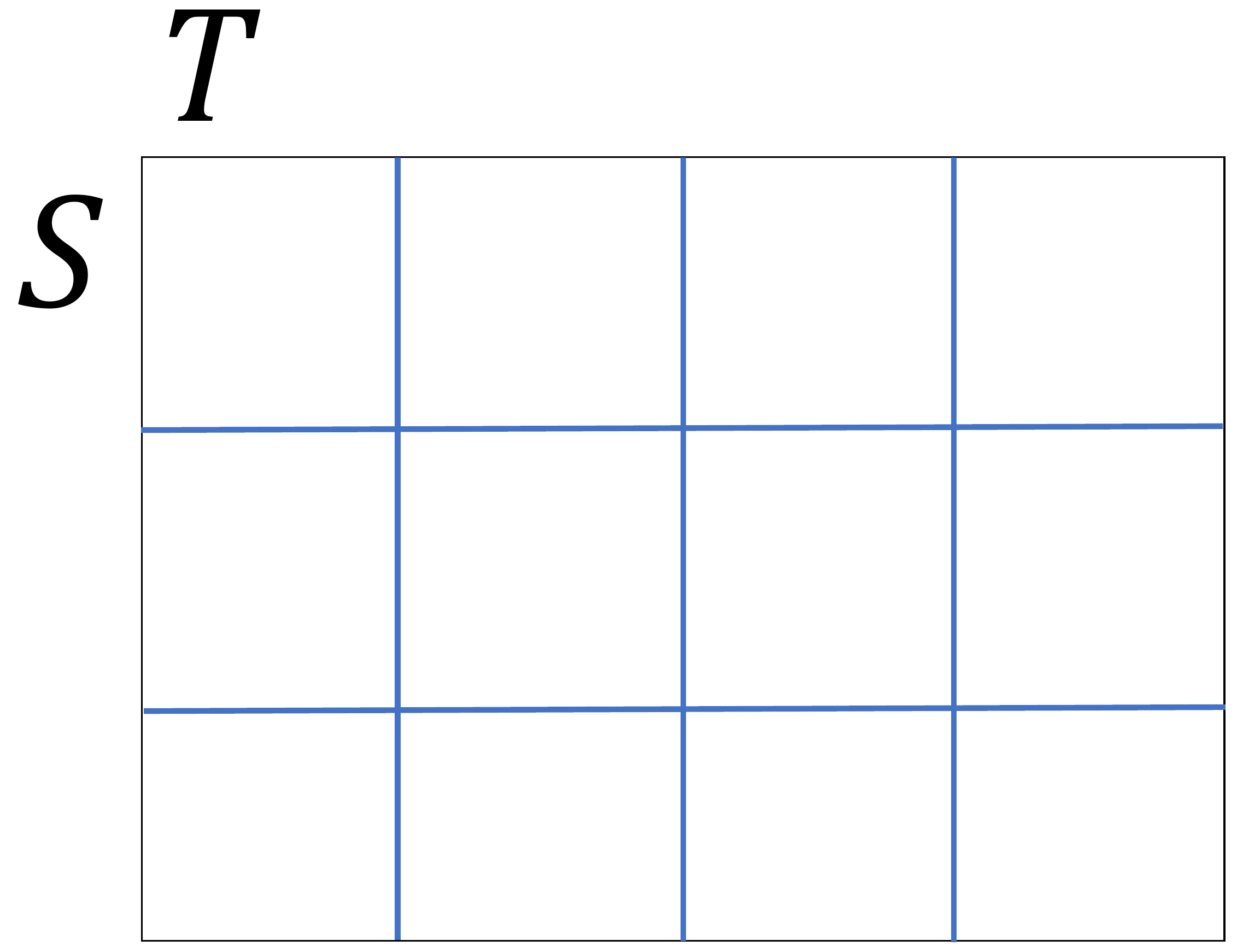} }}  
    & \rule{0pt}{12ex} 
	  \includegraphics[width=0.2\linewidth]{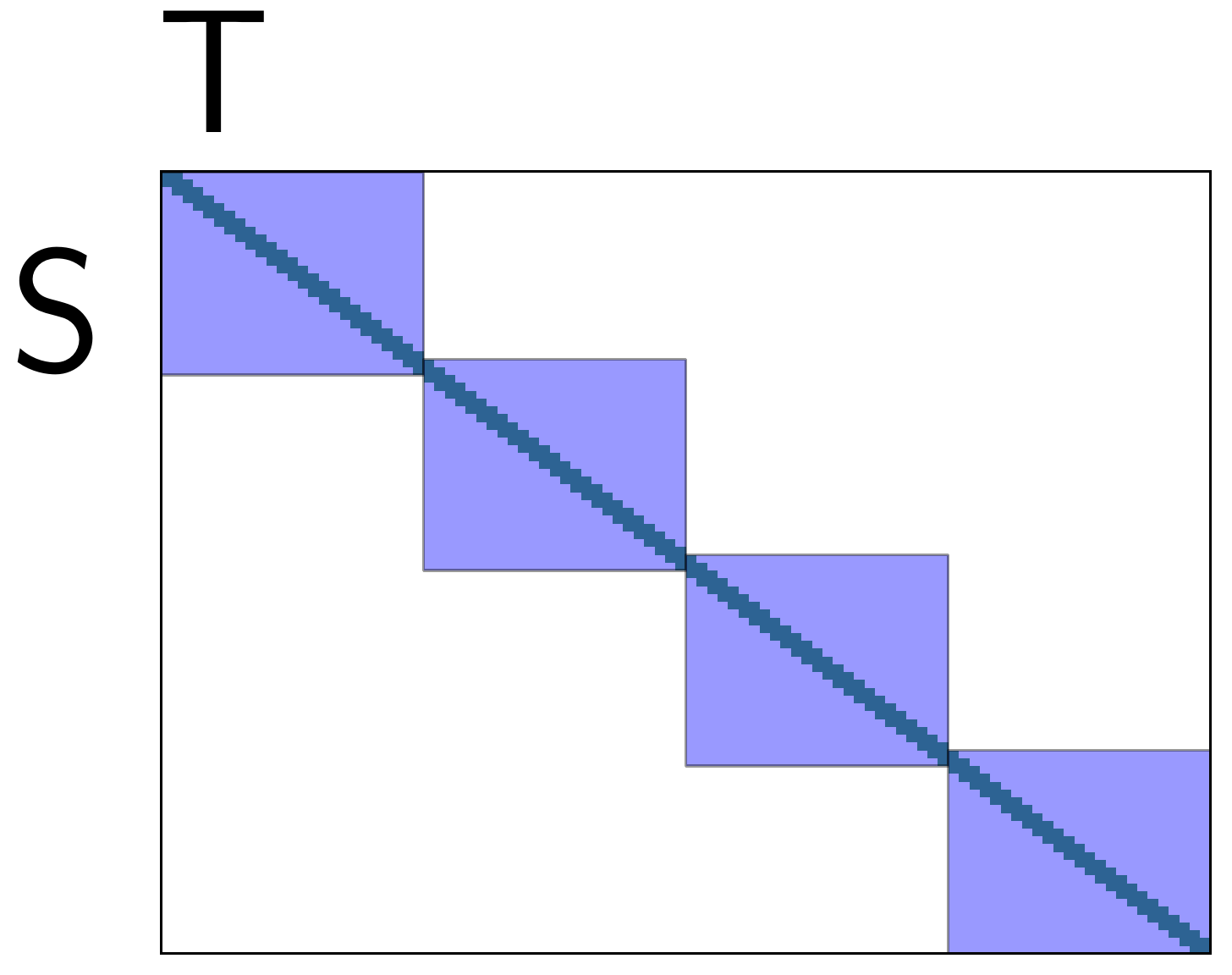} 
	& \includegraphics[width=0.2\linewidth]{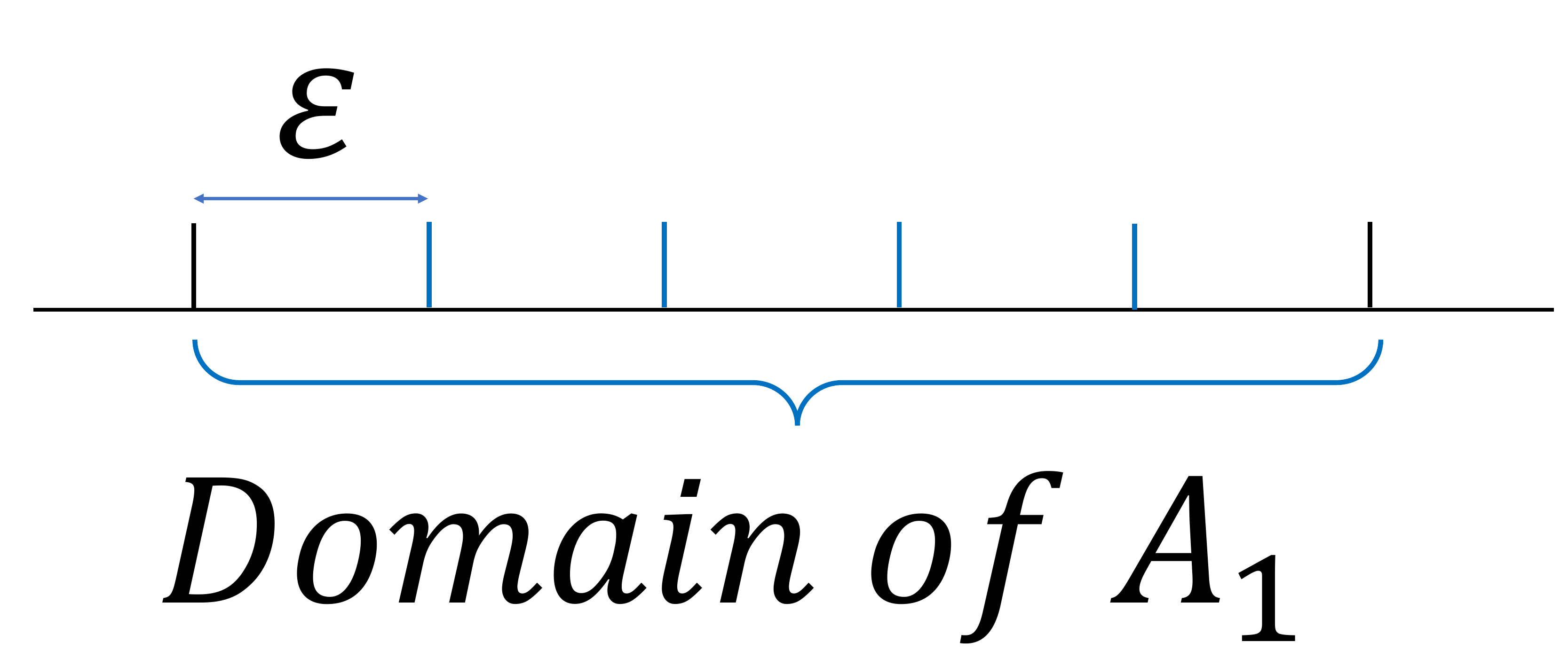} 
    & \includegraphics[width=0.2\linewidth]{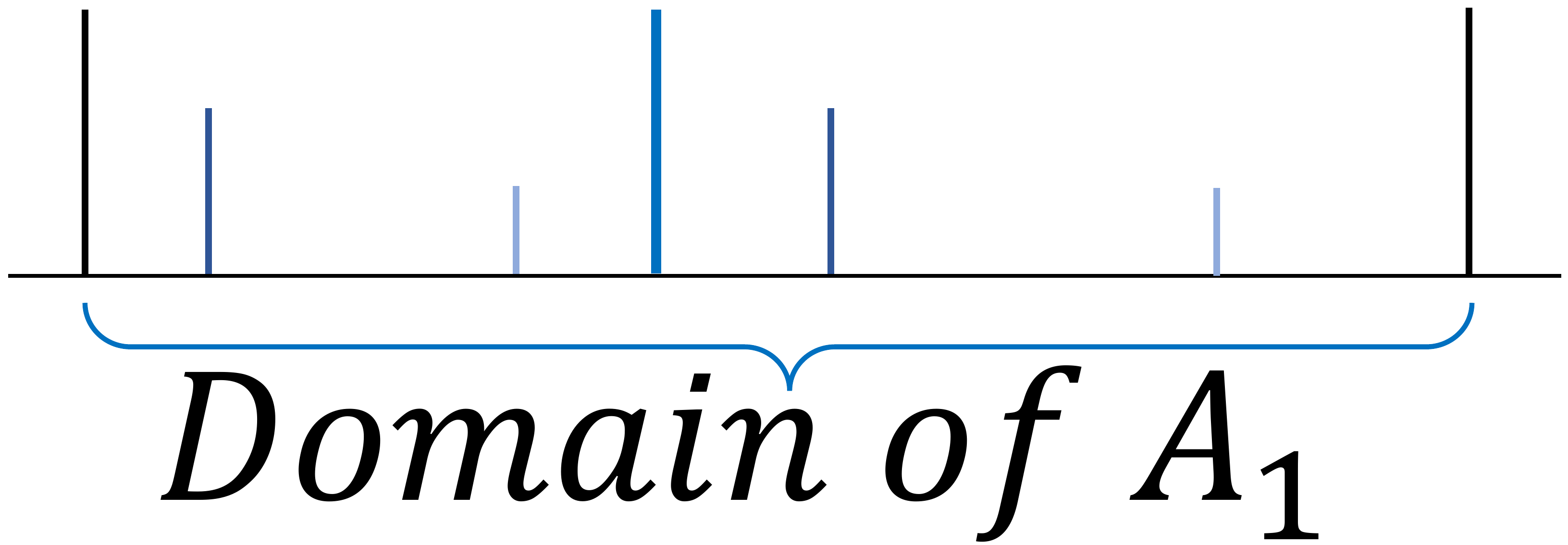} 
	& \raisebox{5mm}[0mm][0mm]{1}   	
\\\cline{2-5}
    & \rule{0pt}{12ex}
	  \includegraphics[width=0.2\linewidth]{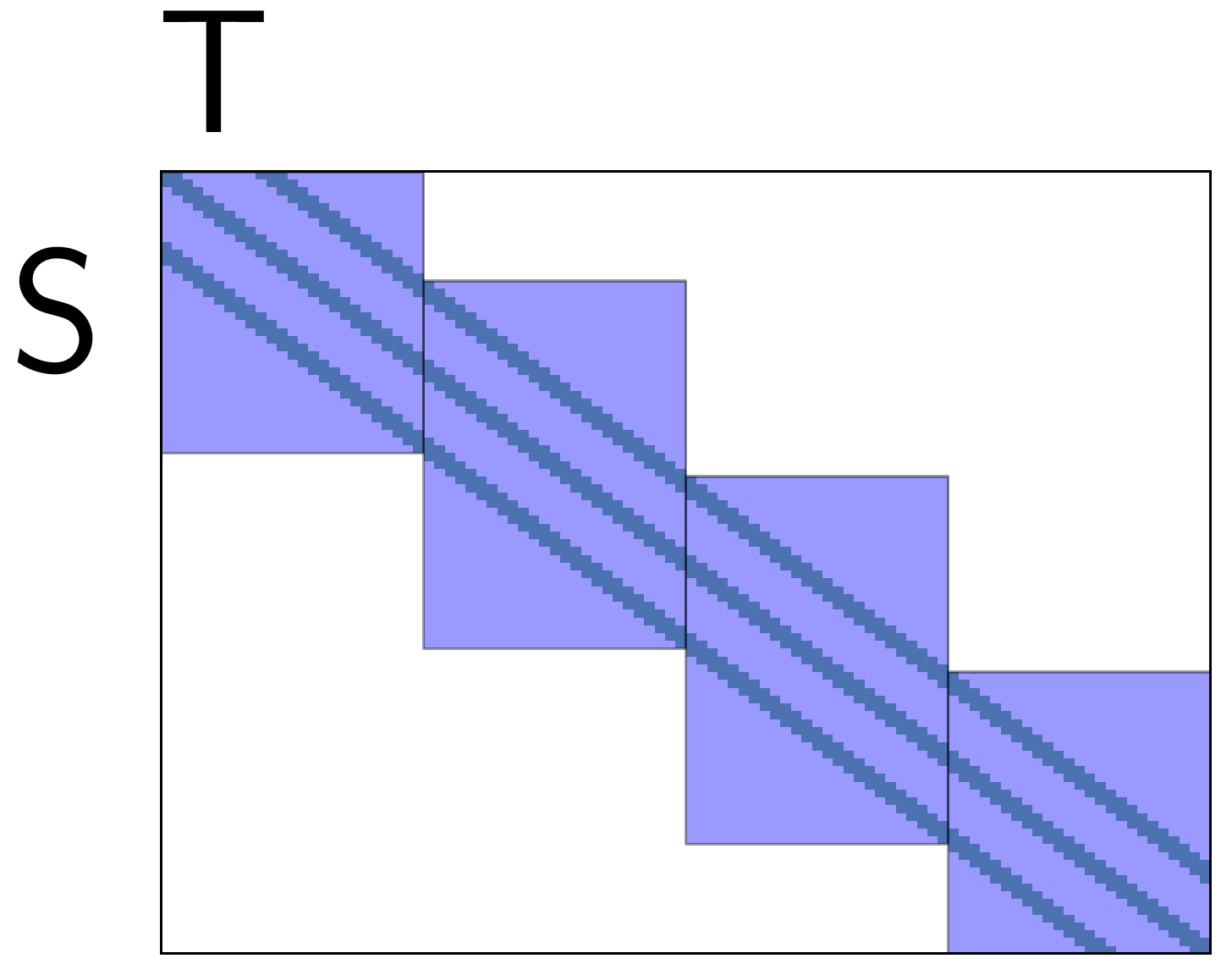} 
	& \includegraphics[width=0.2\linewidth]{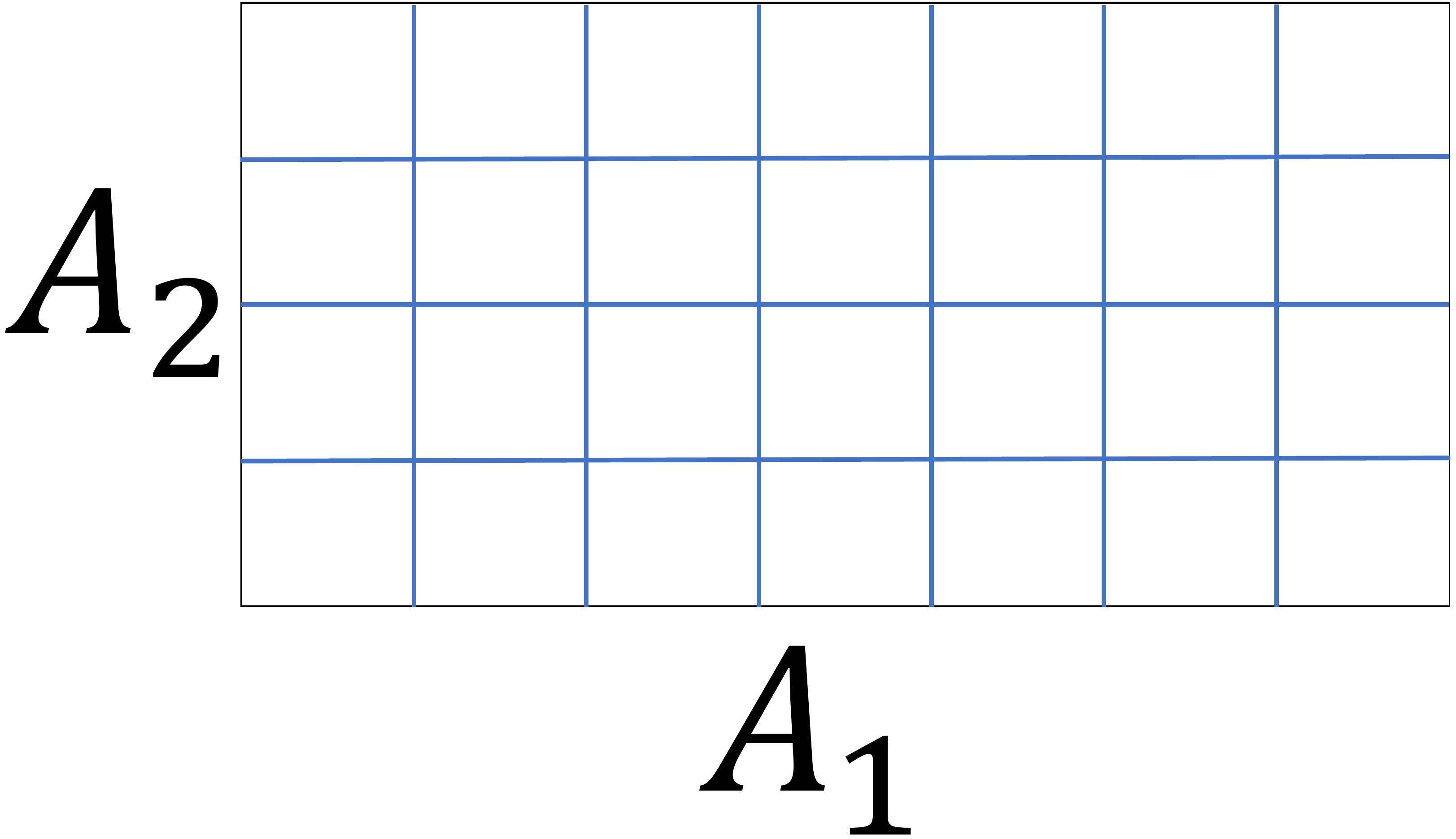} 
    & \includegraphics[width=0.2\linewidth]{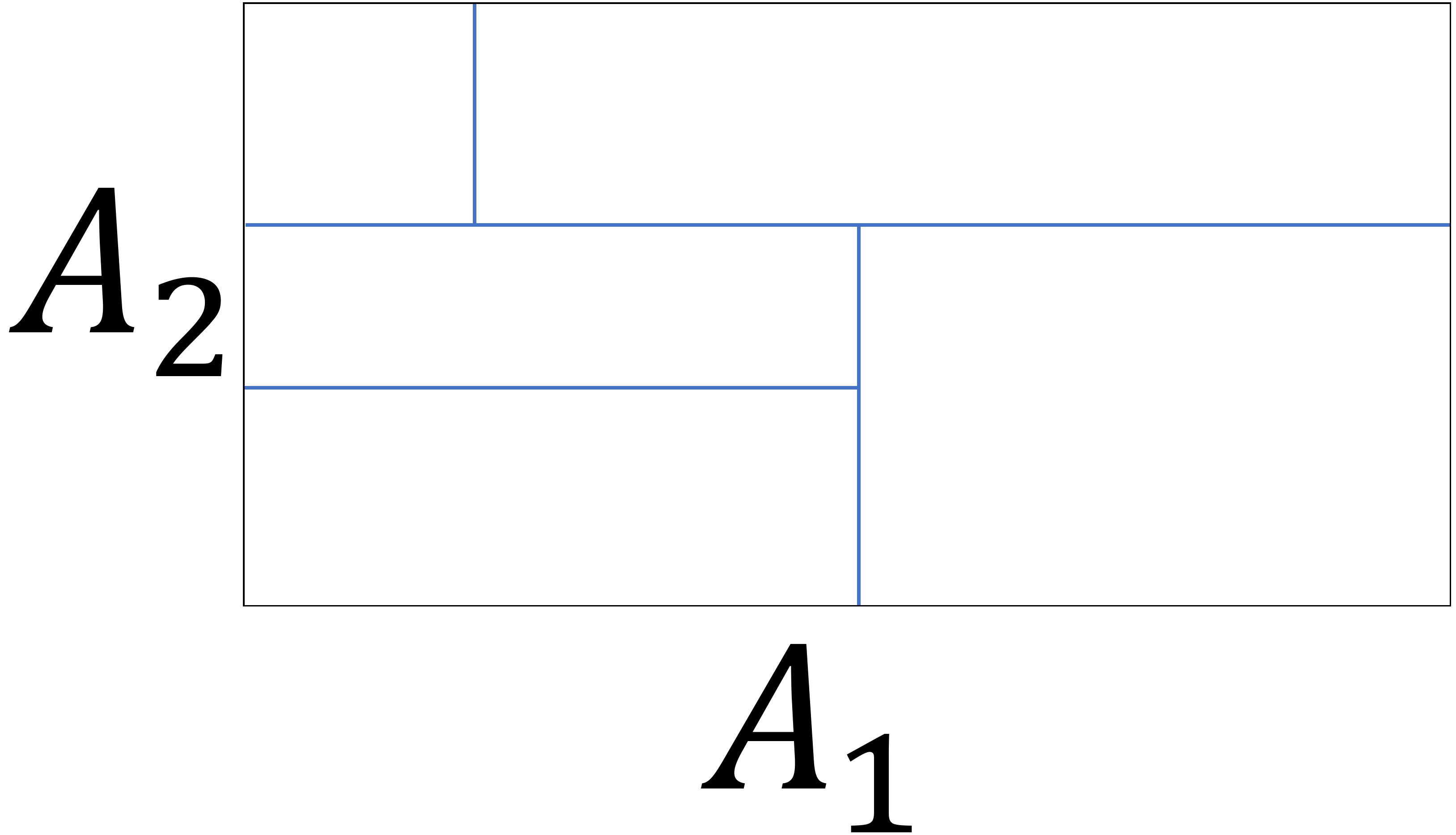} 
	& \raisebox{5mm}[0mm][0mm]{2}
\\\cline{2-5}
    & \rule{0pt}{12ex} 
	  \includegraphics[width=0.2\linewidth]{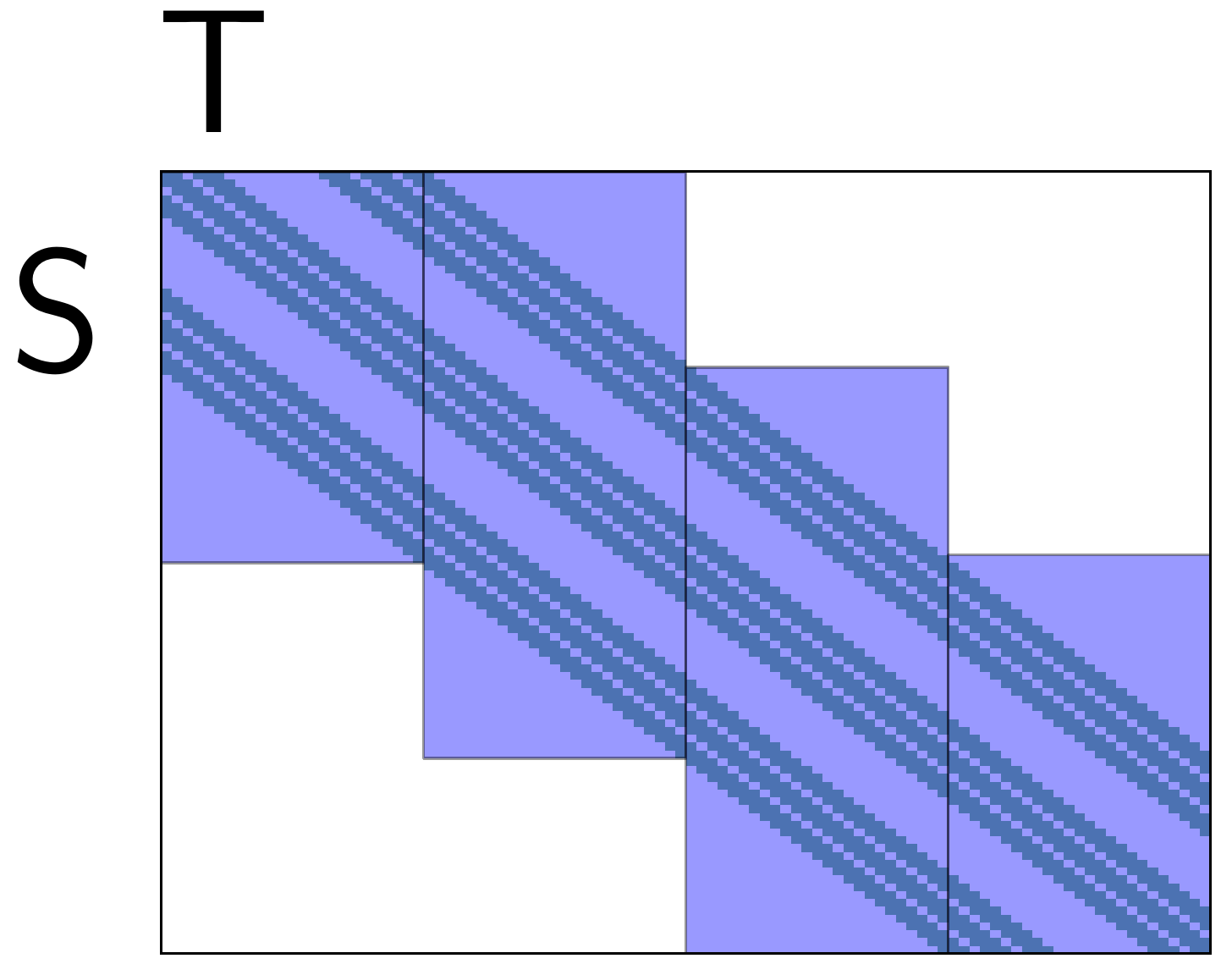}
    & \includegraphics[width=0.2\linewidth]{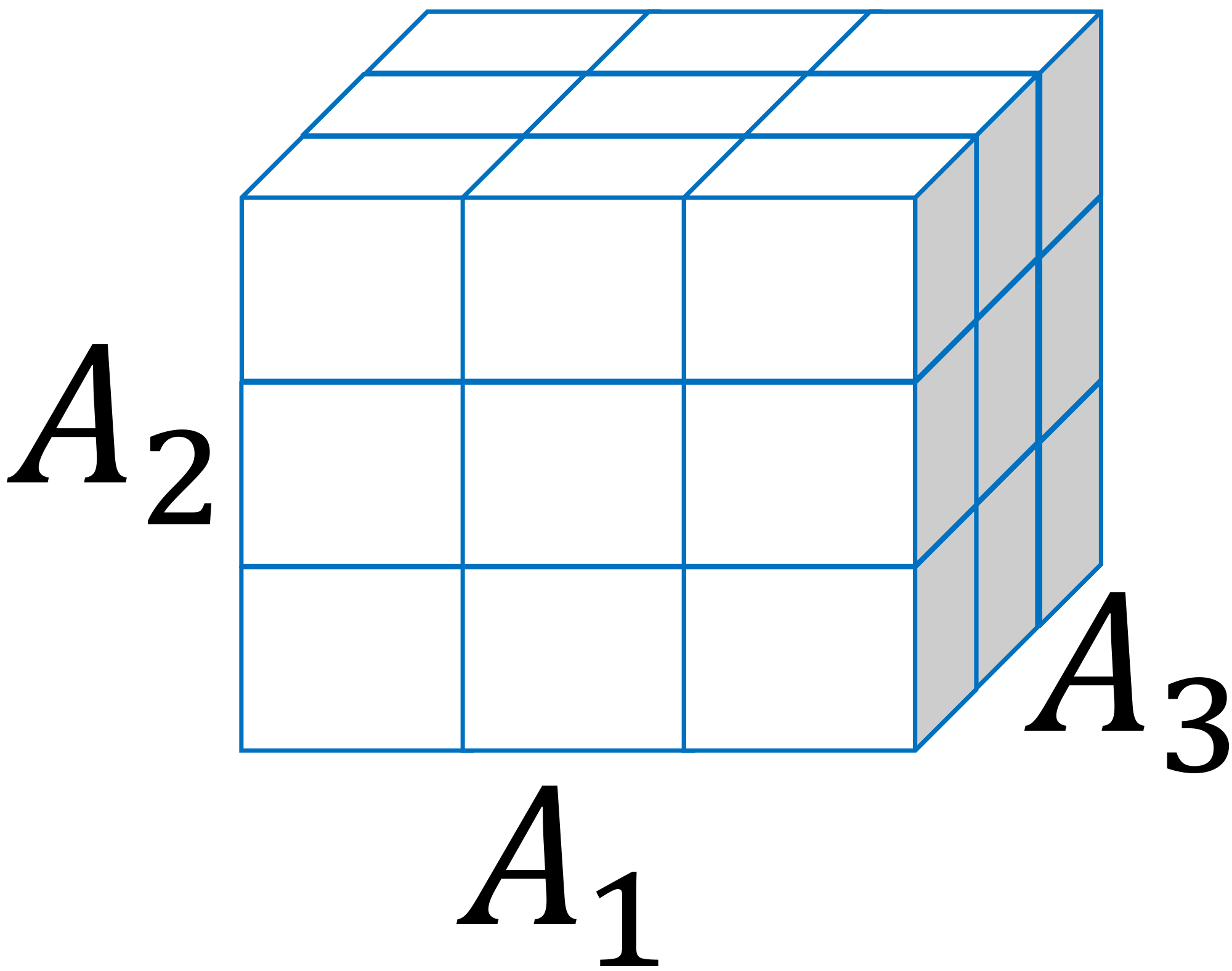} 
    & \includegraphics[width=0.2\linewidth]{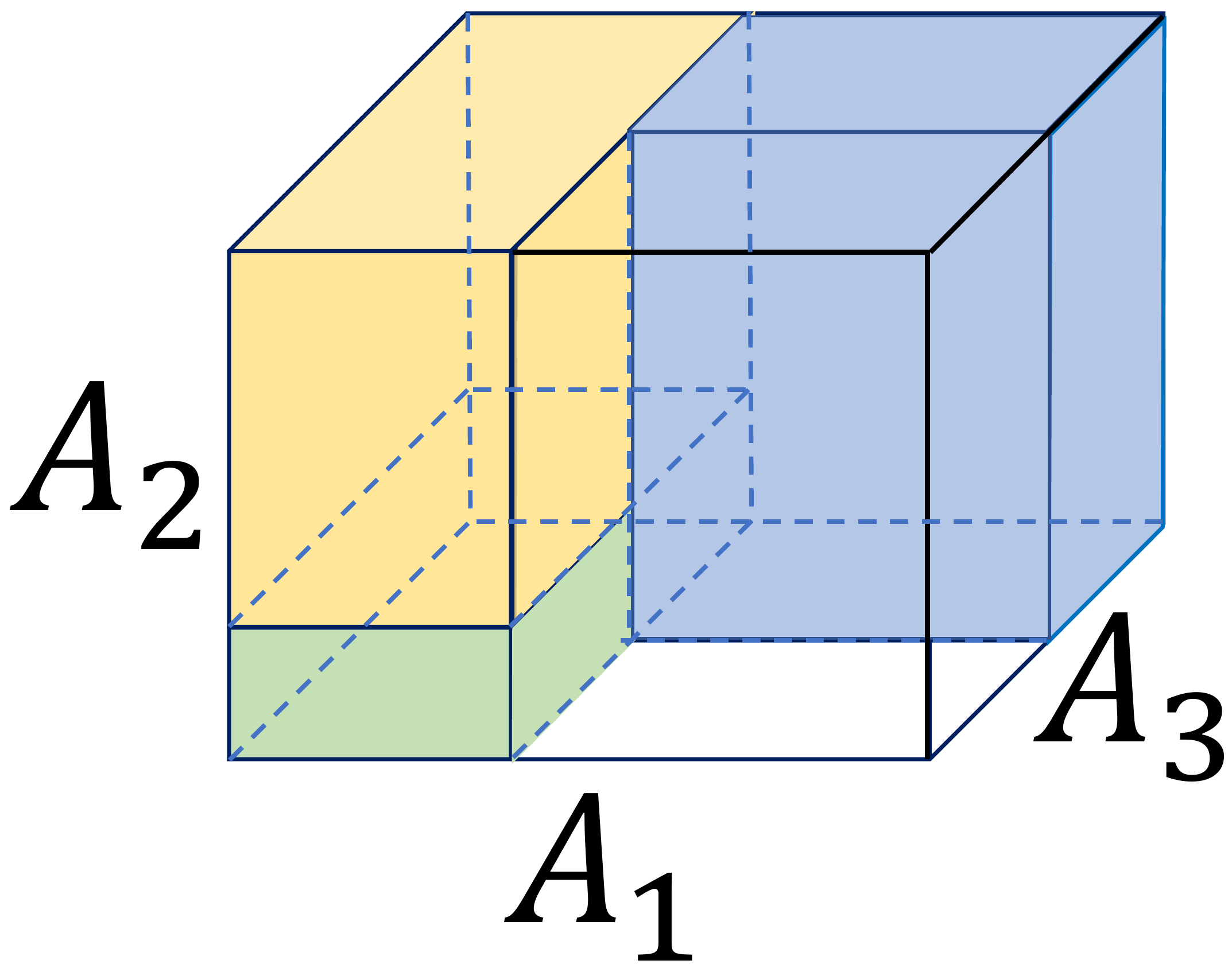} 
    & \raisebox{5mm}[0mm][0mm]{3}
\\\hline
\end{tabular}

\caption{
Illustration of partitioning methods for band-joins in $d$-dimensional space for $d = 1, 2, 3$; the $A_i$ are the join attributes. 
\methodGrid and \method partition the $d$-dimensional \emph{join-attribute space}, 
while \methodVit and \methodOneB create partitions by finding a cover of the 2-dimensional \emph{join matrix} 
$S \times T$, 
whose dimensions are independent of the dimensionality of the join condition. 
Bar height for \method and $d=1$ indicates recursive partition order.
}
\label{tab:partitioning}

\end{figure}

Direct competitors are approaches that (1) support distributed band-joins and (2) optimize
for load balance, i.e., max worker load or a similar measure. We classify them
into \emph{join-matrix covering} vs \emph{attribute-space partitioning}.

\introparagraph{Join-matrix covering}
These approaches model distributed join computation as a covering problem
for the \emph{join matrix}
$J = S \times T$, whose rows correspond to $S$-tuples and columns to $T$-tuples. A cell $J(s, t)$
is ``relevant'' iff $(s, t)$ satisfies the join condition. Any theta-join, including band-joins,
can be represented as the corresponding set of relevant cells in $J$.
A join partitioning can then be obtained by covering all relevant cells with non-overlapping
regions. Since the exact set of relevant cells is not known a priori (it corresponds to the
to-be-computed output), the algorithm covers a larger region of the matrix that
is guaranteed to contain all relevant cells. For instance, for inequality predicates,
M-Bucket-I~\cite{Okcan:theta-join} partitions both inputs on approximate quantiles
in one dimension and then covers with $w$ rectangles all regions corresponding to
combinations of inter-quantile ranges from $S$ and $T$ that could potentially
contain relevant cells. IEJoin~\cite{khayyat2017fast} directly uses the same
quantile-based range partitioning, but without attempting to find a $w$-rectangle
cover. Its main contribution is a clever \emph{in-memory} algorithm for queries with
two join predicates. Optimizing local processing is orthogonal to our focus on how
to \emph{assign input tuples to multiple workers}. In fact, one can use
an adaptation of their idea for local band-join computation on each worker.

To support any theta-join, \methodOneB~\cite{Okcan:theta-join} covers the entire join
matrix with a grid of $r$ rows and $c$ columns.
This is illustrated for $r=3$ and $c=4$ in \Cref{tab:partitioning}. 
Each $S$-tuple is randomly assigned to one of the $r$ rows
(which implies that it is sent to \emph{all} $c$ partitions in this row);
this process is analogous for $T$-tuples, which are assigned to random columns.
While randomization achieves near-perfect load balance,
input is duplicated approximately $\sqrt{w}$ times.

Zhang et al.~\cite{Zhang2012:multiwayThetaJoin} extend \methodOneB to
joins between many relations.
Koumarelas et al.~\cite{Koumarelas2018:selectiveThetaJoins} explore re-ordering of
join matrix rows and columns to improve the running time of M-Bucket-I~\cite{Okcan:theta-join}.
However, like M-Bucket-I, their technique does not take output distribution into account.
This was shown to lead to poor partitionings by Vitorovic et al.~\cite{vitorovic2016load}
whose method \methodVit represents the state of the art for distributed theta-joins.
It relies on a carefully tuned optimization pipeline that first range-partitions $S$ and $T$
using approximate quantiles, then coarsens those partitions, and finally finds the
optimal (in terms of max worker load) rectangle covering of the coarsened matrix.
The resulting partitioning was shown to be superior---including for band-joins---to
direct quantile-based partitioning, which is used by IEJoin.
\Cref{tab:partitioning} illustrates \methodVit for a covering with four rectangles for 1, 2 and 3 dimensions. 
The darker diagonal ``bands'' show relevant matrix cells, i.e., cells that have to be covered.
Notice how join dimensionality affects relevant-cell locations, but does not
affect the dimensionality of the join matrix:
\emph{for a join between two input relations $S$ and $T$, the join matrix
is always two-dimensional, with one dimension per input relation}.

\methodVit suffers from high optimization cost to find the covering rectangles,
which uses a tiling algorithm of complexity $\bigO(n^5 \log n)$ for $n$ input tuples.
Optimization cost can be reduced by coarsening the statistics used.
Further reduction in optimization cost is achieved for monotonic join matrices, a property that holds for
1-dimensional band-joins but \emph{not for multidimensional ones}. As our experiments will show,
the high optimization cost hampers the approach for multidimensional band-joins.

\introparagraph{Attribute-space partitioning}
Instead of using the 2-dimensional $S \times T$ join matrix, attribute-space partitioning works in
the $d$-dimensional space
$A_1 \times A_2 \times\cdots\times A_d$ defined by the domains of the join attributes. 
Grid partitioning of the attribute space was explored in the early days of parallel databases,
yet only for one-dimensional conditions. Soloviev~\cite{soloviev1993truncating} proposes
the \emph{truncating hash algorithm} and shows that it improves over a parallel
implementation of the hybrid partitioned band-join algorithm by
DeWitt et al.~\cite{dewitt1991evaluation}.
The method generalizes to more dimensions as illustrated in the \methodGrid column in \Cref{tab:partitioning}.
Grid cells define partitions and are assigned to the workers.

By default, \methodGrid sets grid size for attribute $A_i$ to the band width $\varepsilon_i$ in that
dimension. This results in near-zero optimization cost, but may create a poor load balance
(for skewed input) and high input duplication (when a partition boundary cuts through a dense region).
A coarser grid reduces input duplication, but the larger partitions make load balancing more challenging.
Our approach \method, which also applies attribute-space partitioning, mitigates the problem
by considering recursive partitionings that avoid cutting through dense regions.

\subsection{Other Related Work}

\textbf{Similarity joins} are related to band-joins, but neither generalizes the other: the similarity joins closest to band-joins
define a pair $(s \in S, t \in T)$ as similar if $\mathrm{sim}(s,t) > \theta$,
for some similarity function $\mathrm{sim}$ and threshold $\theta$.
This includes \emph{1D} band-joins as a special case, but does not support band-joins in multiple dimensions.
(A band-join in $d$ dimensions has $2d$ threshold parameters for lower and upper limits in each dimension.)
A recent survey~\cite{FierABLF18:setSimilarityJoinExperimentSurvey} compares 10 distributed
set-similarity join algorithms. The main focus of previous work on similarity joins is on addressing the specific
challenges posed by working in a \emph{general metric space} where vector-space operations such as addition and
scalar multiplication (which band-joins can exploit) are not available.
A particular focus is on (1) identifying fast filters that prune away
a large fraction of candidate pairs without computing their similarity and 
(2) selecting pivot elements
or anchor points to form partitions, e.g., via sampling~\cite{SarmaHC14:clusterJoin}.

Duggan et al~\cite{duggan2015skew} study skew-aware optimization for \textbf{distributed array equi-joins}
(not band-joins). The work by Zhao et al~\cite{zhaoRDW16:arraySimJoin} is closest to ours, because they
introduce array similarity joins that can encode multi-dimensional band-join conditions.
However, it is not a direct competitor for \method, because it considers a different optimization problem:
The array is assumed to be already grid-partitioned into chunks on the join attributes and the main challenge
is to co-locate with minimal network cost those partitions that need to be joined.
Our approach is orthogonal for two reasons: 
First, we do not make any assumptions
about existing pre-partitioning on the join attributes and hence the join requires a full data shuffle.
Second, we show that for band-joins, grid partitioning is inferior to \method's recursive partitioning.
Hence \method provides new insights for choosing better array partitions when the array DBMS
anticipates band-join queries.

\textbf{Attribute-space partitioning} is explored in other contexts for optimization goals that are very different
from distributed band-join optimization. For array tiling, the goal is to minimize page accesses
of range queries~\cite{FurtadoB99:arrayTiling}. Here, like for histogram
construction~\cite{poosalaIHS96improvedHistograms,poosala:1997:selectivity}, the band-join's
data duplication across a split boundary is not taken into account.
Histogram techniques optimize for
a different goal: maximizing the information captured with a given number of partitions.
Only the equi-weight histograms by
Vitorovic et al.~\cite{vitorovic2016load} take input duplication into account.
We include their approach \methodVit in our comparison.

For equi-joins, several algorithms address skew by \textbf{partitioning heavy hitters}~\cite{beame2014skew,bruno2014advanced,dewitt1992practical,lu1994load,polychroniou2014track,poosala1996estimation,wolf1994new,xu2008handling}.
Other than the high-level idea of splitting up large partitions to improve load balance, the concrete approaches do
not carry over to band-joins: They rely on the property that tuples with different join values cannot be
matched, i.e., do not capture that tuples within band width of a split boundary
must be duplicated. However, our decision to use load variance for measuring load balance was inspired by the
state-of-the-art equi-join algorithm of Li et al.~\cite{li2018submodularity}. Earlier work relied on hash
partitioning~\cite{dewitt1985multiprocessor,kitsuregawa1983application} and focused on assigning partitions to processors~\cite{dewan1994predictive,elseidy2014scalable,harada1995dynamic,hua1991handling,kitsuregawa1990bucket,rodiger2016flow,walton1991taxonomy}. 
Empirical studies of parallel and distributed equi-joins include~\cite{blanas2010comparison,chu2015theory,dewitt1985multiprocessor,schneider1989performance,schneider1990tradeoffs}.

\section{Recursive Partitioning}
\label{sec:method}

We introduce \method and analyze its complexity.

\subsection{Main Structure of the Algorithm}

\begin{algorithm}[tb]
\caption{\method}
\label{alg:adapFull}
\SetAlgoLined
\LinesNumbered
\KwData{$S$, $T$, band-join condition, sample size $k$}
\KwResult{Hierarchical partitioning $\mathcal{P}^*$ of $A_1 \times\cdots\times A_d$}
Draw random input sample of size $k/2$ from $S$ and $T$\;
Draw random output sample~\cite{vitorovic2016load} of size $k/2$\;
Initialize $\mathcal{P}$ with root partition $p_r = A_1 \times\cdots\times A_d$\;
$p_r.(\mathrm{bestSplit}, \mathrm{topScore}) = \mathrm{best\_split}(p_r)$\;
\Repeat{termination condition}{
    Let $p \in \mathcal{P}$ be the leaf node with the highest $\mathrm{topScore}$\;
    Apply $p.\mathrm{bestSplit}$\;
    \ForEach{newly created (for regular leaf split) or updated (for small leaf split) leaf node $p'$}{
        $p'.(\mathrm{bestSplit}, \mathrm{topScore}) = \mathrm{best\_split(p')}$\;
    }
}
Return best partitioning $\mathcal{P}^*$ found
\end{algorithm}

\begin{algorithm}[tb]
\caption{$\mathrm{best\_split}$}
\label{alg:bestSplit}
\SetAlgoLined
\LinesNumbered
\KwData{Partition $p$, input and output sample tuples in $p$,
number of row sub-partitions $r$ and column sub-partitions $c$ ($r=c=1$ for regular partitions)}
\KwResult{Split predicate $\mathrm{bestSplit}$ and its score $\mathrm{topScore}$}
Initialize $\mathrm{topScore} = 0$ and $\mathrm{bestSplit} = \mathrm{NULL}$\;
\If{$p$ is a regular partition}{
    \tcp{Find best decision-tree style split}
    \ForEach{regular dimension $A_i$}{
        Let $x_i$ be the split predicate on dimension $A_i$ that has the highest ratio
        $\sigma_i = \Delta\mathrm{Var}(x_i)/\Delta\mathrm{Dup}(x_i)$ among all
        possible splits in dimension $A_i$\;
        \If{$\sigma_i > \mathrm{topScore}$}{
            Set $\mathrm{topScore} = \sigma_i$ and set $\mathrm{bestSplit} = x_i$\;
        }
    }    
}\Else{
    \tcp{Small partition: increment number of row or column sub-partitions using \methodOneB}
    Let $\sigma_r = \Delta\mathrm{Var}(r+1, c)/\Delta\mathrm{Dup}(r+1, c)$\;
    Let $\sigma_c = \Delta\mathrm{Var}(r, c+1)/\Delta\mathrm{Dup}(r, c+1)$\;
    \If{$\sigma_r > \sigma_c$}{
        Set $\mathrm{topScore} = \sigma_r$ and set $\mathrm{bestSplit} = \mathrm{row}$
    }\Else{
        Set $\mathrm{topScore} = \sigma_c$ and set $\mathrm{bestSplit} = \mathrm{column}$
    }
}
Return $(\mathrm{bestSplit}, \mathrm{topScore})$
\end{algorithm}

\begin{figure}[tb]
\centering
	\includegraphics[width=\linewidth]{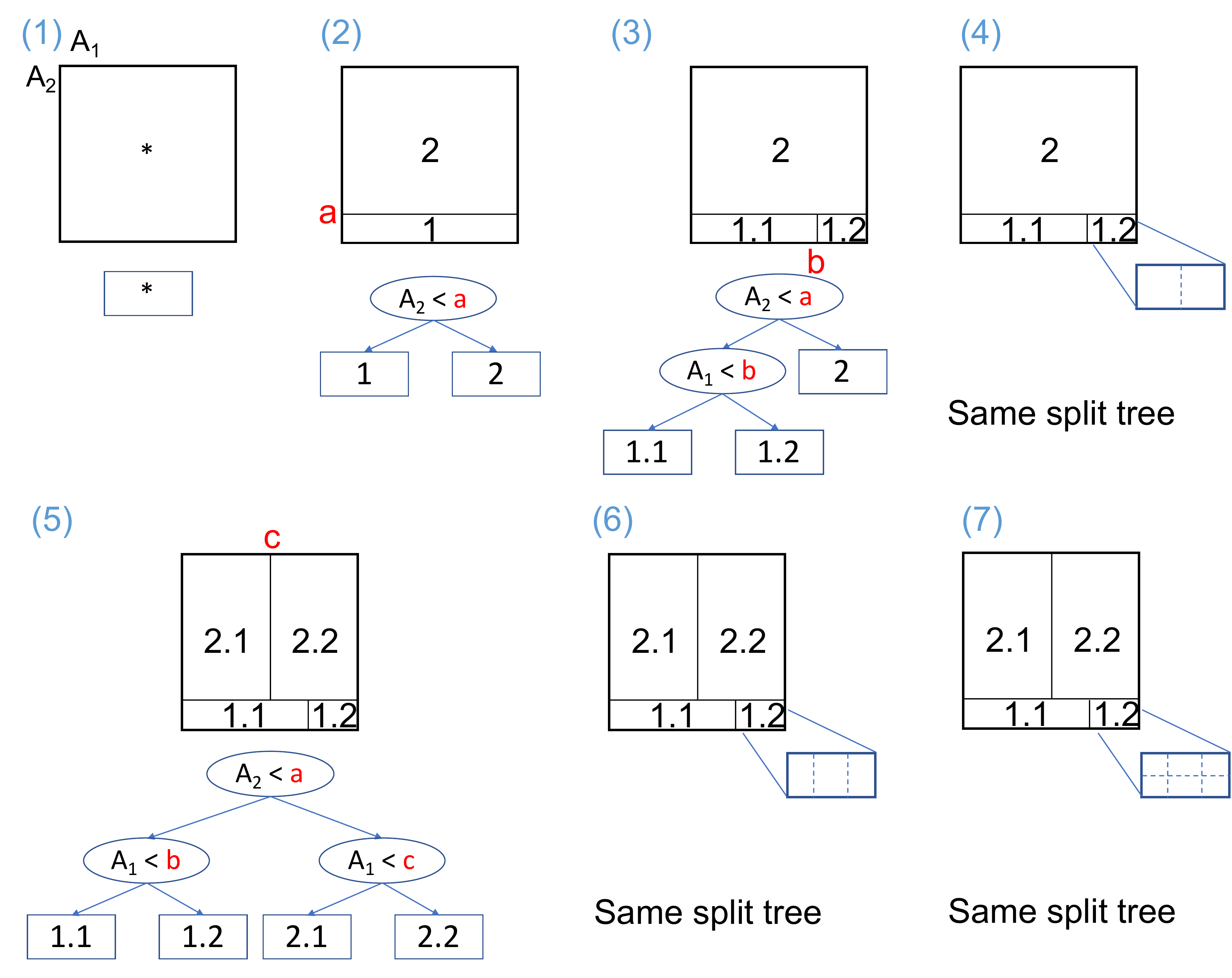}
\caption{
Recursive partitioning for a 2D band-join on attributes $A_1$ and $A_2$. In the split tree, a path
from the root to a leaf defines a rectangular partition in $A_1 \times A_2$ as the conjunction of all predicates along the path.
(By convention the left child is the branch satisfying the split predicate.) In small partitions such as 1.2, \method applies
\methodOneB. Those partitions are terminal leaves in the split tree and only change their ``internal'' partitioning.}
\label{fig:recPartSteps}
\end{figure}

\method (\Cref{alg:adapFull}) is inspired by decision trees~\cite{han2011data} and recursively partitions the
$d$-dimensional space spanned by all join attributes. To adapt this high-level
idea to running-time optimization for band-joins, we (1) identify a new split-scoring measure that determines
the selection of split boundaries, (2) propose a new stopping condition for the splitting process, and (3)
propose an ordering to determine which tree leaf to consider next for splitting. The splitting process is illustrated
 in \Cref{fig:recPartSteps}.

\Cref{alg:adapFull} starts with a single leaf node covering the entire join-attribute space and calls
$\mathrm{best\_split}$ (\Cref{alg:bestSplit}) on this leaf to find the best possible
split and its score. Assuming the leaf is a ``regular'' partition (we discuss small partitions below), $\mathrm{best\_split}$
sorts the input sample on $A_1$ and tries all middle-points between consecutive $A_1$-values as possible split boundaries.
Then it does the same for $A_2$. \emph{The winning split boundary is the one with the highest ratio
between load-variance reduction and input duplication increase} (see details in \cref{sec:algorithmDetails}).

Assume the best split is $A_2 < a$. The first execution of the repeat-loop then applies this split, creating
new leaves ``1'' and ``2'' and finding the best split for each of them. Which leaf should be split next?
\emph{\method manages all split-tree leaves in a priority queue based on their $\mathrm{topScore}$ value.}
Assuming
leaf ``1'' has the higher score, the next repeat-loop iteration will split it, creating new leaves ``1.1'' and ``1.2''.
This process continues until the appropriate termination condition is reached (discussed below).
As the split tree is grown, the algorithm also keeps track of the best partitioning found so far.

\subsection{Algorithm Details}
\label{sec:algorithmDetails}

\introparagraph{Small partitions} When a partition becomes ``small'' relative to band width in
a dimension, then no further \emph{recursive} splitting in that dimension is allowed.
When the partition is small in \emph{all}
dimensions, then it switches into a different partitioning mode inspired by \methodOneB~\cite{Okcan:theta-join}.
This is motivated by the observation that when the length of a partition approaches band width in each
dimension, then all $S$ and $T$-tuples in that partition join with each other. And for Cartesian products,
\methodOneB was shown to be near-optimal.

We define a partition as ``small'' as soon as its size is below twice the band width
in all dimensions. In \Cref{fig:recPartSteps} step (4), leaf ``1.2'' is small and hence when it is picked
in the repeat-loop, applying the best split leaves the split tree unchanged, but instead
increases the number of column partitions $c$ to 2. Afterward, the $\mathrm{topScore}$
value of leaf ``1.2'' may have decreased and leaf ``2'' is split next, using a regular recursive split.
This may be followed by more ``internal'' splits of leaf ``1.2'' in later iterations as shown in \Cref{fig:recPartSteps},
steps (6) and (7).

It is easy to show that having some leaves in ``regular'' and others in ``small'' split mode
does not affect correctness. Intuitively, this is guaranteed because duplication only needs to be considered
inside the region that is further partitioned, i.e., it does not ``bleed'' beyond partition boundaries.

\introparagraph{Split scoring}
In split score $\Delta\mathrm{Var}(x_i)/\Delta\mathrm{Dup}(x_i)$ for a regular dimension (\Cref{alg:bestSplit}),
$\Delta\mathrm{Var}(x_i)$ is defined as follows:
Let $\mathcal{L}$ denote the set of split-tree leaves.
Each (sub-partition in a) leaf corresponds to a region $p$ in $A_1 \times\cdots\times A_d$, for which we
estimate input $I_p$ and output $O_p$ from the random samples drawn by \Cref{alg:adapFull}.
The load induced by $p$ is $l_p = \beta_2 I_p + \beta_3 O_p$. Load variance is computed as follows:
Assign each leaf in $\mathcal{L}$ to a randomly selected worker.
Then per-worker load is a random variable $\mathcal{P}$ (we slightly abuse notation to avoid notational clutter)
whose variance can be shown to be $\Var[\mathcal{P}] = \frac{w-1}{w^2} \sum_{p \in \mathcal{L}} l_p^2$.
We analogously obtain $\Var(\mathcal{P'})$ for a partitioning $\mathcal{P'}$ that results from
splitting some leaf $p' \in\mathcal{L}$ into sub-partitions $p_1$ and $p_2$ using predicate $x_i$.
Then $\Delta\Var(x_i) = \Var(\mathcal{P'}) - \Var(\mathcal{P})$.
$\Var(\mathcal{P'})$ can be computed from $\Var(\mathcal{P})$ in \emph{constant} time by subtracting
$\frac{w-1}{w^2} l_{p'}$ and adding $\frac{w-1}{w^2} (l_{p_1} + l_{p_2})$.

The additional duplication caused by a split is obtained by estimating the number of $T$-tuples within band width
of the new split boundary using the input sample. When multiple split predicates cause no input duplication,
then the best split is the one with the greatest variance reduction among them. The calculation of
load variance and input duplication for ``small'' leaves is analogous.

The split score reflects our goal of reducing max worker load with minimal input duplication.
For the former, load-variance reduction could be replaced by other measures, but precise estimation
of input and output on the most loaded worker is difficult due to dynamic load balancing applied by schedulers
at runtime. We therefore selected load variance as a scheduler-independent proxy.

\introparagraph{Termination condition and winning partitioning}
We propose a \emph{theoretical} and an \emph{applied} termination condition for the repeat-loop
in \Cref{alg:adapFull}. For the theoretical approach,
the winning partitioning is the one with the \emph{lowest overhead over the lower bound
in terms of both max worker load and input duplication}, i.e.,
the one with the minimal value of $\max\big\{\frac{I-(|S|+|T|)}{|S|+|T|}; \frac{L_m-L_0}{L_0}\big\}$.
It is easy to show that each iteration of the repeat-loop monotonically increases input $I$,
because each new split boundary (regular leaf) and more fine-grained sub-partitioning (small leaf)
can only increase the number of input duplicates. At the same time, the loop iteration may or may not decrease
$\frac{L_m-L_0}{L_0}$. Hence repeat-loop iterations can be terminated as soon as
$\frac{I-(|S|+|T|)}{|S|+|T|}$ exceeds the smallest value of $\frac{L_m-L_0}{L_0}$ encountered so far.

The theoretical approach only needs input and output samples, as well as an estimate of the relative impact of
an input tuple versus an output tuple on join computation time. Input sampling is straightforward; for output
sampling we use the method from~\cite{vitorovic2016load}. If the output is large, it efficiently produces a large sample.
If the output is small, then the output has negligible impact on join computation cost. The experiments show that
we get good results when limiting sample size based on memory size and sampling
cost to at most 5\% of join time.

For estimating load impact, we run band-joins with different
input and output sizes $I$ and $O$ on an individual worker and use linear regression to determine
$\beta_2$ and $\beta_3$ in load function $\beta_2 I + \beta_3 O$. In our Amazon cloud cluster,
$\beta_2 / \beta_3 \approx 4$. Note that $\beta_2$ and $\beta_3$ tend to increase with input
and output. For the lower bound, we use the smallest values, i.e., those obtained for scenarios
where a node receives about $1/w$ of the input (recall that $w$ is the number of workers).
This establishes a lower value for the lower bound, i.e., it is more challenging for \method to be close to it.

For the applied approach, we use the cost model as discussed in the end of \Cref{sec:problemDef}.
The winning partitioning is the one with the lowest running time predicted by the cost model.
Repeat-loop iterations terminate when estimated join time bottoms out. We detect this based on
a window of the join times over the last $w$ repeat-loop iterations: loop execution terminates
when improvement is below 1\% (or join time even increased) over those last $w$ iterations.
(We chose $w$ as window size because it would take at least one extra split per worker to break
up each worker's load. This dovetails with the prioritization of leaves: The most promising leaves in
terms of splitting up load with low input duplication overhead are greedily selected.)

\introparagraph{Extension: symmetric partitioning}
Like classic grid partitioning, \method as discussed so far treats inputs $S$ and $T$ differently:
at an inner node in the split tree, $S$ is partitioned (without duplication), while $T$-tuples near
the split boundary are duplicated. For regions where $S$ is sparse and $T$ is dense,
we want to reverse these roles. Consider \Cref{fig:Duplication-LoadBalance-Tradeoff_e}, where
$y_1$ and $y_2$ enabled a zero-duplication partitioning with perfect load balance. What if the input
distribution was reversed in another region of the join-attribute space, e.g., $S'=\{21,25,26,30\}$ and
$T'=\{21,22,23,25,26,28,29,30\}$? Then no split in range 21 to 30 could avoid duplication of $T'$-tuples
because for band width 1 at least one of the $T'$-values would be within 1 of the split point.
In that scenario we want to reverse the roles of $S'$
and $T'$, i.e., perform the partitioning on $T'$ and the partition/duplication on $S'$.

For grid partitioning, it is not clear how to reverse the roles of $S$ and $T$ in some of the grid cells.
For \method, this turns out to be easy. When exploring split options in a regular leaf
(\methodOneB in small partitions already treats both inputs symmetrically),
\Cref{alg:bestSplit} computes the duplication for both cases: partition $S$ and partition/duplicate $T$ as well
as the other way round. We call the former a \emph{$T$-split} and the latter an \emph{$S$-split}.
The split type information is added to the corresponding node in the split tree.

\Cref{alg:tupleAssignment} is used
to determine which tuples to assign to the sub-trees. (Only the version for $T$-tuples is shown; the
one for $S$-tuples is analogous.) It is easy to show that for each result $(s, t) \in S \bowtie_B T$,
exactly one leaf in the split tree receives both $s$ and $t$.

\begin{algorithm}[tb]
\caption{assign\_input}
\label{alg:tupleAssignment}
\SetAlgoLined
\LinesNumbered
\KwData{Input tuple $t \in T$; node $p$ in split tree}
\KwResult{Set of leaves to which $t$ is copied}
\If{$p$ is a leaf}{
    Return (\methodOneB($t$, $p$)) \label{alg:oneBstep}
}\Else{
    \If{$p$ is a $T$-split node}{
        \For{each child partition $p'$ of $p$ that intersects with the $\varepsilon$-range of $t$}{
            assign\_input($t$, $p'$)
        }
    }\Else{
        Let $p'$ be the child partition of $p$ that contains $t$\;
        assign\_input(t, $p'$)
    }
}
\end{algorithm}

\subsection{Algorithm Analysis}
\label{sec:optCost}

\method has low complexity, resulting in low optimization time.
Let $\lambda$ denote the number of repeat-loop executions in \Cref{alg:adapFull}.
Each iteration increases the number of leaves in the split tree by at most one.
The algorithm manages all leaves in a priority queue based on the score returned by \Cref{alg:bestSplit}.
With a priority queue where inserts have constant cost and removal of the top element takes time
logarithmic in queue size, an iteration of the repeat-loop takes $\bigO(\log \lambda)$
to remove the top-scoring leaf $p$. If $p$ is regular, then it takes $\bigO(1)$ to create sub-partitions
$p_1$ and $p_2$ and to distribute all input and output samples in $p$ over them.
(Sample size is bounded by constant $k$, a fraction of machine-memory size.)
Checking if $p_1$ and $p_2$ are small partitions takes time $\bigO(d)$.
Then best\_split is executed, which for a regular leaf requires sorting of the input sample on each dimension
and trying all possible split points, each a middle point between two consecutive sample tuples in that dimension.
Since sample size is upper-bounded by a constant, the cost is $\bigO(d)$. For a small leaf, the cost is $\bigO(1)$.
Finally, inserting $p_1$ and $p_2$ into the priority queue takes $\bigO(1)$.
In total, splitting a regular or small leaf has complexity $\bigO(d)$ and $\bigO(1)$, respectively.

After $\lambda$ executions of the repeat-loop in \Cref{alg:adapFull}, the next iteration
has complexity $\bigO(\log \lambda + d)$. Hence the total cost of $\lambda$ iterations is
$\bigO(\lambda \log \lambda + \lambda d)$. In our experience, the algorithm will terminate after a number
of iterations bounded by a small multiple of the number of worker machines.
To see why, note that each iteration breaks up a large partition $p$
and replaces it by two (for regular $p$) or more (for small $p$) sub-partitions. 
The split-scoring metric favors breaking up heavy partitions, therefore load can be balanced across workers
fairly evenly as soon as the total number of partitions reaches a small multiple of the number of workers.
This in turn implies for \Cref{alg:adapFull}, given samples of fixed size, a total complexity of
$\bigO(w \log w + w d)$.

\section{Analytical Insights}
\label{sec:previousPlus}

We present two surprising results about the ability of grid partitioning to address load imbalances.
For join-matrix covering approaches like \methodVit that depend on the notion of a total ordering of the join-attribute
space, we explore how to enumerate the multi-dimensional space.

\subsection{Properties of Grid Partitioning}
\label{sec:gridProperties}

Without loss of generality, let $S$ be the input that is partitioned and $T$ be the input that is partitioned/duplicated.

\introparagraph{Input duplication} With grid-size in each
dimension set to the corresponding band width, the $\varepsilon$-range of a $T$-tuple intersects with up to
3 grid cells per dimension, for a total replication rate of $\bigO(3^d)$ in $d$ dimensions. Can this exponential
dependency on the dimensionality be addressed through coarser partitioning? Unfortunately, for any non-trivial
partitioning, asymptotically the replication rate is still at least $\bigO(2^d)$. To construct the worst case,
an adversary places all input tuples near the corner of a centrally located grid cell.

\introparagraph{Max worker load}
We now show an even stronger negative result,
indicating that grid-partitioning is inherently limited in its ability to reduce max worker load, no matter the number
of workers or the grid size.
Consider a partitioning where one of the grid cells contains half of $S$ and half of $T$.
Does there exist a more fine-grained grid partitioning where none of the grid
cells receives more than 10\% of $S$ and 10\% of $T$?
One may be tempted to answer in the affirmative: just keep decreasing grid size in all dimensions until
the target is reached. Unfortunately, this does not hold as we show next.

\begin{lemma}\label{lem:gridBalanceBad}
If there exists an $\varepsilon$-range in the join-attribute space with $n$ tuples from $T$, then
grid partitioning will create a partition with at least $n$ $T$-tuples, \emph{no matter the grid size}.
\end{lemma}

\begin{proof}
For $d=1$ consider an interval of size $\varepsilon_1$ that contains $n$ $T$-tuples.
If no split point partitions the interval, then the partition containing it has all $n$ $T$-tuples.
Otherwise, i.e., if at least one split point partitions the interval, pick one of the split points inside the interval,
say $X$, and consider the $T$-tuples copied to the grid cells adjacent to $X$. Since all $T$-tuples in an
interval of size $\varepsilon_1$ are within $\varepsilon_1$ of $X$, both grid cells receive all $n$ $T$-tuples
from the interval. It is straightforward to generalize this analysis to $d>1$.
\end{proof}

In short, even though a more fine-grained partitioning can split $S$ into ever smaller pieces, the same is not
possible for $T$ because of the duplication needed to ensure correctness for a band-join.
For \emph{skewed} input, $n$ in \Cref{lem:gridBalanceBad} can be $\bigO(|T|)$.
Interestingly, we now show that grid partitioning can
behave well in terms of load distribution for skewed input, as long as the input is ``sufficiently'' large.

\begin{lemma}\label{lem:gridBalanceGood}
Let $c_0>0$ be a constant such that $|S \bowtie_B T| \le c_0(|S|+|T|)$.
Let $R$ ($R'$) denote the region of size $\varepsilon_1 \times \varepsilon_2 \times\cdots\times \varepsilon_d$
in the join attribute space containing the most tuples from $S$ ($T$); and let $x$ ($x'$) and $y$ ($y'$)
denote the fraction of tuples from $S$ and $T$, respectively, it contains.
If there exist constants $0 < c_1 \le c_2$ such that $c_1 \le x/y \le c_2$ and $c_1 \le x'/y' \le c_2$,
then no region of size $\varepsilon_1 \times \varepsilon_2 \times\cdots\times \varepsilon_d$ contains
more than $\bigO(\sqrt{1/|S| + 1/|T|})$ input tuples.
\end{lemma}

\begin{proof}
By definition, all $S$ and $T$ tuples in region $R$ join with each other. Together
with $|S \bowtie_B T| \le c_0(|S|+|T|)$ this implies
\begin{equation}
x|S| \cdot y|T| \le c_0 (|S|+|T|)\quad\Rightarrow\quad xy \le c_0(1/|S| + 1/|T|) \label{eqn:inputFraction}
\end{equation}
From $x/y \le c_2$ follows $x^2 \le c_2 x y$,
then $x^2\le c_2 c_0 (1/|S| + 1/|T|)$ (from \Cref{eqn:inputFraction}) and thus
$x \le \sqrt{c_0 c_2}\sqrt{1/|S| + 1/|T|} = \bigO(\sqrt{1/|S| + 1/|T|})$.
We show analogously for region $R'$ that
$y' \le \sqrt{c_0/c_1}\sqrt{1/|S| + 1/|T|} = \bigO(\sqrt{1/|S| + 1/|T|})$.

Since $R$ is the region of size $\varepsilon_1 \times \varepsilon_2 \times\cdots\times \varepsilon_d$ with most $S$-tuples and $R'$ the region with most $T$-tuples, no region of size $\varepsilon_1 \times \varepsilon_2 \times\cdots\times \varepsilon_d$ can contain more than $x$ fraction of $S$-tuples and $y'$ fraction of $T$-tuples.
\end{proof}

\Cref{lem:gridBalanceGood} is surprising. It states that for larger inputs $S$ and $T$, the \emph{fraction}
of $S$ and $T$ in any partition of \methodGrid (recall that its partitions are of size
$\varepsilon_1 \times \varepsilon_2 \times\cdots\times \varepsilon_d$) is upper-bounded by
a function that decreases proportionally with $\sqrt{|S|}$ and $\sqrt{|T|}$.
For instance,
when $S$ and $T$ double in size, then the upper bound on the
input fraction in any partition decreases by a factor of $\sqrt{2} \approx 1.4$.

Clearly, this does not hold for all band joins.
The proof of \Cref{lem:gridBalanceGood} required (1) that the region with most $S$-tuples contain
a sufficiently large fraction of $T$, and vice versa; and (2) that output size is bounded by $c_0$
times input size. The former is satisfied when $S$ and $T$ have a similar distribution in
join-attribute space, e.g., a self-join. For the latter, we are aware of two scenarios. First,
the user may actively try to avoid huge outputs by setting a smaller band width. Second,
for any output-cost-dominated theta-join, \methodOneB was shown to be
near-optimal~\cite{Okcan:theta-join}. Hence specialized band-join solutions such as
\method, \methodGrid, and \methodVit would only
be considered when output is ``sufficiently'' small.

\subsection{Ranges in Multidimensional Space}
\label{sec:csio-sort}

\begin{figure}[t]
    \begin{subfigure}[t]{0.48\linewidth}
	\centering
	\includegraphics[width=\linewidth]{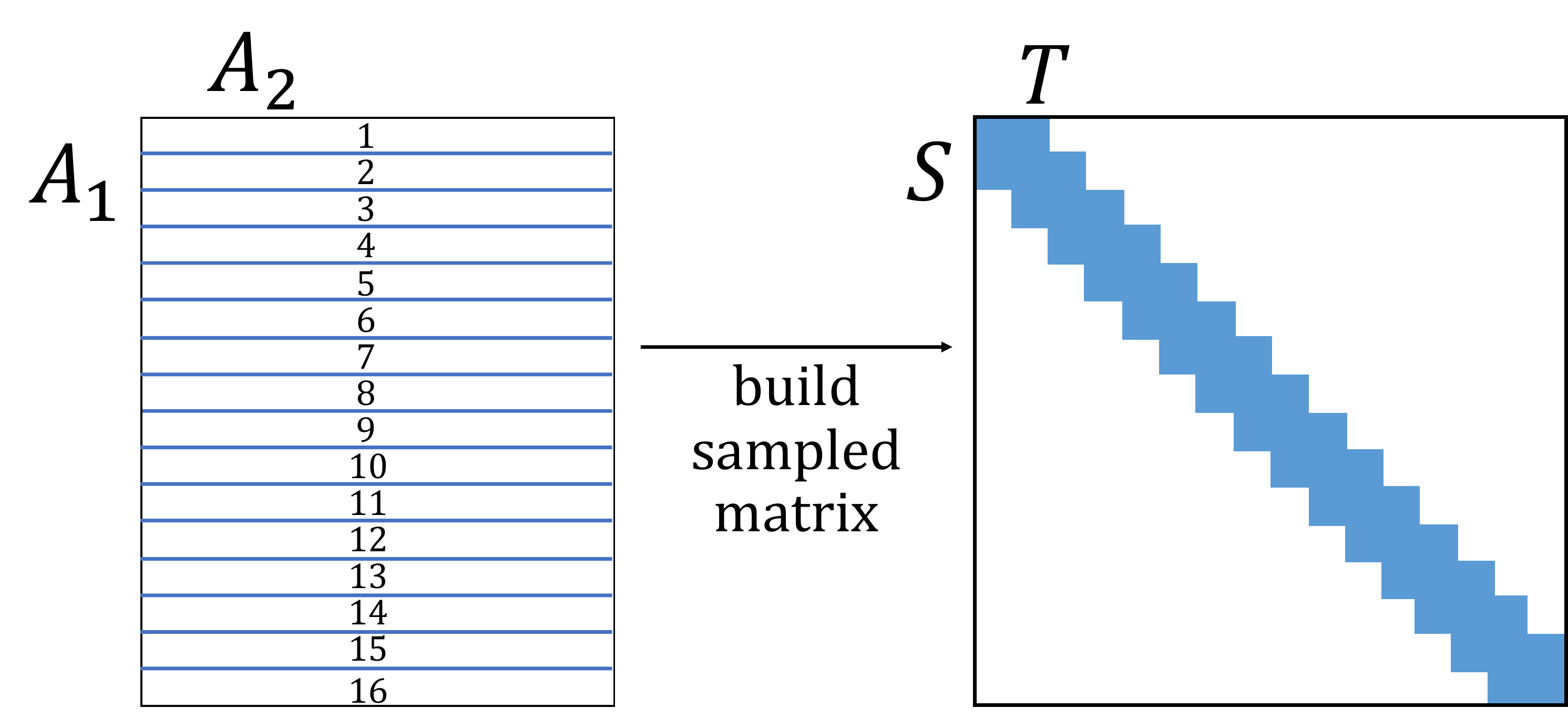}
	\caption{Row-major order}
	\label{fig:cs_io-alphabetical}
	\end{subfigure}
	\hspace{1mm}
	\begin{subfigure}[t]{0.48\linewidth}
	\centering
	\includegraphics[width=\linewidth]{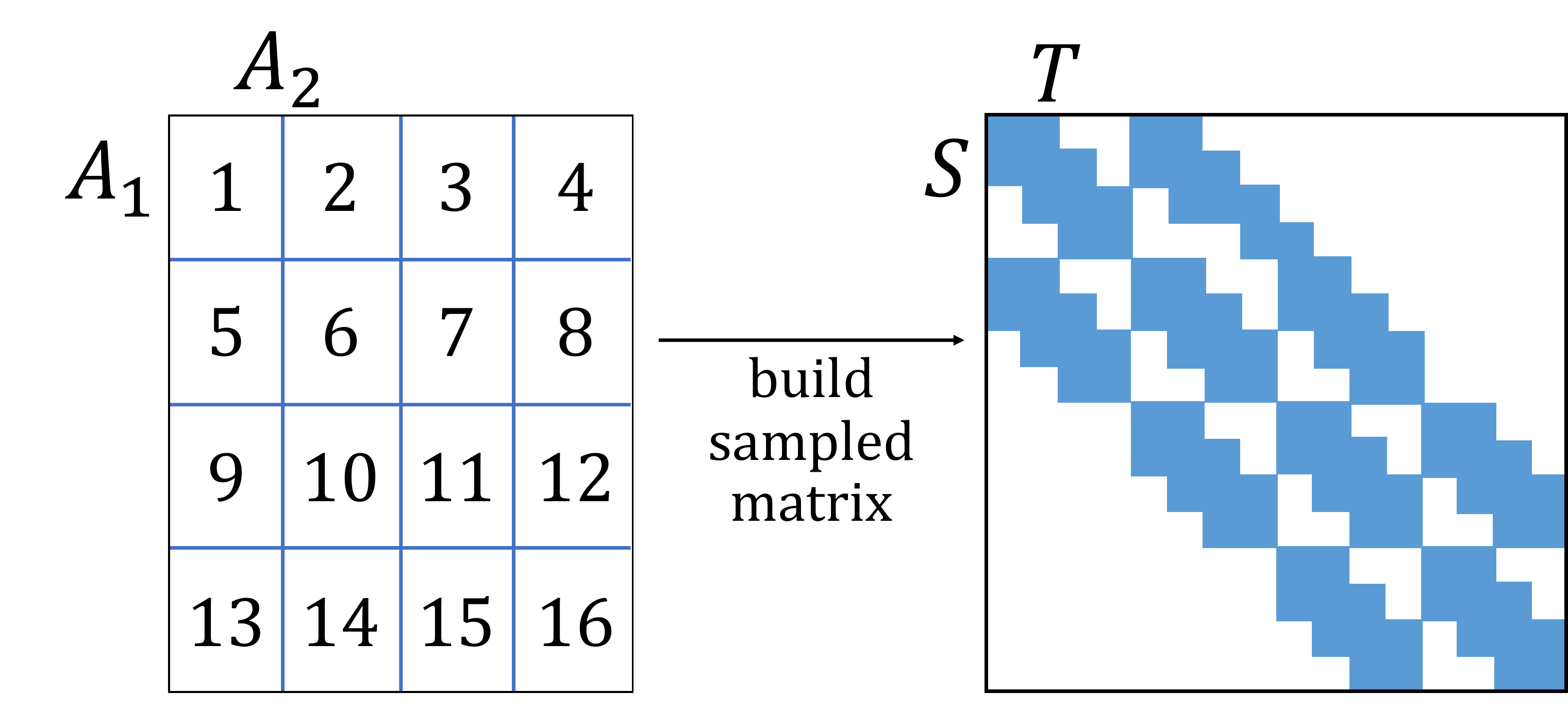}
	\caption{Block-style order}
	\label{fig:cs_io-block}
	\end{subfigure}
\caption{
Impact of enumeration order of a multidimensional space, here for a band-join between $S$ and $T$
on attributes $A_1$ and $A_2$, on the location of cells in the join matrix that may produce output.
Here band width is smaller than the height and width of each partition, resulting in a significantly sparser
join matrix for row-major order.}
\label{fig:cs_io-sampled-matrix}
\end{figure}

\methodVit starts with a range partitioning of the join-attribute space based on approximate quantiles.
For quantiles to be well-defined, a total order must be established on the multi-dimensional space---this
was left unspecified for \methodVit.
The ordering can significantly impact performance as illustrated in \Cref{fig:cs_io-sampled-matrix}
for a 2D band-join. In row-major order, ranges correspond to long horizontal stripes.
Alternatively, each range may correspond to a ``block,'' creating more square-shaped regions.
This choice affects the candidate regions in the join matrix that need to be covered.
Assume each horizontal stripe in
\Cref{fig:cs_io-alphabetical} is at least $\varepsilon_1$ high. Then an $S$-tuple in stripe $i$ can only
join with $T$-tuples in the three stripes $i-1$, $i$, and $i+1$.
This creates the compact candidate-region diagonal of width 3 in the join matrix.
For the block partitioning in \Cref{fig:cs_io-block}, $S$-tuples in a block may join with $T$-tuples in up
to nine neighboring blocks. This creates the wider band in the join matrix. For \methodVit, a wider
diagonal and denser matrix cause additional input duplication.

In general, the number of candidate cells in the join matrix is minimized by
row-major ordering, \emph{if the distance between
the hyperplanes in the most significant dimension is greater than or equal to the band width
in that dimension}. This was the case in our experiments and therefore row-major ordering was selected for \methodVit.

\section{Experiments}
\label{sec:exp}

We compare the total running time (optimization plus join time) of \method to the state of the art
(\methodGrid, \methodVit, \methodOneB). Note that \methodGrid is not defined
for band width zero. Reported times are measured on a real cloud, unless stated otherwise.
In large tables, we mark cells with the main results in blue; red color highlights a weak spot,
e.g., excessive optimization time or input duplication.

\subsection{Experimental Setup}
\label{sec:expSetup}

\introparagraph{Environments} Both MapReduce \cite{Dean:2008:mapreduce}
and Spark \cite{zaharia2010spark}
are well-suited for band-join implementation. Spark's ability to keep large data in memory
makes little difference for the map-shuffle-reduce pipeline of a band-join, therefore we use
MapReduce, where it is easier to control low-level behavior such as custom data
partitioning.  All experiments were conducted on Amazon's Elastic MapReduce (EMR) cloud,
using 30 \texttt{m3.xlarge} machines (15GB RAM, 40GB disk, high network performance)
by default. All clusters run Hadoop 2.8.4 with the YARN scheduler in default configuration. 
\later{All our code will be made available on Github to allow reproducible research.}

\begin{table}[tb]
\centering
\caption{Band-join characteristics used in the experiments. Input and output size are reported in [million tuples].}
\label{tab:query}
\scalebox{0.9}
{
	\begin{tabular}{|l|r|r|r|r|}
	\hline
	Data Set                         & $d$ & Band width & Input Size & Output Size  \\
	\hline
	\texttt{pareto-1.5}              & 1 & $0$              & 400  & 2430  \\ 
	\texttt{pareto-1.5}              & 1 & $10^{-5}$        & 400  & 4580  \\
	\texttt{pareto-1.5}              & 1 & $2\cdot 10^{-5}$ & 400  & 9120  \\
	\texttt{pareto-1.5}              & 1 & $3\cdot 10^{-5}$ & 400  & 11280 \\
	\texttt{pareto-1.5}              & 3 & $(0,0,0)$        & 400  & 0     \\
	\texttt{pareto-1.5}              & 3 & $(2,2,2)$        & 400  & 1120  \\
	\texttt{pareto-1.5}              & 3 & $(4,4,4)$        & 400  & 8740  \\
	\hline	
	\texttt{pareto-0.5}              & 3 & $(2,2,2)$        & 400  & 12    \\
	\texttt{pareto-1.0}              & 3 & $(2,2,2)$        & 400  & 420   \\
	\texttt{pareto-2.0}              & 3 & $(2,2,2)$        & 400  & 3200  \\
	\hline
	\texttt{pareto-1.5}              & 8 & $(20,\ldots,20)$     & 100  & 9  \\
	\texttt{pareto-1.5}              & 8 & $(20,\ldots,20)$     & 200  & 57  \\
	\texttt{pareto-1.5}              & 8 & $(20,\ldots,20)$     & 400  & 219  \\
	\texttt{pareto-1.5}              & 8 & $(20,\ldots,20)$     & 800  & 857  \\
	\hline
	\texttt{rv-pareto-1.5}      & 1 & $2$        & 400    & 0  \\
	\texttt{rv-pareto-1.5}      & 1 & $1000$     & 400 	 & 0  \\
	\texttt{rv-pareto-1.5}      & 3 & $(1000,1000,1000)$     & 400 & 0  \\
	\texttt{rv-pareto-1.5}      & 3 & $(2000,2000,2000)$     & 400 & 0  \\
	\hline
	\texttt{ebird} and \texttt{cloud} & 3 & ($0,0,0$) & 890     & 0     \\
	\texttt{ebird} and \texttt{cloud} & 3 & ($1,1,1$) & 890        & 320  \\
	\texttt{ebird} and \texttt{cloud} & 3 & ($2,2,2$) & 890        & 2134 \\
	\texttt{ebird} and \texttt{cloud} & 3 & ($4,4,4$) & 890        & 16998 \\
    \hline
	\end{tabular}
}
\end{table}

\introparagraph{Data} For synthetic data, we use a Pareto distribution where join-attribute
value $x$ is drawn from domain $[1.0, \infty)$ of real numbers and follows PDF
$z / x^{z+1}$ (greater $z$ creates more skew). This models the famous power-law
distribution observed in many real-world contexts, including the 80-20 rule for
$z=\log_4 5 \approx 1.16$. We explore $z$ in the range $[0.5,2.0]$, which covers
power-law distributions observed in real data.
\texttt{pareto-$z$} denotes a pair of tables, each with 200 million tuples, with Pareto-distributed
join attributes for skew $z$. High-frequency values in $S$ are also high-frequency values in $T$.
\texttt{rv-pareto-$z$} is the same as \texttt{pareto-$z$}, but high-frequency values in $S$
have low frequency in $T$, and vice versa. Specifically, $T$ follows a Pareto distribution from
$10^6$ down to $-\infty$.  
($T$ is skewed toward larger values. We generate $T$ by drawing numbers from 
$[1.0, \infty)$ following Pareto distribution and then converting each number $y$ to $10^6 - y$.)

\texttt{cloud} is a real dataset containing $382$ million cloud reports~\cite{hahn1999extended},
each reporting time, latitude, longitude, and 25 weather attributes.
\texttt{ebird} is a real dataset containing $508$ million bird sightings, each with attributes
describing time, latitude, longitude, species observed, and 1655 features of the observation
site~\cite{munson2014ebird}.

For each input, we explore different band widths as summarized in \Cref{tab:query}.
\emph{Output sizes below 0.5 million are reported as 0.} For the real data, the three join
attributes are time ([days] since January 1st, 1970), latitude ([degrees] between $-90$ and $90$)
and longitude ([degrees] between $-180$ and $180$).

\introparagraph{Local join algorithm} After partitions are assigned to workers, each worker
needs to locally perform a band-join on its partition(s). Many algorithms could
be used, ranging from nested-loop to adaptations
of IEJoin's sorted arrays and bit-arrays. Since we focus on the \emph{partitioning} aspect, the
choice of local implementation is orthogonal, because it only affects the relative importance
of optimizing for input duplication vs optimizing for max worker load.
We observed that \method wins \emph{no matter what this ratio is set to}, therefore in our
experiments we selected a fairly standard local band-join algorithm based on index-nested-loops.
(Let $S_p$ and $T_p$ be the input in partition $p$.) (1) range-partition $T_p$ on $A_1$ into
ranges of size $\varepsilon_1$. (Here $A_1$ is the most selective dimension.) (2) For each
$s \in S_p$, use binary search to find the $T$-range $i$ containing $s$. Then check
band condition on $(s,t)$ for all $t$ in ranges $(i-1)$, $i$, and $i+1$.

Since \methodGrid partitions are of size $\varepsilon_1$ in dimension $A_1$, we slightly modify the
above algorithm and sort both $S_p$ and $T_p$ on $A_1$. The binary search for $s \in S_p$ then
searches for $s.A_1 - \varepsilon_1$ (the smallest tuple $t \in T_p$ it could join with) and scans
the sorted $T$-array from there until $s.A_1+\varepsilon_1$.

\introparagraph{Statistics and running-time model} We sample 100,000 input records and set
output-sample size so that total time for statistics gathering does not exceed $5\%$ of the fastest
time (optimization plus join time) observed for any method. 
For output sampling we use the method introduced for \methodVit~\cite{vitorovic2016load}.
For the cost model (see \Cref{sec:problemDef}), we determine the model coefficients ($\beta$-values)
as discussed in \cite{Li2018DAPD} from a benchmark of 100 queries. The benchmark is run offline
once to profile the performance aspects of a cluster.

\hide{
\introparagraph{Result reporting} We report measured \parTime{} as well as measured join time on the EMR cluster. 
Optimization (i.e., finding a partitioning) is performed on a single machine, while the join is executed on all machines. 
However, measured join time heavily depends on tuning parameters of cluster and computation environment, as well as features of the underlying hardware and virtualization environment. 
In order to show more general results, 
we report more insightful numbers that are independent of these tuning and setup aspects for most experiments: (1) The total amount of data shuffled $I$ 
(which is the input size plus all duplicates created due to the partitioning). 
Clearly, the more data is shuffled, the longer it takes to transfer the data from map to reduce phase. 
(2) The total amount of input ($I_m$) and output ($O_m$) assigned to the worker machine with the greatest load. 
It determines the longest reduce time and hence the end of the join computation. Recall that we define load as the weighted sum of input and output as for \methodVit~\cite{vitorovic2016load}. 
Our offline measurements on the EMR cluster yielded a 4-to-1 weight ratio between input and output.

When relative time over \method (or \methodBase) is reported for the other methods, this number is calculated by dividing the time of the other method by that of \method (or \methodBase).
}

\begin{table*}[t]
\setlength{\tabcolsep}{1mm}
\caption{Impact of band width: 
\methodS wins in all cases, and the winning margin gets bigger for band joins with more dimensions. (Blue color
highlights the main results; red color highlights a weak spot.)}
\label{tab:time-single-attr}
\centering
\begin{subtable}{\linewidth}
	\centering
	\caption{\texttt{Pareto-$1.5$}, $d=1$, varying band width.}
	\label{tab:time-single-attr}
	\scalebox{0.8}
	{
	    \begin{tabular}{|r||r|r|r|r||r|r|r||rrr|rrr|rrr|rrr|}
		\hline
	    Band width 
		    & \multicolumn{4}{c||}{Runtime (\parTime+join time) in [sec]}
	        & \multicolumn{3}{c||}{Relative time over \methodS} 
	        & \multicolumn{12}{c|}{I/O sizes in [millions]: $I$, $I_m$, $O_m$} \\
		\cline{2-20}
		    & \methodS & \methodVit & \methodOneB & \methodGrid  
			& \methodVit & \methodOneB & \methodGrid 
			& \multicolumn{3}{c|}{\methodS} & \multicolumn{3}{c|}{\methodVit} 
			& \multicolumn{3}{c|}{\methodOneB} & \multicolumn{3}{c|}{\methodGrid}   \\
		\hline
		0                
		& 351(3+348) 
		& 512(29+483)   
		& 762  
		& --- 
		& \colorbox{Blue1}{1.46} 
		& \colorbox{Blue1}{2.17}
		& \colorbox{Blue1}{N/A}
		& 400 
		& 14 
		& 83  
		& 496 
		& 13 
		& 131  
		& 2200 
		& 73 
		& 81   
		& --- & --- & ---
		\\
		$10^{-5}$        & 539(7+532) & 684(29+655)   & 1004 & 540 
			& \colorbox{Blue1}{1.27} & \colorbox{Blue1}{1.86} & \colorbox{Blue1}{1.00} 
			& 400 & 12 & 158  & 475 & 8  & 266  & 2200 & 73 & 153  & 800 & 27 & 153  
		\\ 
		$2\!\cdot\! 10^{-5}$ & 813(3+810) & 992(30+962)   & 1316 & 834 
			& \colorbox{Blue1}{1.22} & \colorbox{Blue1}{1.62} & \colorbox{Blue1}{1.03} 
			& 401 & 13 & 305  & 488 & 10 & 388  & 2200 & 73 & 304  & 800 & 27 & 304  
		\\
		$3\!\cdot\! 10^{-5}$ & 878(3+875) & 1170(30+1140) & 1520 & 956 
			& \colorbox{Blue1}{1.33} & \colorbox{Blue1}{1.73} & \colorbox{Blue1}{1.09} 
			& 401 & 12 & 384  & 479 & 10 & 503  & 2200 & 73 & 376  & 800 & 27 & 376  
		\\ 
		\hline 
		\end{tabular}
	}
\end{subtable}
\vspace{5px}
\newline
\begin{subtable}{\linewidth}
	\centering
	\caption{\texttt{Pareto-$1.5$}, $d=3$, varying band width.}
	\label{tab:time-pareto-3d}
	\scalebox{0.8}
	{
	    \begin{tabular}{|r||r|r|r|r||r|r|r||rrr|rrr|rrr|rrr|}
		\hline
		\multirow{2}{*}{Band width}
		    & \multicolumn{4}{c||}{Runtime (\parTime+join time) in [sec]}
	        & \multicolumn{3}{c||}{Relative time over \methodS}
	        & \multicolumn{12}{c|}{I/O sizes in [millions]: $I$, $I_m$, $O_m$} \\
		\cline{2-20}
		    & \methodS & \methodVit & \methodOneB & \methodGrid
			& \methodVit & \methodOneB & \methodGrid
			& \multicolumn{3}{c|}{\methodS} & \multicolumn{3}{c|}{\methodVit}
			& \multicolumn{3}{c|}{\methodOneB} & \multicolumn{3}{c|}{\methodGrid}   \\
		\hline
		($0,0,0$) & 230(1+229) & 366(46+320)     & 792  & ---  & \colorbox{Blue1}{1.59} & \colorbox{Blue1}{3.44} & \colorbox{Blue1}{N/A} & 401 & 14 & 0   & 497 & 17 & 0   & 2200 & 73 & 0   & ---  & --- & --- \\
		($2,2,2$) & 344(2+342) & 1339(\colorbox{Red1}{694}+645)   & 1149 & 1412 & \colorbox{Blue1}{3.89} & \colorbox{Blue1}{3.34} & \colorbox{Blue1}{4.10}  & 404 & 15 & 29  & 652 & 19 & 69  & 2200 & 73 & 37  & \colorbox{Red1}{5541} & 185 & 37 \\
		($4,4,4$) & 860(2+858) & 2557(\colorbox{Red1}{1345}+1212) & 1772 & 1816 & \colorbox{Blue1}{2.97} & \colorbox{Blue1}{2.06} & \colorbox{Blue1}{2.11}  & 413 & 14 & 290 & 838 & 31 & 321 & 2200 & 73 & 291 & \colorbox{Red1}{5485} & 183 & 291 \\
		\hline
		\end{tabular}
	}
\end{subtable}
\vspace{5px}
\newline
\begin{subtable}{\linewidth}
	\centering
	\caption{Join of \texttt{ebird} with \texttt{cloud}, $d=3$, varying band width.}
	\label{tab:time-real}
	\scalebox{0.8}
	{
	    \begin{tabular}{|r||r|r|r|r||r|r|r||rrr|rrr|rrr|rrr|}
		\hline
		\multirow{2}{*}{Band width}
		    & \multicolumn{4}{c||}{Runtime (\parTime+join time) in [sec]}
	        & \multicolumn{3}{c||}{Relative time over \methodS}
	        & \multicolumn{12}{c|}{I/O sizes in [millions]: $I$, $I_m$, $O_m$} \\
		\cline{2-20}
		    & \methodS & \methodVit & \methodOneB & \methodGrid
			& \methodVit & \methodOneB & \methodGrid
			& \multicolumn{3}{c|}{\methodS} & \multicolumn{3}{c|}{\methodVit}
			& \multicolumn{3}{c|}{\methodOneB} & \multicolumn{3}{c|}{\methodGrid}   \\
		\hline
		($0,0,0$) & 248(3+245) & 346(38+308)   & 1418 & ---   
			&  \colorbox{Blue1}{1.40} & \colorbox{Blue1}{5.72} & \colorbox{Blue1}{N/A}
			&  890 & 30 & 0  &  951 & 32 & 0  &  4832 & 161 & 0  & --- & --- & --- \\
		($1,1,1$) & 332(3+329) & 1945(\colorbox{Red1}{968}+977) & 1532 & 1419  
			&  \colorbox{Blue1}{5.86} & \colorbox{Blue1}{4.61} & \colorbox{Blue1}{4.27}  
			&  895 & 35 & 5  & 1490 & 95 & 9  & 4832 & 161 & 11 & \colorbox{Red1}{10891} & 361 & 11 \\
		($2,2,2$) & 423(3+420) & 2615(\colorbox{Red1}{1553}+1062) & 1573 & 1377  
			&  \colorbox{Blue1}{6.18} & \colorbox{Blue1}{3.72} & \colorbox{Blue1}{3.26}  
			&  899 & 32 & 66 & 1830 & 107 & 74 & 4832 & 161 & 67 & \colorbox{Red1}{10783} & 361 & 74 \\
		\hline
		\end{tabular}
	}
\end{subtable}
\end{table*}

\subsection{Impact of Band Width}
\label{sec:expIORatio}

We explore the impact of band width, which affects output.
For comparison with grid partitioning, we turn \method's
symmetric partitioning off, i.e., $T$ is always the partitioned/duplicated relation.
Since the grid approaches do not apply symmetric partitioning by design,
all advantages of \method in the experiments are due to the better partition boundaries,
not due to symmetric partitioning.
(The impact of symmetric partitioning is explored separately later.) To avoid confusion, we
refer to \method without symmetric partitioning as \methodS.

\subsubsection{Single Join Attribute}
\label{sec:expSingleAttr}

The left block in \Cref{tab:time-single-attr} reports running times.
\emph{\methodS wins in all cases, by up to a factor of 2}, but the other methods are
competitive because the join is in 1D and skew is moderate.
For \methodVit we tried different parameter settings
that control optimization time versus partitioning quality, reporting the best found.
The right block in \Cref{tab:time-single-attr} reports input plus duplicates ($I$), and the input ($I_m$)
and output ($O_m$) on the most loaded worker machine.
Recall that profiling revealed $\beta_2/\beta_3 \approx 4$, i.e., each input
tuple incurs 4 times the load compared to an output tuple.

\methodS and \methodVit achieve similar load characteristics, because
both intelligently leverage input and output samples to balance load while avoiding input duplication.
However, \methodS finds a better partitioning (lower max worker load and input duplication)
with 10x lower optimization time. The other two methods produce significantly higher input duplication, affecting
$I$ and $I_m$. \methodGrid still shows competitive running time because it works with a
very fine-grained partitioning, i.e., each worker receives its input already split into
many small grid cells. Data in a cell can be joined independent of the other cells, resulting in
efficient in-memory processing. As a result, \methodGrid has lower
per-tuple processing time than the other methods. (There each worker receives its input in a large
``chunk'' that needs to be range-partitioned.)

\subsubsection{Multiple Join Attributes}
\label{sec:expMultAttr}

\Cref{tab:time-pareto-3d,tab:time-real} show that
the performance gaps widen when joining on 3 attributes:
\emph{\methodS is the clear winner in total running time as well as join time alone}.
It finds the partitioning with the lowest max worker load, while keeping input duplication below 4\%,
while the competitors created up to 12x input duplication.

\methodVit is severely
hampered by the complexity of the optimization step. (Lowering optimization time resulted in higher join time
due to worse partitioning.) \methodGrid suffers from $\bigO(3^d)$ input duplication in $d$ dimensions.
For \methodOneB, note that the numbers in \Cref{tab:time-single-attr} and \Cref{tab:time-pareto-3d} are virtually
identical. This is due to the fact that it covers the entire join matrix $S \times T$, i.e., the matrix cover is not affected
by the dimensionality of the join condition.

\methodGrid has by far the highest input duplication, but again recovers some of this cost due to its faster
local processing that exploits that each worker's  input arrives already partitioned into small grid cells.
This is especially visible when comparing to \methodOneB for band width $(2, 2, 2)$ in
\Cref{tab:time-real}.

\begin{table*}[tbp]
\centering
\setlength{\tabcolsep}{1mm}
\caption{Skew resistance: 
\methodS is fastest and has much less data duplication than other methods.
(\texttt{Pareto-$z$}, $d=3$, band width ($2,2,2$), and increasing skew $z = 0.5, \ldots, 2$.)
\later{1-bucket and gird-eps for z=1.5 appear too low times, when looking at the trend. any possibility that something went wrong there in the runs? 
Did you run it just once, or 3 times and then took the median?
Answer: It should be due to system fluctuation as I only ran each query on AWS once. (Some I ran 2-3 times but it costs a lot for the ones with longer join time. }
}
\label{tab:pareto-time}
\scalebox{0.8}
{
    \begin{tabular}{|r||r|r|r|r||r|r|r||rrr|rrr|rrr|rrr|}
    \hline
    \multirow{2}{*}{Data Sets} 
	    & \multicolumn{4}{c||}{Runtime (\parTime+join time) in [sec]}
        & \multicolumn{3}{c||}{Relative time over \methodS} 
        & \multicolumn{12}{c|}{I/O sizes in [millions]: $I$, $I_m$, $O_m$} \\
    \cline{2-20}
	    & \methodS & \methodVit & \methodOneB & \methodGrid  
		& \methodVit & \methodOneB & \methodGrid 
		& \multicolumn{3}{c|}{\methodS} & \multicolumn{3}{c|}{\methodVit} 
		& \multicolumn{3}{c|}{\methodOneB} & \multicolumn{3}{c|}{\methodGrid}   \\
    \hline
	\texttt{pareto-0.5} & 230 (3+227) & 609 (\colorbox{Red1}{263}+346) & 1137  & 1146  & \colorbox{Blue1}{2.65} & \colorbox{Blue1}{4.94} & \colorbox{Blue1}{4.98}  & 401 & 13 & 0.3 & 577  & 20 & 1   & 2200 & 73 & 0.4  & \colorbox{Red1}{5582} & 186 & 0.4 \\
	\texttt{pareto-1.0} & 290 (3+287) & 1064 (\colorbox{Red1}{525}+539) & 1235 & 1335  & \colorbox{Blue1}{3.67} & \colorbox{Blue1}{4.26} & \colorbox{Blue1}{4.60}  & 401 & 13 & 17  & 616  & 20 & 31  & 2200 & 73 & 14   & \colorbox{Red1}{5554} & 185 & 14  \\
	\texttt{pareto-1.5} & 344 (2+342) & 1339 (\colorbox{Red1}{694}+645) & 1149 & 1412  & \colorbox{Blue1}{3.89} & \colorbox{Blue1}{3.34} & \colorbox{Blue1}{4.10}  & 404 & 15 & 29  & 652  & 19 & 69  & 2200 & 73 & 37   & \colorbox{Red1}{5541} & 185 & 37  \\
	\texttt{pareto-2.0} & 485 (2+483) & 1811 (\colorbox{Red1}{1000}+811) & 1369 & 2417 & \colorbox{Blue1}{3.73} & \colorbox{Blue1}{2.82} & \colorbox{Blue1}{4.98}  & 406 & 14 & 111 & 747  & 19 & 168 & 2200 & 73 & 107  & \colorbox{Red1}{5522} & 184 & 107 \\
	\hline
	\end{tabular}
}
\end{table*}

\subsection{Skew Resistance}
\label{sec:expSkew}

\Cref{tab:pareto-time} investigates the impact of join-attribute skew, showing that
\emph{\methodS handles it the best, again achieving the lowest max worker load with almost no
input duplication}. The competitors suffer from high input duplication; \methodVit also from high optimization cost.
Note that as skew increases, output size increases as well. This is due to the power-law distribution and the
correlation of high-frequency join-attribute values in the two inputs. 
Greater output size implies a denser join matrix for \methodVit, increasing its optimization time.

\begin{table*}[tbp]
\setlength{\tabcolsep}{1mm}
\caption{Scalability experiments: 
\methodS and \method have almost perfect scalability and beat all competitors.
(Dataset $X/Y/w$ in (a) and (b) refers to $X$ million input and $Y$ million output on $w$ workers.)}
\label{tab:scalability}
    \begin{subtable}{\linewidth}
	\centering
	\caption{\texttt{Pareto-1.5}, $d=3$, band width ($2, 2, 2$).}
	\label{tab:scal-pareto}
	\scalebox{0.8}
	{
	    \begin{tabular}{|r||r|r|r|r||r|r|r||rrr|rrr|rrr|rrr|}
	    \hline
	    \multirow{2}{*}{Data Sets} 
		    & \multicolumn{4}{c||}{Runtime (\parTime+join time) in [sec]}
	        & \multicolumn{3}{c||}{Relative time over \methodS} 
	        & \multicolumn{12}{c|}{I/O sizes in [millions]: $I$, $I_m$, $O_m$} \\
	    \cline{2-20}
		    & \methodS & \methodVit & \methodOneB & \methodGrid  
			& \methodVit & \methodOneB & \methodGrid 
			& \multicolumn{3}{c|}{\methodS} & \multicolumn{3}{c|}{\methodVit} 
			& \multicolumn{3}{c|}{\methodOneB} & \multicolumn{3}{c|}{\methodGrid}   \\
	    \hline
        $200/282/15$  & 306 (1+305) & 1227 (\colorbox{Red1}{767}+460) & 779  & 1381    &  \colorbox{Blue1}{4.01} & \colorbox{Blue1}{2.55} & \colorbox{Blue1}{4.51}  &  202  & 13 & 20 & 290  & 19 & 36 & 800  & 53  & 19 & \colorbox{Red1}{2772}  & 185 & 19 \\
        $400/1120/30$ & 344 (2+342) & 1374 (\colorbox{Red1}{729}+645) & 1149 & 1412    &  \colorbox{Blue1}{3.99} & \colorbox{Blue1}{3.34} & \colorbox{Blue1}{4.10}  &  404  & 15 & 29 & 652  & 19 & 69 & 2200 & 73  & 37 & \colorbox{Red1}{5541}  & 185 & 37 \\
        $800/4460/60$ & 438 (4+434) & 1721 (\colorbox{Red1}{801}+920) & 1731 & \textsc{failed}  
																	  &  \colorbox{Blue1}{3.93} & \colorbox{Blue1}{3.95} & \colorbox{Blue1}{N/A} &  809  & 21 & 45 & 1690 & 42 & 74 & 6400 & 107 & 74 & \colorbox{Red1}{11089} & 185 & 74 \\
		\hline
		\end{tabular}
	}
    \end{subtable}
    \vspace{5px}
    \newline
    \begin{subtable}{\linewidth}
	\centering
	\caption{Join of \texttt{ebird} with \texttt{cloud}, $d=3$, band width ($2, 2, 2$).}
	\label{tab:scal-real}
    \scalebox{0.8}
	{
	    \begin{tabular}{|r||r|r|r|r||r|r|r||rrr|rrr|rrr|rrr|}
	    \hline
	    \multirow{2}{*}{Data Sets} 
		    & \multicolumn{4}{c||}{Runtime (\parTime+join time) in [sec]}
	        & \multicolumn{3}{c||}{Relative time over \methodS} 
	        & \multicolumn{12}{c|}{I/O sizes in [millions]: $I$, $I_m$, $O_m$} \\
	    \cline{2-20}
		    & \methodS & \methodVit & \methodOneB & \methodGrid  
			& \methodVit & \methodOneB & \methodGrid 
			& \multicolumn{3}{c|}{\methodS} & \multicolumn{3}{c|}{\methodVit} 
			& \multicolumn{3}{c|}{\methodOneB} & \multicolumn{3}{c|}{\methodGrid}   \\
	    \hline
        $222/134/15$  & 207 (3+204) & 1213 (\colorbox{Red1}{942}+271)   & 547  & 812  & \colorbox{Blue1}{5.86} & \colorbox{Blue1}{2.64} & \colorbox{Blue1}{3.92}  &  223  & 15 & 11 & 307  & 22 & 11 & 856  & 57  & 9  & \colorbox{Red1}{2688} & 179 & 9  \\
        $445/530/30$  & 193 (3+190) & 1778 (\colorbox{Red1}{1447}+331)  & 688  & 771  & \colorbox{Blue1}{9.21} & \colorbox{Blue1}{3.56} & \colorbox{Blue1}{3.99}  &  448  & 16 & 14 & 748  & 26 & 27 & 2420 & 81  & 18 & \colorbox{Red1}{5403} & 180 & 18 \\
        $890/2000/60$ & 215 (2+213) & 1919 (\colorbox{Red1}{1479}$^*$+440) & 1117 & 793  & \colorbox{Blue1}{8.93} & \colorbox{Blue1}{5.20} & \colorbox{Blue1}{3.69}  &  899  & 13 & 44 & 2040 & 38 & 35 & 6870 & 114 & 36 & \colorbox{Red1}{10805} & 180 & 36 \\
		\hline
        \multicolumn{20}{l}{$^*$ \small \ParTime{} for $890/2000/60$ is similar to that of $445/530/30$ after we tuned parameters for \methodVit s.t.\ it could finish optimization within 90 minutes.} \\
		\end{tabular}
	}
    \end{subtable}
\vspace{5px}
\newline
	\begin{subtable}{\linewidth}\centering
	\caption{Varying input size: \texttt{pareto-1.5}, $d=8$, band width is 20 in each dimension, 30 workers.}
	\label{tab:varying-input}
	\scalebox{0.8}{
	\begin{tabular}{|r|r||r|r|r|r||rrr|rrr|rrr|rrr|}
		\hline
		\multirow{2}{*}{\begin{tabular}[r]{@{}r@{}}Input Size \\in [millions]\end{tabular}} &
          \multirow{2}{*}{\begin{tabular}[r]{@{}r@{}}Join Result Size \\in [millions]\end{tabular}} &
		  \multicolumn{4}{c||}{Runtime (\parTime+join time) in [sec]} & \multicolumn{12}{c|}{I/O sizes in [millions]: $I,I_m,O_m$} \\ 
		\cline{3-18}
		 & & \method & \methodVit & \methodOneB & \methodGrid
			& \multicolumn{3}{c|}{\method} & \multicolumn{3}{c|}{\methodVit} 
			& \multicolumn{3}{c|}{\methodOneB} & \multicolumn{3}{c|}{\methodGrid}   \\
		\hline
	$100$ & 9 & \colorbox{Blue1}{61} (5+56) & \colorbox{Blue1}{528} (\colorbox{Red1}{449}+79) & \colorbox{Blue1}{292} & \colorbox{Blue1}{173,581} & 104 & 3 & 2 & 142 & 5 & 1 & 550 & 18 & 0.3 & \colorbox{Red1}{297,421} & 9,914 & 0.3 \\
	$200$ & 57 & \colorbox{Blue1}{120} (5+115) & \colorbox{Blue1}{612} (\colorbox{Red1}{448}+164) & \colorbox{Blue1}{587} & \colorbox{Blue1}{347,944} & 210 & 7 & 2 & 285 & 10 & 5 & 1100 & 37 & 2 & \colorbox{Red1}{594,834} & 19,828 & 2 \\
	$400$ & 219 & \colorbox{Blue1}{240} (8+232) & \colorbox{Blue1}{760} (\colorbox{Red1}{418}+342) & \colorbox{Blue1}{1180} & \colorbox{Blue1}{694,574} & 420 & 14 & 7 & 574 & 7 & 67 & 2200 & 73 & 7 & \colorbox{Red1}{1,189,996} & 39,667 & 7 \\
	$800$ & 857 & \colorbox{Blue1}{510} (17+493) & \colorbox{Blue1}{1166} (\colorbox{Red1}{423}+743) & \colorbox{Blue1}{2390} & \colorbox{Blue1}{$1.39\cdot 10^6$} & 847 & 26 & 31 & 1180 & 53 & 4 & 4400 & 147 & 29 & \colorbox{Red1}{2,379,329} & 79,311 & 29 \\\hline
	\end{tabular}}
    \end{subtable}
\vspace{5px}
\newline
	\begin{subtable}{\linewidth}\centering
	\caption{Varying number of workers ($w$): \texttt{pareto-1.5}, $d=8$, band width is 20 in each dimension, input size 400 million.}
	\label{tab:varying-w}
	\scalebox{0.8}{
	\begin{tabular}{|r|r||r|r|r|r||rrr|rrr|rrr|rrr|}
		\hline
		\multirow{2}{*}{$w$} & \multirow{2}{*}{\begin{tabular}[r]{@{}r@{}}Join Result Size \\in [millions]\end{tabular}} &
		  \multicolumn{4}{c||}{Runtime (\parTime+join time) in [sec]} & \multicolumn{12}{c|}{I/O sizes in [millions]: $I,I_m,O_m$} \\ 
		\cline{3-18}
		 & & \method & \methodVit & \methodOneB & \methodGrid
			& \multicolumn{3}{c|}{\method} & \multicolumn{3}{c|}{\methodVit} 
			& \multicolumn{3}{c|}{\methodOneB} & \multicolumn{3}{c|}{\methodGrid}   \\
		\hline
	1 & \multirow{4}{*}{219} & \colorbox{Blue1}{3655} & \colorbox{Blue1}{3655} & \colorbox{Blue1}{3655} & \colorbox{Blue1}{8,527,502} & 400 & 400 & 219 & 400 & 400 & 219 & 400 & 400 & 219 & \colorbox{Red1}{1,189,996} & 1,189,996 & 219 \\
	15 &  & \colorbox{Blue1}{358} (5+353) & \colorbox{Blue1}{710} (\colorbox{Red1}{190}+520) & \colorbox{Blue1}{1295} & \colorbox{Blue1}{1,040,000} & 420 & 28 & 10 & 565 & 40 & 29 & 1600 & 107 & 15 & \colorbox{Red1}{1,189,996} & 79,333 & 15 \\
	30 &  & \colorbox{Blue1}{240} (8+232) & \colorbox{Blue1}{760} (\colorbox{Red1}{418}+342) & \colorbox{Blue1}{1180} & \colorbox{Blue1}{695,000} & 420 & 14 & 7 & 574 & 7 & 67 & 2200 & 73 & 7 & \colorbox{Red1}{1,189,996} & 39,667 & 7 \\
	60 &  & \colorbox{Blue1}{182} (10+172) & \colorbox{Blue1}{3703} (\colorbox{Red1}{3431}+272) & \colorbox{Blue1}{1287} & \colorbox{Blue1}{525,000} & 425 & 6 & 5 & 619 & 13 & 2 & 3200 & 53 & 4 & \colorbox{Red1}{1,189,996} & 19,833 & 4 \\\hline
	\end{tabular}}
    \end{subtable}
\end{table*}

\subsection{Scalability}
\label{sec:expScalability}

\Cref{tab:scal-pareto,tab:scal-real,tab:varying-input,tab:varying-w} show that \emph{\methodS and \method have almost perfect scalability
and beat all competitors}. In \Cref{tab:scal-pareto,tab:scal-real}, from row to row, we double both input size and
number of workers. In \Cref{tab:varying-input}, only the input size varies while the number of workers is constant.
In \Cref{tab:varying-w}, we only change the number of workers. The latter two results are for an 8D band-join
to explore \emph{which techniques can scale beyond dimensionality common today}.
For cost reasons, we use the running-time model to predict join time in \Cref{tab:varying-input,tab:varying-w}.
For queries on real data, the smaller inputs are random samples from the full data. Note that join output grows super-linearly,
therefore perfect scalability cannot be achieved. Nevertheless, \methodS and \method come close to the ideal when taking
the output increase into account. When the same query is run on different-sized clusters
(\Cref{tab:varying-w}), \method scales out best. \methodVit's optimization time grows
substantially as $w$ increases. We explored various settings, but reducing optimization time
resulted in even higher join time and vice versa. The numbers shown represent the best tradeoff found. 
\methodGrid failed on the largest synthetic input due to a memory
exception caused by one of the grid cells receiving too many input records.

\begin{table}[tbp]
\setlength{\tabcolsep}{1mm}
	\centering
	\caption{\methodGrid vs. \methodGridOpt on \texttt{pareto-$1.5$}, band width ($2,2,2$), varying grid size ($I$, $I_m$ and $O_m$ in [millions]).}
	\label{tab:vary-grid}
	\scalebox{0.8}
	{
		\begin{tabular}{|r|rrrr||rrrr|}
		\hline
		\multicolumn{5}{|c||}{\methodGrid} &  $I$\quad & $I_m$ & $O_m$ & Join Time \\ \hline
			  Grid Size & $I$\quad & $I_m$ & $O_m$ & Join Time &  \multicolumn{4}{c|}{\methodGridOpt} \\ \cline{1-5}
		      (1,1,1) & 5610 & 180 & 38 & 2993 & 460 & 16 & 46 & \colorbox{Blue1}{335} \\ \cline{6-9}
			  (2,2,2) & 5541 & 185 & 37 & 3021 & \multicolumn{4}{c|}{\methodS} \\
		      (4,4,4) & 1780 & 60 & 38 & 1023 & 404 & 15 & 29 & \colorbox{Blue1}{286} \\ \cline{6-9}
		      (8,8,8) & 861 & 29 & 38 & 533 & \multicolumn{4}{c|}{\methodVit} \\
		      (16,16,16) & 582 & 20 & 39 & 389 &  652 & 19 & 69 & \colorbox{Blue1}{459}\\ \cline{6-9}
		      (32,32,32) & 478 & 16 & 42 & \colorbox{Blue1}{336} & \multicolumn{4}{c|}{\methodOneB} \\
		      (64,64,64) & 435 & 15 & 56 & 344 & 2200 & 73 & 37 & \colorbox{Blue1}{1236} \\\hline
		\end{tabular}
	}
\end{table}

\begin{table}[tbp]
\setlength{\tabcolsep}{1mm}
\centering
\caption{\methodGridOpt vs. \method. I/O sizes in [millions].}
\label{tab:gridopt}
\scalebox{0.8}
{
	\begin{tabular}{|r|r||rrr|r|rrr|}
	\hline
	\multirow{2}{*}{Data Sets}
        & \multirow{2}{*}{Band width}
	    & \multicolumn{3}{c|}{\method} & \multicolumn{4}{c|}{\methodGridOpt} \\
		\cline{3-9}
		& & $I$\qquad & $I_m$ & $O_m$ & Grid Size & $I$\quad & $I_m$ & $O_m$ \\
	\hline
    \texttt{pareto-$2.0$} & (2,2,2)
        & 406 & 14 & 111 & 8 & 497 & 17 & 130 \\
    \texttt{rv-pareto-$1.5$} & (1K,1K,1K)
        & 400 & 13 & 0 & 2750 & 882 & \colorbox{Red1}{237} & 0 \\
    \texttt{rv-pareto-$1.5$} & (2K,2K,2K)
    	& 401 & 13 & 0 & 11500 & 1207 & \colorbox{Red1}{401} & 0 \\
	\hline
	\end{tabular}
}
\end{table}

\subsection{Optimizing Grid Size}
\label{sec:vary-grid}

\Cref{tab:vary-grid} shows that grid granularity has a significant impact on join time, here model estimated,
of \methodGrid. With the default grid size $(2, 2, 2)$, join time is 9x
higher compared to grid size $(32, 32, 32)$, caused by input duplication.
Our extension \emph{\methodGridOpt automatically explores different grid sizes},
using the same running-time model $\mathcal{M}$ as \method and
\methodVit to find the best setting.
Starting with grid size $\varepsilon_i$ in dimension $i$ for all join attributes $A_i$,
it tries coarsening the grid to size $j \cdot \varepsilon_i$ in dimension $i$, for $j =2, 3,\ldots$ For each
resulting grid partitioning $\mathcal{G}$, we execute $\mathcal{M}(\mathcal{G})$ to let model
$\mathcal{M}$ predict the running time, until a local minimum is found.

Automatic grid tuning works well for \texttt{Pareto-$z$} where both
inputs are similarly distributed and band width is small: \Cref{lem:gridBalanceGood} applies with
large $c_1$ and small $c_2$, providing strong upper bounds on the amount of input in
any $\varepsilon$-range. This confirms that grid-partitioning can indeed work well for ``sufficiently large''
input even in 3D space. However, \methodGridOpt fails on the reverse Pareto
distribution as \Cref{tab:gridopt} shows.
There $S$ and $T$ have very different density, resulting in small $c_1$ and large $c_2$, and
therefore much weaker upper bounds on the input per $\varepsilon$-range. The resulting dense
regions, as stated by \Cref{lem:gridBalanceBad}, cause high input duplication and
high input $I_m$ assigned to the most loaded worker.

\begin{table}[tbp]
\setlength{\tabcolsep}{1mm}
\centering
\caption{Comparing to distributed IEJoin: Input duplication and max worker load on \texttt{pareto-$z$}, $w=30$, varying skew and
band width ($I$, $I_m$ and $O_m$ in [millions]).}
\label{tab:comp-iejoin}
\scalebox{0.8}{
	\begin{tabular}{|r|r|r||rrr||rrr|r|}
	\hline
	\multirow{2}{*}{$Z$} & Band & \multirow{2}{*}{\begin{tabular}[c]{@{}r@{}}Output \\size\end{tabular}} & \multicolumn{3}{c||}{\methodS} & \multicolumn{4}{c|}{IEJoin} \\
	\cline{4-10}
	 & width &  & $I$ & $I_m$ & $O_m$ & $I$ & $I_m$ & $O_m$ & sizePerBlock \\
	\hline
	\multirow{3}{*}{1.5} & \multirow{3}{*}{{[}0,0,0{]}} & \multirow{3}{*}{0} & \multirow{3}{*}{401} & \multirow{3}{*}{14} & \multirow{3}{*}{0} & 780 & 40 & 0 & 10000 \\
	 &  &  &  &  &  & 726 & 25 & 0 & \colorbox{Blue1}{12524} \\
	 &  &  &  &  &  & 756 & 28 & 0 & 14000 \\\hline
	\multirow{3}{*}{1.5} & \multirow{3}{*}{{[}2,2,2{]}} & \multirow{3}{*}{1120} & \multirow{3}{*}{404} & \multirow{3}{*}{15} & \multirow{3}{*}{29} & 1092 & 48 & 14 & 6000 \\
	 &  &  &  &  &  & 1070 & 45 & 21 & \colorbox{Blue1}{7422} \\
	 &  &  &  &  &  & 1062 & 36 & 85 & 9000 \\\hline
	\multirow{3}{*}{1.0} & \multirow{3}{*}{{[}2,2,2{]}} & \multirow{3}{*}{420} & \multirow{3}{*}{401} & \multirow{3}{*}{13} & \multirow{3}{*}{17} & 1176 & 40 & 21 & 4000 \\
	 &  &  &  &  &  & 1080 & 37 & 26 & \colorbox{Blue1}{6263} \\
	 &  &  &  &  &  & 1088 & 48 & 4 & 8000 \\\hline
	\multirow{3}{*}{0.5} & \multirow{3}{*}{{[}2,2,2{]}} & \multirow{3}{*}{12} & \multirow{3}{*}{401} & \multirow{3}{*}{13} & \multirow{3}{*}{0.3} & 828 & 24 & 1 & 6000 \\
	 &  &  &  &  &  & 796 & 17 & 2 & \colorbox{Blue1}{8295} \\
	 &  &  &  &  &  & 820 & 20 & 2 & 10000 \\\hline
	\end{tabular}}
\end{table}

\subsection{Comparing to Distributed IEJoin}
\label{sec:exp-iejoin}

\Cref{tab:comp-iejoin} shows representative results for a comparison to the quantile-based partitioning
used by IEJoin~\cite{khayyat2017fast}. We explore a wide range of inter-quantile
ranges (sizePerBlock) and report results for those at and near the best setting found.
Note that \method can use the
same local processing algorithm as IEJoin, therefore we are interested in comparing based on
\emph{the quality of the partitioning}, i.e., $I$, $I_m$, and $O_m$.
It is clearly visible that \emph{\methodS finds significantly better partitionings}, providing more evidence
that simple quantile-based partitioning of the join matrix does not suffice.

\begin{table}[tbp]
\setlength{\tabcolsep}{1mm}
\centering
\caption{Impact of local join algorithm: Input duplication $I$ vs max worker load $L_m = 4 I_m + O_m$ for varying ratios $\beta_2/\beta_1$
for join of \texttt{ebird} with \texttt{cloud}, band width $(2,2,2)$, $w=30$.}
\label{tab:varyingWeights}
\scalebox{0.8}{
\begin{tabular}{|r||r|r||r|r||r|r||r|r|}\hline
\multirow{2}{*}{$\beta_2 / \beta_1$} & \multicolumn{2}{c||}{\method} & \multicolumn{2}{c||}{\methodVit} & \multicolumn{2}{c||}{\methodOneB} & \multicolumn{2}{c|}{\methodGrid} \\
\cline{2-9}
 & $I$ & $L_m$ & $I$ & $L_m$ & $I$ & $L_m$ & $I$ & $L_m$ \\
\hline
0.0001 & 890.34 & 289 & \multirow{9}{*}{1830} & \multirow{9}{*}{502} & \multirow{9}{*}{4832} & \multirow{9}{*}{711} & \multirow{9}{*}{10800} & \multirow{9}{*}{1518} \\
0.001 & 890.36 & 223 &  &  &  &  &  & \\
0.01 & 890.42 & 195 &  &  &  &  &  & \\
0.1 & 890.52 & 191 &  &  &  &  &  & \\
1 & 890.8 & 189 &  &  &  &  &  & \\
10 & 890.8 & 189 &  &  &  &  &  & \\
100 & 890.8 & 189 &  &  &  &  &  & \\
1000 & 890.8 & 189 &  &  &  &  &  & \\
10000 & 890.8 & 189 &  &  &  &  &  & \\\hline
\end{tabular}}
\end{table}

\subsection{Impact of Local Join Algorithm}
\label{sec:exp-shuffleVSlocal}

\Cref{tab:varyingWeights} shows a typical result of the impact of ratio $\beta_2 / \beta_1$.
A high ratio occurs in systems with fast data transfer and slow local computation. Replacing the local band-join
algorithm with a faster one would lower the ratio. While
the competitors are not affected by the ratio (they all ignore network shuffle time), it is visible
how \emph{increasing weight on local join cost makes \method reduce max worker load}, incurring slightly
higher input duplication.

\begin{table}[ptb]
\setlength{\tabcolsep}{1mm}
\centering
\caption{\methodS vs. \method (I/O sizes in [millions]).}
\label{tab:dupBothOrDupT}
\scalebox{0.8}
{
	\begin{tabular}{|r|r||rrr|rrr|}
	\hline
	\multirow{2}{*}{Data Sets}
        & \multirow{2}{*}{Band width}
	    & \multicolumn{3}{c|}{\methodS} 
		& \multicolumn{3}{c|}{\method} \\
	\cline{3-8}
		& & $I$\qquad & $I_m$ & $O_m$ 
		& $I$\quad & $I_m$ & $O_m$ \\
	\hline
    \texttt{pareto-$1.0$} & (2,2,2)
		& 401 & 13 & 17  
        & 401 & 12 & 21 \\
    \texttt{ebird} and \texttt{cloud} & (0,0,0)
		& 890 & 30 & 0  
        & 890 & 30 & 0 \\		
    \texttt{ebird} and \texttt{cloud} & (2,2,2)
		& 899 & 32 & 66  
        & 891 & 31 & 67 \\		
    \texttt{ebird} and \texttt{cloud} & (4,4,4)
		& 918 & 31 & 567  
        & 894 & 30 & 515 \\
	\hline		
    \texttt{rv-pareto-$1.5$} & (1000,1000,1000)
		& 452 & \colorbox{Red1}{143} & 0  
        & 400 & 13 & 0 \\		
    \texttt{rv-pareto-$1.5$} & (2000,2000,2000)
		& 430 & \colorbox{Red1}{173} & 0  
    	& 401 & 13 & 0 \\		
    \texttt{rv-pareto-$1.5$} & 2
		& 433 & 40 & 0  
    	& 401 & 14 & 0\\		
    \texttt{rv-pareto-$1.5$} & 1000
		& 402 & \colorbox{Red1}{200} & 0
    	& 402 & 14 & 0 \\		
	\hline
	\end{tabular}
}
\end{table}

\subsection{Impact of Symmetric Partitioning}
\label{sec:exp-symmetric}

While \method and \methodS find similar partitionings on \texttt{pareto-1.0} and the real data,
the advantages of symmetric partitioning are revealed on the reverse Pareto data in \Cref{tab:dupBothOrDupT}.
Here \methodS cannot split regions with high density of $T$ without incurring high input duplication.
In contrast, \method switches the roles of $S$ and $T$, because $S$ is sparse in those regions
and hence the split creates few input duplicates.

\begin{figure}[tbp]
\centering
\includegraphics[width=0.9\linewidth]{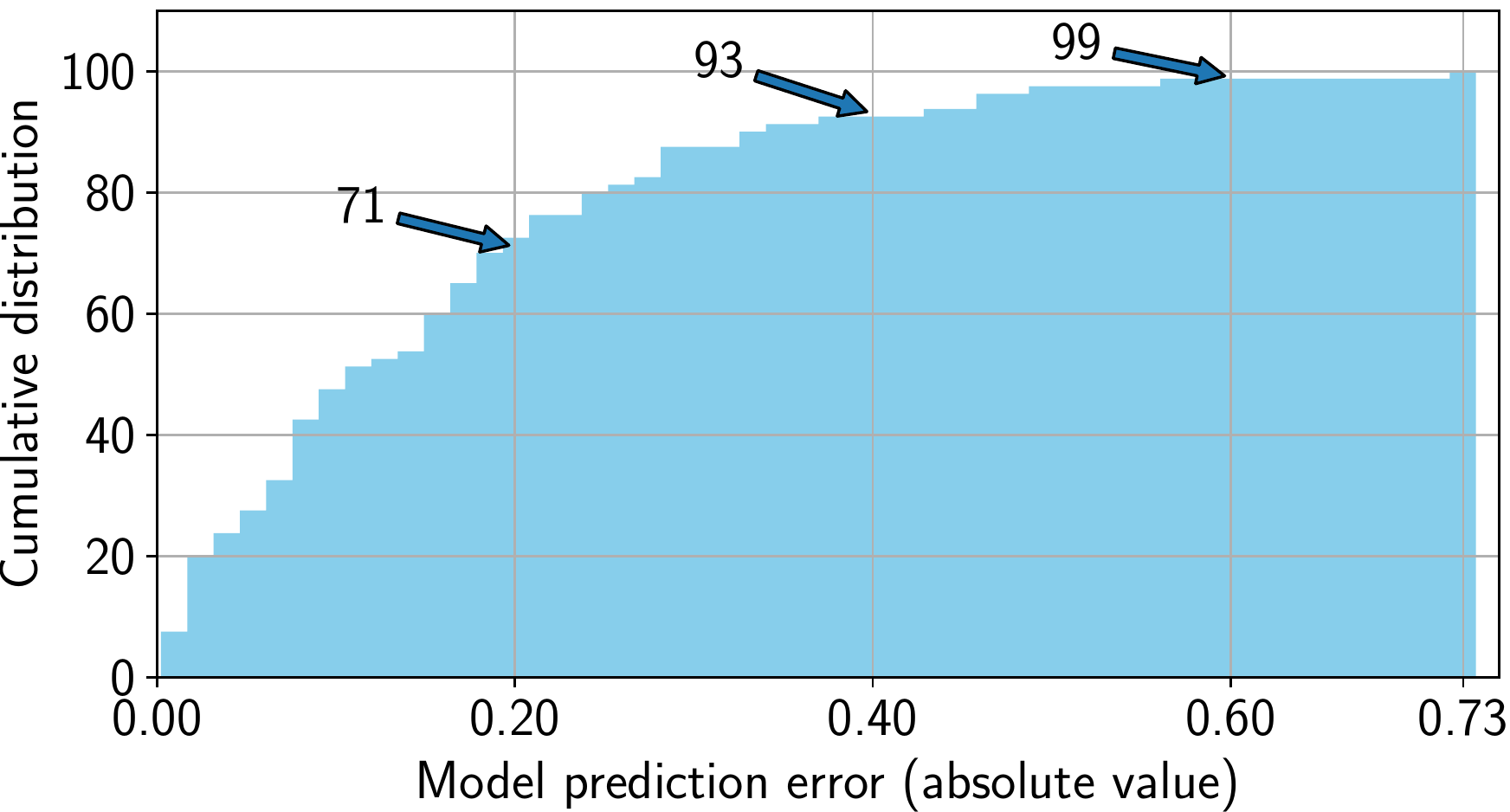}
\caption{
Accuracy of the running-time model: 
Cumulative distribution of model error.
}
\label{fig:model-cumulative-dist}
\end{figure}

\subsection{Accuracy of the Running-Time Model}
\label{sec:exp-model}

In all experiments, the running-time model's predictions were sufficiently accurate for identifying excellent
partitionings for \method. Due to the high cost in terms of time and money of executing computations in the cloud,
we sometimes rely on the running-time model also to report join time.
\Cref{fig:model-cumulative-dist} shows the cumulative distribution of the relative
error between predicted and measured join time for over 80 experiments selected randomly from all runs.
E.g., the model has less than $20\%$ error in over $70\%$ of the cases and it is never off by more than
a factor 1.8.

\subsection{Near-Optimality of Our Approach}
\label{sec:exp-4alg-summary}

\resultbox{
\Cref{fig:comp-overhead} in \Cref{sec:intro} summarizes the results from all tables, making it easy to
see how \method achieves significantly lower max worker load ($y$-axis)
with less input duplication ($x$-axis).
\method is always within 10\% of the lower bound
on both measures, beating the competition by a wide margin.
}

\section{Conclusions}
\label{sec:concl}

For distributed band-joins,
we showed that recursive partitioning with the appropriate split-scoring measure achieves both low
optimization time (a few seconds) and near-optimal join partitioning (within 10\% of the respective
lower bounds) on input duplication and max worker load.
Even if system parameters change, e.g., due to faster networks or CPUs,
\method's achievement will still stand, because the issues of low input duplication and low input and
output per worker will remain crucial optimization goals.

There are several exciting directions for future work. First, for band-joins between more than two relations,
can one do better than reducing the problem to multiple two-way joins? Second, how do we partition
for join conditions that contain a mix of various types of predicates, e.g., equality, inequality, band, and non-equality?
Third, what other types of join conditions give rise to specialized methods, like \method for band-join,
so that similarly significant improvements over generic theta-join approaches can be achieved?

More information about \method and other techniques for distributed data-intensive computations
can be found at \url{https://northeastern-datalab.github.io/distributed/}.

\begin{acks}
This work was supported in part by
the National Institutes of Health (NIH) under award number R01 NS091421 and by
the National Science Foundation (NSF) under award number CAREER IIS-1762268.
The content is solely the responsibility of the authors and does not necessarily represent
the official views of NIH or NSF.
We also would like to thank Aditya Ghosh for his contributions to design and implementation of \methodGridOpt,
Dhavalkumar Patel for implementing the join-sampling algorithm,
and the anonymous reviewers and Niklas Smedemark-Margulies for their constructive feedback.
\end{acks}
\bibliographystyle{ACM-Reference-Format}
\balance
\bibliography{band-join}


\begin{thebibliography}{44}

%%% the Web address.

\ifx \showCODEN    \undefined \def \showCODEN     #1{\unskip}     \fi
\ifx \showDOI      \undefined \def \showDOI       #1{#1}\fi
\ifx \showISBNx    \undefined \def \showISBNx     #1{\unskip}     \fi
\ifx \showISBNxiii \undefined \def \showISBNxiii  #1{\unskip}     \fi
\ifx \showISSN     \undefined \def \showISSN      #1{\unskip}     \fi
\ifx \showLCCN     \undefined \def \showLCCN      #1{\unskip}     \fi
\ifx \shownote     \undefined \def \shownote      #1{#1}          \fi
\ifx \showarticletitle \undefined \def \showarticletitle #1{#1}   \fi
\ifx \showURL      \undefined \def \showURL       {\relax}        \fi
% invisible to TeX
\providecommand\bibfield[2]{#2}
\providecommand\bibinfo[2]{#2}
\providecommand\natexlab[1]{#1}
\providecommand\showeprint[2][]{arXiv:#2}

\bibitem[\protect\citeauthoryear{Beame, Koutris, and Suciu}{Beame
  et~al\mbox{.}}{2014}]%
        {beame2014skew}
\bibfield{author}{\bibinfo{person}{Paul Beame}, \bibinfo{person}{Paraschos
  Koutris}, {and} \bibinfo{person}{Dan Suciu}.}
  \bibinfo{year}{2014}\natexlab{}.
\newblock \showarticletitle{Skew in parallel query processing}. In
  \bibinfo{booktitle}{\emph{PODS}}. \bibinfo{pages}{212--223}.
\newblock
\urldef\tempurl%
\url{https://doi.org/10.1145/2594538.2594558}
\showDOI{\tempurl}


\bibitem[\protect\citeauthoryear{Binnig, Crotty, Galakatos, Kraska, and
  Zamanian}{Binnig et~al\mbox{.}}{2016}]%
        {binnigCGKZ16:endOfSlowNetworks}
\bibfield{author}{\bibinfo{person}{Carsten Binnig}, \bibinfo{person}{Andrew
  Crotty}, \bibinfo{person}{Alex Galakatos}, \bibinfo{person}{Tim Kraska},
  {and} \bibinfo{person}{Erfan Zamanian}.} \bibinfo{year}{2016}\natexlab{}.
\newblock \showarticletitle{The End of Slow Networks: It's Time for a
  Redesign}.
\newblock \bibinfo{journal}{\emph{{PVLDB}}} \bibinfo{volume}{9},
  \bibinfo{number}{7} (\bibinfo{year}{2016}), \bibinfo{pages}{528--539}.
\newblock
\urldef\tempurl%
\url{https://doi.org/10.14778/2904483.2904485}
\showDOI{\tempurl}


\bibitem[\protect\citeauthoryear{Blanas, Patel, Ercegovac, Rao, Shekita, and
  Tian}{Blanas et~al\mbox{.}}{2010}]%
        {blanas2010comparison}
\bibfield{author}{\bibinfo{person}{Spyros Blanas}, \bibinfo{person}{Jignesh~M
  Patel}, \bibinfo{person}{Vuk Ercegovac}, \bibinfo{person}{Jun Rao},
  \bibinfo{person}{Eugene~J Shekita}, {and} \bibinfo{person}{Yuanyuan Tian}.}
  \bibinfo{year}{2010}\natexlab{}.
\newblock \showarticletitle{A comparison of join algorithms for log processing
  in mapreduce}. In \bibinfo{booktitle}{\emph{SIGMOD}}.
  \bibinfo{pages}{975--986}.
\newblock
\urldef\tempurl%
\url{https://doi.org/10.1145/1807167.1807273}
\showDOI{\tempurl}


\bibitem[\protect\citeauthoryear{Bruno, Kwon, and Wu}{Bruno
  et~al\mbox{.}}{2014}]%
        {bruno2014advanced}
\bibfield{author}{\bibinfo{person}{Nicolas Bruno}, \bibinfo{person}{YongChul
  Kwon}, {and} \bibinfo{person}{Ming-Chuan Wu}.}
  \bibinfo{year}{2014}\natexlab{}.
\newblock \showarticletitle{Advanced join strategies for large-scale
  distributed computation}.
\newblock \bibinfo{journal}{\emph{PVLDB}} \bibinfo{volume}{7},
  \bibinfo{number}{13} (\bibinfo{year}{2014}), \bibinfo{pages}{1484--1495}.
\newblock
\urldef\tempurl%
\url{https://doi.org/10.14778/2733004.2733020}
\showDOI{\tempurl}


\bibitem[\protect\citeauthoryear{Chu, Balazinska, and Suciu}{Chu
  et~al\mbox{.}}{2015}]%
        {chu2015theory}
\bibfield{author}{\bibinfo{person}{Shumo Chu}, \bibinfo{person}{Magdalena
  Balazinska}, {and} \bibinfo{person}{Dan Suciu}.}
  \bibinfo{year}{2015}\natexlab{}.
\newblock \showarticletitle{From theory to practice: Efficient join query
  evaluation in a parallel database system}. In
  \bibinfo{booktitle}{\emph{SIGMOD}}. \bibinfo{pages}{63--78}.
\newblock
\urldef\tempurl%
\url{https://doi.org/10.1145/2723372.2750545}
\showDOI{\tempurl}


\bibitem[\protect\citeauthoryear{Dean and Ghemawat}{Dean and Ghemawat}{2008}]%
        {Dean:2008:mapreduce}
\bibfield{author}{\bibinfo{person}{Jeffrey Dean} {and} \bibinfo{person}{Sanjay
  Ghemawat}.} \bibinfo{year}{2008}\natexlab{}.
\newblock \showarticletitle{MapReduce: Simplified Data Processing on Large
  Clusters}.
\newblock \bibinfo{journal}{\emph{Commun. ACM}} \bibinfo{volume}{51},
  \bibinfo{number}{1} (\bibinfo{year}{2008}), \bibinfo{pages}{107--113}.
\newblock
\urldef\tempurl%
\url{https://doi.org/10.1145/1327452.1327492}
\showDOI{\tempurl}


\bibitem[\protect\citeauthoryear{Dewan, Mok, Hern{\'a}ndez, and Stolfo}{Dewan
  et~al\mbox{.}}{1994}]%
        {dewan1994predictive}
\bibfield{author}{\bibinfo{person}{Hasanat~M Dewan}, \bibinfo{person}{Kui~W
  Mok}, \bibinfo{person}{Mauricio Hern{\'a}ndez}, {and}
  \bibinfo{person}{Salvatore~J Stolfo}.} \bibinfo{year}{1994}\natexlab{}.
\newblock \showarticletitle{Predictive dynamic load balancing of parallel
  hash-joins over heterogeneous processors in the presence of data skew}. In
  \bibinfo{booktitle}{\emph{PDIS}}. \bibinfo{pages}{40--49}.
\newblock
\showISBNx{0818664010}
\urldef\tempurl%
\url{https://doi.org/10.1109/PDIS.1994.331734}
\showDOI{\tempurl}


\bibitem[\protect\citeauthoryear{DeWitt and Gerber}{DeWitt and Gerber}{1985}]%
        {dewitt1985multiprocessor}
\bibfield{author}{\bibinfo{person}{David~J DeWitt} {and}
  \bibinfo{person}{Robert Gerber}.} \bibinfo{year}{1985}\natexlab{}.
\newblock \showarticletitle{Multiprocessor hash-based join algorithms}. In
  \bibinfo{booktitle}{\emph{VLDB}}. \bibinfo{pages}{151--164}.
\newblock
\urldef\tempurl%
\url{https://dl.acm.org/doi/10.5555/1286760.1286774}
\showURL{%
\tempurl}


\bibitem[\protect\citeauthoryear{DeWitt, Naughton, and Schneider}{DeWitt
  et~al\mbox{.}}{1991}]%
        {dewitt1991evaluation}
\bibfield{author}{\bibinfo{person}{David~J DeWitt}, \bibinfo{person}{Jeffrey~F
  Naughton}, {and} \bibinfo{person}{Donovan~A Schneider}.}
  \bibinfo{year}{1991}\natexlab{}.
\newblock \showarticletitle{An evaluation of non-equijoin algorithms}. In
  \bibinfo{booktitle}{\emph{{VLDB}}}. \bibinfo{pages}{443--452}.
\newblock
\urldef\tempurl%
\url{https://dl.acm.org/doi/10.5555/645917.672320}
\showURL{%
\tempurl}


\bibitem[\protect\citeauthoryear{DeWitt, Naughton, Schneider, and
  Seshadri}{DeWitt et~al\mbox{.}}{1992}]%
        {dewitt1992practical}
\bibfield{author}{\bibinfo{person}{David~J DeWitt}, \bibinfo{person}{Jeffrey~F
  Naughton}, \bibinfo{person}{Donovan~A Schneider}, {and}
  \bibinfo{person}{Srinivasan Seshadri}.} \bibinfo{year}{1992}\natexlab{}.
\newblock \showarticletitle{Practical skew handling in parallel joins}. In
  \bibinfo{booktitle}{\emph{VLDB}}. \bibinfo{pages}{27--40}.
\newblock
\urldef\tempurl%
\url{https://dl.acm.org/doi/10.5555/645918.672512}
\showURL{%
\tempurl}


\bibitem[\protect\citeauthoryear{Duggan, Papaemmanouil, Battle, and
  Stonebraker}{Duggan et~al\mbox{.}}{2015}]%
        {duggan2015skew}
\bibfield{author}{\bibinfo{person}{Jennie Duggan}, \bibinfo{person}{Olga
  Papaemmanouil}, \bibinfo{person}{Leilani Battle}, {and}
  \bibinfo{person}{Michael Stonebraker}.} \bibinfo{year}{2015}\natexlab{}.
\newblock \showarticletitle{Skew-aware join optimization for array databases}.
  In \bibinfo{booktitle}{\emph{SIGMOD}}. \bibinfo{pages}{123--135}.
\newblock
\urldef\tempurl%
\url{https://doi.org/10.1145/2723372.2723709}
\showDOI{\tempurl}


\bibitem[\protect\citeauthoryear{Elseidy, Elguindy, Vitorovic, and
  Koch}{Elseidy et~al\mbox{.}}{2014}]%
        {elseidy2014scalable}
\bibfield{author}{\bibinfo{person}{Mohammed Elseidy}, \bibinfo{person}{Abdallah
  Elguindy}, \bibinfo{person}{Aleksandar Vitorovic}, {and}
  \bibinfo{person}{Christoph Koch}.} \bibinfo{year}{2014}\natexlab{}.
\newblock \showarticletitle{Scalable and adaptive online joins}.
\newblock \bibinfo{journal}{\emph{PVLDB}} \bibinfo{volume}{7},
  \bibinfo{number}{6} (\bibinfo{year}{2014}), \bibinfo{pages}{441--452}.
\newblock
\urldef\tempurl%
\url{https://doi.org/10.14778/2732279.2732281}
\showDOI{\tempurl}


\bibitem[\protect\citeauthoryear{Fier, Augsten, Bouros, Leser, and
  Freytag}{Fier et~al\mbox{.}}{2018}]%
        {FierABLF18:setSimilarityJoinExperimentSurvey}
\bibfield{author}{\bibinfo{person}{Fabian Fier}, \bibinfo{person}{Nikolaus
  Augsten}, \bibinfo{person}{Panagiotis Bouros}, \bibinfo{person}{Ulf Leser},
  {and} \bibinfo{person}{Johann{-}Christoph Freytag}.}
  \bibinfo{year}{2018}\natexlab{}.
\newblock \showarticletitle{Set Similarity Joins on MapReduce: An Experimental
  Survey}.
\newblock \bibinfo{journal}{\emph{{PVLDB}}} \bibinfo{volume}{11},
  \bibinfo{number}{10} (\bibinfo{year}{2018}), \bibinfo{pages}{1110--1122}.
\newblock
\urldef\tempurl%
\url{https://doi.org/10.14778/3231751.3231760}
\showDOI{\tempurl}


\bibitem[\protect\citeauthoryear{Furtado and Baumann}{Furtado and
  Baumann}{1999}]%
        {FurtadoB99:arrayTiling}
\bibfield{author}{\bibinfo{person}{Paula Furtado} {and} \bibinfo{person}{Peter
  Baumann}.} \bibinfo{year}{1999}\natexlab{}.
\newblock \showarticletitle{Storage of Multidimensional Arrays Based on
  Arbitrary Tiling}. In \bibinfo{booktitle}{\emph{{ICDE}}}.
  \bibinfo{pages}{480--489}.
\newblock
\urldef\tempurl%
\url{https://doi.org/10.1109/ICDE.1999.754964}
\showDOI{\tempurl}


\bibitem[\protect\citeauthoryear{Hahn, Warren, and Eastman}{Hahn
  et~al\mbox{.}}{2012}]%
        {hahn1999extended}
\bibfield{author}{\bibinfo{person}{CJ Hahn}, \bibinfo{person}{SG Warren}, {and}
  \bibinfo{person}{R Eastman}.} \bibinfo{year}{2012}\natexlab{}.
\newblock \showarticletitle{Extended Edited Synoptic Cloud Reports from Ships
  and Land Stations Over the Globe, 1952-2009 ({NDP-026C})}.
\newblock  (\bibinfo{year}{2012}).
\newblock
\urldef\tempurl%
\url{https://doi.org/10.3334/CDIAC/cli.ndp026c}
\showDOI{\tempurl}


\bibitem[\protect\citeauthoryear{Han, Pei, and Kamber}{Han
  et~al\mbox{.}}{2011}]%
        {han2011data}
\bibfield{author}{\bibinfo{person}{Jiawei Han}, \bibinfo{person}{Jian Pei},
  {and} \bibinfo{person}{Micheline Kamber}.} \bibinfo{year}{2011}\natexlab{}.
\newblock \bibinfo{booktitle}{\emph{Data mining: concepts and techniques}
  (\bibinfo{edition}{3rd} ed.)}.
\newblock \bibinfo{publisher}{Morgan Kaufmann}.
\newblock
\urldef\tempurl%
\url{https://doi.org/10.1016/C2009-0-61819-5}
\showDOI{\tempurl}


\bibitem[\protect\citeauthoryear{Harada and Kitsuregawa}{Harada and
  Kitsuregawa}{1995}]%
        {harada1995dynamic}
\bibfield{author}{\bibinfo{person}{Lilian Harada} {and} \bibinfo{person}{Masaru
  Kitsuregawa}.} \bibinfo{year}{1995}\natexlab{}.
\newblock \showarticletitle{Dynamic Join Product Skew Handling for Hash-Joins
  in Shared-Nothing Database Systems.}. In \bibinfo{booktitle}{\emph{DASFAA}}.
  \bibinfo{pages}{246--255}.
\newblock
\urldef\tempurl%
\url{https://dl.acm.org/doi/10.5555/646710.703010}
\showURL{%
\tempurl}


\bibitem[\protect\citeauthoryear{Hua and Lee}{Hua and Lee}{1991}]%
        {hua1991handling}
\bibfield{author}{\bibinfo{person}{Kien~A Hua} {and} \bibinfo{person}{Chiang
  Lee}.} \bibinfo{year}{1991}\natexlab{}.
\newblock \showarticletitle{Handling Data Skew in Multiprocessor Database
  Computers Using Partition Tuning}. In \bibinfo{booktitle}{\emph{VLDB}}.
  \bibinfo{pages}{525--535}.
\newblock
\urldef\tempurl%
\url{https://dl.acm.org/doi/10.5555/645917.672154}
\showURL{%
\tempurl}


\bibitem[\protect\citeauthoryear{Khayyat, Lucia, Singh, Ouzzani, Papotti,
  Quian{\'e}-Ruiz, Tang, and Kalnis}{Khayyat et~al\mbox{.}}{2017}]%
        {khayyat2017fast}
\bibfield{author}{\bibinfo{person}{Zuhair Khayyat}, \bibinfo{person}{William
  Lucia}, \bibinfo{person}{Meghna Singh}, \bibinfo{person}{Mourad Ouzzani},
  \bibinfo{person}{Paolo Papotti}, \bibinfo{person}{Jorge-Arnulfo
  Quian{\'e}-Ruiz}, \bibinfo{person}{Nan Tang}, {and} \bibinfo{person}{Panos
  Kalnis}.} \bibinfo{year}{2017}\natexlab{}.
\newblock \showarticletitle{Fast and scalable inequality joins}.
\newblock \bibinfo{journal}{\emph{VLDBJ}} \bibinfo{volume}{26},
  \bibinfo{number}{1} (\bibinfo{year}{2017}), \bibinfo{pages}{125--150}.
\newblock
\urldef\tempurl%
\url{https://doi.org/10.1007/s00778-016-0441-6}
\showDOI{\tempurl}


\bibitem[\protect\citeauthoryear{Kitsuregawa and Ogawa}{Kitsuregawa and
  Ogawa}{1990}]%
        {kitsuregawa1990bucket}
\bibfield{author}{\bibinfo{person}{Masaru Kitsuregawa} {and}
  \bibinfo{person}{Yasushi Ogawa}.} \bibinfo{year}{1990}\natexlab{}.
\newblock \showarticletitle{Bucket Spreading Parallel Hash: A New, Robust,
  Parallel Hash Join Method for Data Skew in the Super Database Computer
  ({SDC})}. In \bibinfo{booktitle}{\emph{VLDB}}. \bibinfo{pages}{210--221}.
\newblock
\urldef\tempurl%
\url{https://dl.acm.org/doi/10.5555/94362.94416}
\showURL{%
\tempurl}


\bibitem[\protect\citeauthoryear{Kitsuregawa, Tanaka, and Moto-Oka}{Kitsuregawa
  et~al\mbox{.}}{1983}]%
        {kitsuregawa1983application}
\bibfield{author}{\bibinfo{person}{Masaru Kitsuregawa},
  \bibinfo{person}{Hidehiko Tanaka}, {and} \bibinfo{person}{Tohru Moto-Oka}.}
  \bibinfo{year}{1983}\natexlab{}.
\newblock \showarticletitle{Application of hash to data base machine and its
  architecture}.
\newblock \bibinfo{journal}{\emph{New Generation Computing}}
  \bibinfo{volume}{1}, \bibinfo{number}{1} (\bibinfo{year}{1983}),
  \bibinfo{pages}{63--74}.
\newblock
\urldef\tempurl%
\url{https://doi.org/10.1007/BF03037022}
\showDOI{\tempurl}


\bibitem[\protect\citeauthoryear{Koumarelas, Naskos, and Gounaris}{Koumarelas
  et~al\mbox{.}}{2018}]%
        {Koumarelas2018:selectiveThetaJoins}
\bibfield{author}{\bibinfo{person}{Ioannis Koumarelas},
  \bibinfo{person}{Athanasios Naskos}, {and} \bibinfo{person}{Anastasios
  Gounaris}.} \bibinfo{year}{2018}\natexlab{}.
\newblock \showarticletitle{Flexible Partitioning for Selective Binary
  Theta-joins in a Massively Parallel Setting}.
\newblock \bibinfo{journal}{\emph{Distrib. Parallel Databases}}
  \bibinfo{volume}{36}, \bibinfo{number}{2} (\bibinfo{year}{2018}),
  \bibinfo{pages}{301--337}.
\newblock
\urldef\tempurl%
\url{https://doi.org/10.1007/s10619-017-7214-0}
\showDOI{\tempurl}


\bibitem[\protect\citeauthoryear{Krishnamurthy}{Krishnamurthy}{2019}]%
        {OracleSQLlanguage19c}
\bibfield{author}{\bibinfo{person}{Usha Krishnamurthy}.}
  \bibinfo{year}{2019}\natexlab{}.
\newblock \bibinfo{booktitle}{\emph{Oracle Database SQL Language Reference
  19c}}.
\newblock Oracle.
\newblock
\urldef\tempurl%
\url{https://docs.oracle.com/en/database/oracle/oracle-database/19/sqlrf/index.html}
\showURL{%
\tempurl}
\newblock
\shownote{Version 19.1.}


\bibitem[\protect\citeauthoryear{Li, Mi, Riedewald, Sun, and Yao}{Li
  et~al\mbox{.}}{2019}]%
        {Li2018DAPD}
\bibfield{author}{\bibinfo{person}{Rundong Li}, \bibinfo{person}{Ningfang Mi},
  \bibinfo{person}{Mirek Riedewald}, \bibinfo{person}{Yizhou Sun}, {and}
  \bibinfo{person}{Yi Yao}.} \bibinfo{year}{2019}\natexlab{}.
\newblock \showarticletitle{Abstract cost models for distributed data-intensive
  computations}.
\newblock \bibinfo{journal}{\emph{Distrib.Parallel Databases}}
  \bibinfo{volume}{37}, \bibinfo{number}{3} (\bibinfo{year}{2019}),
  \bibinfo{pages}{411--439}.
\newblock
\urldef\tempurl%
\url{https://doi.org/10.1007/s10619-018-7244-2}
\showDOI{\tempurl}


\bibitem[\protect\citeauthoryear{Li, Riedewald, and Deng}{Li
  et~al\mbox{.}}{2018}]%
        {li2018submodularity}
\bibfield{author}{\bibinfo{person}{Rundong Li}, \bibinfo{person}{Mirek
  Riedewald}, {and} \bibinfo{person}{Xinyan Deng}.}
  \bibinfo{year}{2018}\natexlab{}.
\newblock \showarticletitle{Submodularity of Distributed Join Computation}. In
  \bibinfo{booktitle}{\emph{{SIGMOD}}}. \bibinfo{pages}{1237--1252}.
\newblock
\urldef\tempurl%
\url{https://doi.org/10.1145/3183713.3183728}
\showDOI{\tempurl}


\bibitem[\protect\citeauthoryear{Lu and Tan}{Lu and Tan}{1994}]%
        {lu1994load}
\bibfield{author}{\bibinfo{person}{HJ Lu} {and} \bibinfo{person}{Kian-Lee
  Tan}.} \bibinfo{year}{1994}\natexlab{}.
\newblock \showarticletitle{Load-balanced join processing in shared-nothing
  systems}.
\newblock \bibinfo{journal}{\emph{J. Parallel and Distrib. Comput.}}
  \bibinfo{volume}{23}, \bibinfo{number}{3} (\bibinfo{year}{1994}),
  \bibinfo{pages}{382--398}.
\newblock
\urldef\tempurl%
\url{https://doi.org/10.1006/jpdc.1994.1148}
\showDOI{\tempurl}


\bibitem[\protect\citeauthoryear{Munson, Webb, Sheldon, Fink, Hochachka, Iliff,
  Riedewald, Sorokina, Sullivan, Wood, and Kelling}{Munson
  et~al\mbox{.}}{2014}]%
        {munson2014ebird}
\bibfield{author}{\bibinfo{person}{Arthur~M Munson}, \bibinfo{person}{Kevin
  Webb}, \bibinfo{person}{Daniel Sheldon}, \bibinfo{person}{Daniel Fink},
  \bibinfo{person}{Wesley~M Hochachka}, \bibinfo{person}{Marshall Iliff},
  \bibinfo{person}{Mirek Riedewald}, \bibinfo{person}{Daria Sorokina},
  \bibinfo{person}{Brian Sullivan}, \bibinfo{person}{Christopher Wood}, {and}
  \bibinfo{person}{Steve Kelling}.} \bibinfo{year}{2014}\natexlab{}.
\newblock \showarticletitle{The ebird reference dataset, version 2014}.
\newblock \bibinfo{journal}{\emph{Cornell Lab of Ornithology and National
  Audubon Society, Ithaca, NY}} (\bibinfo{year}{2014}).
\newblock
\urldef\tempurl%
\url{https://ebird.org}
\showURL{%
\tempurl}


\bibitem[\protect\citeauthoryear{Okcan and Riedewald}{Okcan and
  Riedewald}{2011}]%
        {Okcan:theta-join}
\bibfield{author}{\bibinfo{person}{Alper Okcan} {and} \bibinfo{person}{Mirek
  Riedewald}.} \bibinfo{year}{2011}\natexlab{}.
\newblock \showarticletitle{Processing Theta-joins Using MapReduce}. In
  \bibinfo{booktitle}{\emph{SIGMOD}}. \bibinfo{pages}{949--960}.
\newblock
\urldef\tempurl%
\url{https://doi.org/10.1145/1989323.1989423}
\showDOI{\tempurl}


\bibitem[\protect\citeauthoryear{Polychroniou, Sen, and Ross}{Polychroniou
  et~al\mbox{.}}{2014}]%
        {polychroniou2014track}
\bibfield{author}{\bibinfo{person}{Orestis Polychroniou},
  \bibinfo{person}{Rajkumar Sen}, {and} \bibinfo{person}{Kenneth~A Ross}.}
  \bibinfo{year}{2014}\natexlab{}.
\newblock \showarticletitle{Track join: distributed joins with minimal network
  traffic}. In \bibinfo{booktitle}{\emph{SIGMOD}}. \bibinfo{pages}{1483--1494}.
\newblock
\urldef\tempurl%
\url{https://doi.org/10.1145/2588555.2610521}
\showDOI{\tempurl}


\bibitem[\protect\citeauthoryear{Poosala and Ioannidis}{Poosala and
  Ioannidis}{1996}]%
        {poosala1996estimation}
\bibfield{author}{\bibinfo{person}{Viswanath Poosala} {and}
  \bibinfo{person}{Yannis~E Ioannidis}.} \bibinfo{year}{1996}\natexlab{}.
\newblock \showarticletitle{Estimation of query-result distribution and its
  application in parallel-join load balancing}. In
  \bibinfo{booktitle}{\emph{VLDB}}. \bibinfo{pages}{448--459}.
\newblock
\urldef\tempurl%
\url{https://dl.acm.org/doi/10.5555/645922.673321}
\showURL{%
\tempurl}


\bibitem[\protect\citeauthoryear{Poosala and Ioannidis}{Poosala and
  Ioannidis}{1997}]%
        {poosala:1997:selectivity}
\bibfield{author}{\bibinfo{person}{Viswanath Poosala} {and}
  \bibinfo{person}{Yannis~E Ioannidis}.} \bibinfo{year}{1997}\natexlab{}.
\newblock \showarticletitle{Selectivity estimation without the attribute value
  independence assumption}. In \bibinfo{booktitle}{\emph{VLDB}}.
  \bibinfo{pages}{486--495}.
\newblock
\urldef\tempurl%
\url{https://dl.acm.org/doi/10.5555/645923.673638}
\showURL{%
\tempurl}


\bibitem[\protect\citeauthoryear{Poosala, Ioannidis, Haas, and Shekita}{Poosala
  et~al\mbox{.}}{1996}]%
        {poosalaIHS96improvedHistograms}
\bibfield{author}{\bibinfo{person}{Viswanath Poosala},
  \bibinfo{person}{Yannis~E. Ioannidis}, \bibinfo{person}{Peter~J. Haas}, {and}
  \bibinfo{person}{Eugene~J. Shekita}.} \bibinfo{year}{1996}\natexlab{}.
\newblock \showarticletitle{Improved Histograms for Selectivity Estimation of
  Range Predicates}. In \bibinfo{booktitle}{\emph{SIGMOD}}.
  \bibinfo{pages}{294--305}.
\newblock
\urldef\tempurl%
\url{https://doi.org/10.1145/233269.233342}
\showDOI{\tempurl}


\bibitem[\protect\citeauthoryear{R{\"o}diger, Idicula, Kemper, and
  Neumann}{R{\"o}diger et~al\mbox{.}}{2016}]%
        {rodiger2016flow}
\bibfield{author}{\bibinfo{person}{Wolf R{\"o}diger}, \bibinfo{person}{Sam
  Idicula}, \bibinfo{person}{Alfons Kemper}, {and} \bibinfo{person}{Thomas
  Neumann}.} \bibinfo{year}{2016}\natexlab{}.
\newblock \showarticletitle{Flow-join: Adaptive skew handling for distributed
  joins over high-speed networks}. In \bibinfo{booktitle}{\emph{ICDE}}.
  \bibinfo{pages}{1194--1205}.
\newblock
\urldef\tempurl%
\url{https://doi.org/10.1109/ICDE.2016.7498324}
\showDOI{\tempurl}


\bibitem[\protect\citeauthoryear{Sarma, He, and Chaudhuri}{Sarma
  et~al\mbox{.}}{2014}]%
        {SarmaHC14:clusterJoin}
\bibfield{author}{\bibinfo{person}{Akash~Das Sarma}, \bibinfo{person}{Yeye He},
  {and} \bibinfo{person}{Surajit Chaudhuri}.} \bibinfo{year}{2014}\natexlab{}.
\newblock \showarticletitle{ClusterJoin: {A} Similarity Joins Framework using
  Map-Reduce}.
\newblock \bibinfo{journal}{\emph{{PVLDB}}} \bibinfo{volume}{7},
  \bibinfo{number}{12} (\bibinfo{year}{2014}), \bibinfo{pages}{1059--1070}.
\newblock
\urldef\tempurl%
\url{https://doi.org/10.14778/2732977.2732981}
\showDOI{\tempurl}


\bibitem[\protect\citeauthoryear{Schneider and DeWitt}{Schneider and
  DeWitt}{1989}]%
        {schneider1989performance}
\bibfield{author}{\bibinfo{person}{Donovan~A Schneider} {and}
  \bibinfo{person}{David~J DeWitt}.} \bibinfo{year}{1989}\natexlab{}.
\newblock \showarticletitle{A performance evaluation of four parallel join
  algorithms in a shared-nothing multiprocessor environment}. In
  \bibinfo{booktitle}{\emph{SIGMOD}}. \bibinfo{pages}{110--121}.
\newblock
\urldef\tempurl%
\url{https://doi.org/10.1145/67544.66937}
\showDOI{\tempurl}


\bibitem[\protect\citeauthoryear{Schneider and DeWitt}{Schneider and
  DeWitt}{1990}]%
        {schneider1990tradeoffs}
\bibfield{author}{\bibinfo{person}{Donovan~A Schneider} {and}
  \bibinfo{person}{David~J DeWitt}.} \bibinfo{year}{1990}\natexlab{}.
\newblock \showarticletitle{Tradeoffs in processing complex join queries via
  hashing in multiprocessor database machines}. In
  \bibinfo{booktitle}{\emph{VLDB}}. \bibinfo{pages}{469--480}.
\newblock
\urldef\tempurl%
\url{https://dl.acm.org/doi/10.5555/645916.672141}
\showURL{%
\tempurl}


\bibitem[\protect\citeauthoryear{Soloviev}{Soloviev}{1993}]%
        {soloviev1993truncating}
\bibfield{author}{\bibinfo{person}{Valery Soloviev}.}
  \bibinfo{year}{1993}\natexlab{}.
\newblock \showarticletitle{A truncating hash algorithm for processing
  band-join queries}. In \bibinfo{booktitle}{\emph{{ICDE}}}.
  \bibinfo{pages}{419--427}.
\newblock
\urldef\tempurl%
\url{https://doi.org/10.1109/ICDE.1993.344039}
\showDOI{\tempurl}


\bibitem[\protect\citeauthoryear{Vitorovic, Elseidy, and Koch}{Vitorovic
  et~al\mbox{.}}{2016}]%
        {vitorovic2016load}
\bibfield{author}{\bibinfo{person}{Aleksandar Vitorovic},
  \bibinfo{person}{Mohammed Elseidy}, {and} \bibinfo{person}{Christoph Koch}.}
  \bibinfo{year}{2016}\natexlab{}.
\newblock \showarticletitle{Load balancing and skew resilience for parallel
  joins}. In \bibinfo{booktitle}{\emph{ICDE}}. \bibinfo{pages}{313--324}.
\newblock
\urldef\tempurl%
\url{https://doi.org/10.1109/ICDE.2016.7498250}
\showDOI{\tempurl}


\bibitem[\protect\citeauthoryear{Walton, Dale, and Jenevein}{Walton
  et~al\mbox{.}}{1991}]%
        {walton1991taxonomy}
\bibfield{author}{\bibinfo{person}{Christopher~B Walton},
  \bibinfo{person}{Alfred~G Dale}, {and} \bibinfo{person}{Roy~M Jenevein}.}
  \bibinfo{year}{1991}\natexlab{}.
\newblock \showarticletitle{A Taxonomy and Performance Model of Data Skew
  Effects in Parallel Joins.}. In \bibinfo{booktitle}{\emph{VLDB}}.
  \bibinfo{pages}{537--548}.
\newblock
\urldef\tempurl%
\url{https://dl.acm.org/doi/10.5555/645917.672307}
\showURL{%
\tempurl}


\bibitem[\protect\citeauthoryear{Wolf, Dias, Yu, and Turek}{Wolf
  et~al\mbox{.}}{1994}]%
        {wolf1994new}
\bibfield{author}{\bibinfo{person}{Joel~L. Wolf}, \bibinfo{person}{Daniel~M.
  Dias}, \bibinfo{person}{Philip~S. Yu}, {and} \bibinfo{person}{John Turek}.}
  \bibinfo{year}{1994}\natexlab{}.
\newblock \showarticletitle{New algorithms for parallelizing relational
  database joins in the presence of data skew}.
\newblock \bibinfo{journal}{\emph{TKDE}} (\bibinfo{year}{1994}),
  \bibinfo{pages}{990--997}.
\newblock
\urldef\tempurl%
\url{https://doi.org/10.1109/69.334888}
\showDOI{\tempurl}


\bibitem[\protect\citeauthoryear{Xu, Kostamaa, Zhou, and Chen}{Xu
  et~al\mbox{.}}{2008}]%
        {xu2008handling}
\bibfield{author}{\bibinfo{person}{Yu Xu}, \bibinfo{person}{Pekka Kostamaa},
  \bibinfo{person}{Xin Zhou}, {and} \bibinfo{person}{Liang Chen}.}
  \bibinfo{year}{2008}\natexlab{}.
\newblock \showarticletitle{Handling data skew in parallel joins in
  shared-nothing systems}. In \bibinfo{booktitle}{\emph{SIGMOD}}.
  \bibinfo{pages}{1043--1052}.
\newblock
\urldef\tempurl%
\url{https://doi.org/10.1145/1376616.1376720}
\showDOI{\tempurl}


\bibitem[\protect\citeauthoryear{Zaharia, Chowdhury, Franklin, Shenker, and
  Stoica}{Zaharia et~al\mbox{.}}{2010}]%
        {zaharia2010spark}
\bibfield{author}{\bibinfo{person}{Matei Zaharia}, \bibinfo{person}{Mosharaf
  Chowdhury}, \bibinfo{person}{Michael~J Franklin}, \bibinfo{person}{Scott
  Shenker}, {and} \bibinfo{person}{Ion Stoica}.}
  \bibinfo{year}{2010}\natexlab{}.
\newblock \showarticletitle{Spark: cluster computing with working sets}. In
  \bibinfo{booktitle}{\emph{USENIX HotCloud}}.
\newblock
\urldef\tempurl%
\url{https://dl.acm.org/doi/10.5555/1863103.1863113}
\showURL{%
\tempurl}


\bibitem[\protect\citeauthoryear{Zhang, Chen, and Wang}{Zhang
  et~al\mbox{.}}{2012}]%
        {Zhang2012:multiwayThetaJoin}
\bibfield{author}{\bibinfo{person}{Xiaofei Zhang}, \bibinfo{person}{Lei Chen},
  {and} \bibinfo{person}{Min Wang}.} \bibinfo{year}{2012}\natexlab{}.
\newblock \showarticletitle{Efficient Multi-way Theta-join Processing Using
  MapReduce}.
\newblock \bibinfo{journal}{\emph{{PVLDB}}} \bibinfo{volume}{5},
  \bibinfo{number}{11} (\bibinfo{year}{2012}), \bibinfo{pages}{1184--1195}.
\newblock
\urldef\tempurl%
\url{https://doi.org/10.14778/2350229.2350238}
\showDOI{\tempurl}


\bibitem[\protect\citeauthoryear{Zhao, Rusu, Dong, and Wu}{Zhao
  et~al\mbox{.}}{2016}]%
        {zhaoRDW16:arraySimJoin}
\bibfield{author}{\bibinfo{person}{Weijie Zhao}, \bibinfo{person}{Florin Rusu},
  \bibinfo{person}{Bin Dong}, {and} \bibinfo{person}{Kesheng Wu}.}
  \bibinfo{year}{2016}\natexlab{}.
\newblock \showarticletitle{Similarity Join over Array Data}. In
  \bibinfo{booktitle}{\emph{{SIGMOD}}}. \bibinfo{pages}{2007--2022}.
\newblock
\urldef\tempurl%
\url{https://doi.org/10.1145/2882903.2915247}
\showDOI{\tempurl}


\end{thebibliography}
\clearpage
\begin{appendix}

\section{Additional Empirical Results}

We present additional measurements and insights from our empirical evaluation. Results from the main
paper may be repeated here for convenience so that the reader can find relevant numbers in one place.
Properties of the datasets used are summarized in \Cref{tab:queryBig}.

\begin{table}[tb]
\centering
\caption{Band-join characteristics used in the experiments. Input and output size are reported in [million tuples].}
\label{tab:queryBig}
\scalebox{0.8}
{
	\begin{tabular}{|l|r|r|r|r|}
	\hline
	Data set                         & $d$ & Band width & Input Size & Output Size  \\
	\hline
	\texttt{pareto-1.5}              & 1 & $0$              & 400  & 2430  \\ 
	\texttt{pareto-1.5}              & 1 & $10^{-5}$        & 400  & 4580  \\
	\texttt{pareto-1.5}              & 1 & $2\cdot 10^{-5}$ & 400  & 9120  \\
	\texttt{pareto-1.5}              & 1 & $3\cdot 10^{-5}$ & 400  & 11280 \\
	\texttt{pareto-1.5}              & 3 & $(0,0,0)$        & 400  & 0     \\
	\texttt{pareto-1.5}              & 3 & $(2,2,2)$        & 400  & 1120  \\
	\texttt{pareto-1.5}              & 3 & $(4,4,4)$        & 400  & 8740  \\
	\hline	
	\texttt{pareto-0.5}              & 3 & $(2,2,2)$        & 400  & 12    \\
	\texttt{pareto-1.0}              & 3 & $(2,2,2)$        & 400  & 420   \\
	\texttt{pareto-2.0}              & 3 & $(2,2,2)$        & 400  & 3200  \\
	\hline
	\texttt{pareto-1.5}              & 8 & $(20,\ldots,20)$     & 100  & 9  \\
	\texttt{pareto-1.5}              & 8 & $(20,\ldots,20)$     & 200  & 57  \\
	\texttt{pareto-1.5}              & 8 & $(20,\ldots,20)$     & 400  & 219  \\
	\texttt{pareto-1.5}              & 8 & $(20,\ldots,20)$     & 800  & 857  \\
	\hline
	\texttt{rv-pareto-1.5}      & 1 & $2$        & 400    & 0  \\
	\texttt{rv-pareto-1.5}      & 1 & $1000$     & 400 	 & 0  \\
	\texttt{rv-pareto-1.5}      & 3 & $(1000,1000,1000)$     & 400 & 0  \\
	\texttt{rv-pareto-1.5}      & 3 & $(2000,2000,2000)$     & 400 & 0  \\
	\hline
	\texttt{ebird} and \texttt{cloud} & 3 & ($0,0,0$) & 890     & 0     \\
	\texttt{ebird} and \texttt{cloud} & 3 & ($1,1,1$) & 890        & 320  \\
	\texttt{ebird} and \texttt{cloud} & 3 & ($1,1,5$) & 890        & 1164 \\
	\texttt{ebird} and \texttt{cloud} & 3 & ($2,2,2$) & 890        & 2134 \\
	\texttt{ebird} and \texttt{cloud} & 3 & ($4,4,4$) & 890        & 16998 \\
    \hline
    \texttt{ptf\_objects} & 2 & ($2.78\cdot 10^{-4},2.78\cdot 10^{-4}$) & 1198 & 876 \\
    \texttt{ptf\_objects} & 2 & ($8.33\cdot 10^{-4},8.33\cdot 10^{-4}$) & 1198 & 1125 \\
    \hline
	\end{tabular}
}
\end{table}

\subsection{Comparing to Distributed IEJoin}
\label{sec:exp-iejoin}

\resultbox{\Cref{tab:comp-iejoin} shows that \methodS has faster join time than the distributed version of
IEJoin~\cite{khayyat2017fast} for band joins.}

IEJoin~\cite{khayyat2017fast} uses a carefully designed data structure to speed up in-memory computation
of joins with inequality predicates. Its distributed version sorts the input datasets on one of the join attributes and
range-partitions them into default-sized blocks based on approximate quantiles, before pairs of joinable blocks
are assigned to the $w$ workers for local
joining. We find that the block size, i.e., the granularity of the quantiles used to define the partitions in each
input dataset, is a key meta-parameter.
\Cref{tab:comp-iejoin} shows how much it affects the efficiency of distributed band-joins.
For each query, we search for the best block size and list a few smaller and larger ones.
With a good block size, distributed IEJoin can have join time close to that of \methodS
in certain cases. However, in most cases, it creates significantly higher input duplication than
\methodS, thus having a longer join time. The higher duplication is caused by (1) partition boundaries that cut
through dense input regions and (2) the direct use of quantiles for defining partitions, instead of applying
a covering algorithm to identify just $w$ partitions
(as \methodVit~\cite{vitorovic2016load} and M-Bucket-I~\cite{Okcan:theta-join} do).
The main benefit of the covering algorithm is that it can often avoid duplication of input blocks that
belong to multiple joinable block pairs by assigning those pairs to the same region in the cover.

\begin{table*}[tb]
\setlength{\tabcolsep}{1mm}
\centering
\caption{\Cref{sec:exp-iejoin}: Comparing \method with distributed IEJoin on 30 workers. We use \texttt{pareto} data sets with 400 million input. The best block size (with the shortest join time) for each query is marked in bold. (Join result size, $I,I_m$ and $O_m$ in [millions])}
\label{tab:comp-iejoin}
\scalebox{0.8}{
	\begin{tabular}{|r|r|r||r|rrr||r|rrr|r|}
	\hline
	\multirow{2}{*}{Data Sets} & \multirow{2}{*}{Band Width} & \multirow{2}{*}{\begin{tabular}[c]{@{}r@{}}Join Result Size\end{tabular}} & \multicolumn{4}{c||}{\methodS} & \multicolumn{5}{c|}{IEJoin} \\
	\cline{4-12}
	 &  &  & Join Time & $I$ & $I_m$ & $O_m$ & Join Time & $I$ & $I_m$ & $O_m$ & sizePerBlock \\
	\hline
	\multirow{5}{*}{Pareto-1.5} & \multirow{5}{*}{{[}0,0,0{]}} & \multirow{5}{*}{0} & \multirow{5}{*}{\colorbox{Blue1}{229}} & \multirow{5}{*}{401} & \multirow{5}{*}{14} & \multirow{5}{*}{0} & 523 & 780 & 40 & 0 & 10000 \\
	 &  &  &  &  &  &  & 551 & 770 & 44 & 0 & 11000 \\
	 &  &  &  &  &  &  & 391 & 726 & 25 & 0 & \textbf{12524} \\
	 &  &  &  &  &  &  & 406 & 754 & 26 & 0 & 13000 \\
	 &  &  &  &  &  &  & 422 & 756 & 28 & 0 & 14000 \\\hline
	\multirow{5}{*}{Pareto-1.5} & \multirow{5}{*}{{[}2,2,2{]}} & \multirow{5}{*}{1120} & \multirow{5}{*}{\colorbox{Blue1}{342}} & \multirow{5}{*}{404} & \multirow{5}{*}{15} & \multirow{5}{*}{29} & 696 & 1092 & 48 & 14 & 6000 \\
	 &  &  &  &  &  &  & 707 & 1092 & 42 & 44 & 7000 \\
	 &  &  &  &  &  &  & \colorbox{Red1}{677} & 1070 & 45 & 21 & \textbf{7422} \\
	 &  &  &  &  &  &  & 695 & 1088 & 48 & 14 & 8000 \\
	 &  &  &  &  &  &  & 733 & 1062 & 36 & 85 & 9000 \\\hline
	\multirow{5}{*}{Pareto-1.5} & \multirow{5}{*}{{[}4,4,4{]}} & \multirow{5}{*}{8740} & \multirow{5}{*}{\colorbox{Blue1}{858}} & \multirow{5}{*}{413} & \multirow{5}{*}{14} & \multirow{5}{*}{29} & 968 & 1088 & 64 & 900 & 16000 \\
	 &  &  &  &  &  &  & 922 & 1054 & 68 & 55 & 17000 \\
	 &  &  &  &  &  &  & 899 & 1116 & 72 & 19 & \textbf{18000} \\
	 &  &  &  &  &  &  & 3512 & 1064 & 38 & 1500 & 19000 \\
	 &  &  &  &  &  &  & 3753 & 1120 & 40 & 1610 & 20000 \\\hline
	\multirow{5}{*}{Pareto-1.0} & \multirow{5}{*}{{[}2,2,2{]}} & \multirow{5}{*}{420} & \multirow{5}{*}{\colorbox{Blue1}{287}} & \multirow{5}{*}{401} & \multirow{5}{*}{13} & \multirow{5}{*}{17} & 670 & 1176 & 40 & 21 & 4000 \\
	 &  &  &  &  &  &  & 650 & 1120 & 40 & 18 & 5000 \\
	 &  &  &  &  &  &  & \colorbox{Red1}{635} & 1080 & 37 & 26 & \textbf{6263} \\
	 &  &  &  &  &  &  & 661 & 1106 & 42 & 18 & 7000 \\
	 &  &  &  &  &  &  & 676 & 1088 & 48 & 4 & 8000 \\\hline
	\multirow{5}{*}{Pareto-0.5} & \multirow{5}{*}{{[}2,2,2{]}} & \multirow{5}{*}{12} & \multirow{5}{*}{\colorbox{Blue1}{227}} & \multirow{5}{*}{401} & \multirow{5}{*}{13} & \multirow{5}{*}{0.3} & 413 & 828 & 24 & 1 & 6000 \\
	 &  &  &  &  &  &  & 436 & 798 & 28 & 1 & 7000 \\
	 &  &  &  &  &  &  & 348 & 796 & 17 & 2 & \textbf{8295} \\
	 &  &  &  &  &  &  & 363 & 810 & 18 & 2 & 9000 \\
	 &  &  &  &  &  &  & 382 & 820 & 20 & 2 & 10000 \\\hline
	\end{tabular}}
\end{table*}

\subsection{Accuracy of the Running-Time Model}
\label{sec:exp-model}

\resultbox{\Cref{tab:modelAccuracy} and \Cref{fig:model-cumulative-dist} show that our running-time model has a less than $20\%$ relative error in over $70\%$ of the cases. \Cref{tab:varyingWeightsBig} sheds more light on how this model helps \method to intelligently choose which overhead to tune to improve join efficiency.}

\Cref{tab:modelAccuracy} lists the predicted and actual join time. 
Our running-time model predicts correctly which algorithm has shorter join time, and in most cases has a less than $20\%$ relative error compared to the ground truth. 
\Cref{fig:model-cumulative-dist} summarizes the prediction errors in a cumulative distribution. 

A major benefit of using a running-time model is to help \method determine which part of the distributed join computation has more impact on join time. Specifically, if the running-time model assigns a greater weight factor
to total input size $I$ than to the local overhead ($4\cdot I_m + O_m$), then \method will try to create less
duplication, perhaps at the expense of a worse load balance.
We conduct a series of experiments exploring how the weight factors in the running-time model affect \method.
As shown in \Cref{tab:varyingWeightsBig}, $\beta_1$ is fixed, and as $\beta_2$ increases, the local overhead has a
greater impact on the overall join time, so it is beneficial to reduce it. This is what \method does: $I$ increases and
 $4\cdot I_m+O_m$ decreases as $\beta_2$ gets bigger, indicating that the algorithm more aggressively
partitions the data to reduce max worker load. The local join cost is proportional to the load $4\cdot I_m + O_m$. The factor 4 here is obtained for our cluster using the method mentioned in \cite{vitorovic2016load}.

\renewcommand\dbltopfraction{1} %

\begin{table*}[tb]
\setlength{\tabcolsep}{1mm}
\centering
\caption{\Cref{sec:exp-model}: Running-Time Model Accuracy. Unless stated otherwise, the join was executed on the 30-worker cluster. Dataset $X/Y/w$ refers to input size of $X$ million, output size of $Y$ million and a cluster with $w$ workers.}
\label{tab:modelAccuracy}
\scalebox{0.8}{
	\begin{tabular}{|rr|rrr|rrr|rrr|rrr|}\hline
	\multirow{2}{*}{Data Sets} & \multirow{2}{*}{Band Width} & \multicolumn{3}{c|}{\methodS} & \multicolumn{3}{c|}{\methodVit} & \multicolumn{3}{c|}{\methodOneB} & \multicolumn{3}{c|}{\methodGrid} \\\cline{3-14}
	 &  & Predicted & Actual & Error & Predicted & Actual & Error & Predicted & Actual & Error & Predicted & Actual & Error \\
	\hline
	pareto-1.5, d=1 & 0 & 375.35 & 348 & 7.86\% & 491.55 & 483 & 1.77\% & 1323.69 & 762 & 73.71\% & - & - & - \\
	pareto-1.5, d=1 & 1.00E-05 & 507.36 & 532 & -4.63\% & 711.22 & 655 & 8.58\% & 1464.09 & 1004 & 45.83\% & 722.61 & 540 & 33.82\% \\
	pareto-1.5, d=1 & 2.00E-05 & 807.54 & 810 & -0.30\% & 969.48 & 962 & 0.78\% & 1758.54 & 1316 & 33.63\% & 1017.06 & 834 & 21.95\% \\
	pareto-1.5, d=1 & 3.00E-05 & 948.33 & 875 & 8.38\% & 1184.12 & 1140 & 3.87\% & 1898.94 & 1520 & 24.93\% & 1157.46 & 956 & 21.07\% \\ \hline
	pareto-1.5, d=3 & (0,0,0) & 215.13 & 229 & -6.06\% & 265 & 320 & -17.19\% & 1165.74 & 792 & 47.19\% & - & - & - \\
	pareto-1.5, d=3 & (2,2,2) & 286.53 & 342 & -16.22\% & 462.5 & 645 & -28.29\% & 1238.48 & 1149 & 7.79\% & 1379.65 & 1412 & -2.29\% \\
	pareto-1.5, d=3 & (4,4,4) & 786.99 & 858 & -8.28\% & 1090.11 & 1212 & -10.06\% & 1733.19 & 1772 & -2.19\% & 2211.65 & 1816 & 21.79\% \\ \hline
	ebird \& cloud & (0,0,0) & 331.96 & 245 & 35.49\% & 354.46 & 308 & 15.08\% & 1791.93 & 1418 & 26.37\% & - & - & - \\
	ebird \& cloud & (1,1,1) & 364.77 & 329 & 10.87\% & 787.67 & 977 & -19.38\% & 1806.54 & 1532 & 17.92\% & 1456.17 & 1419 & 2.62\% \\
	ebird \& cloud & (1,1,5) & 459.77 & 432 & 6.43\% & 902.82 & 985 & -8.34\% & 1844.18 & 1597 & 15.48\% & 1577.02 & 1505 & 4.79\% \\
	ebird \& cloud & (2,2,2) & 468.43 & 420 & 11.53\% & 1030.42 & 1212 & -14.98\% & 1893.32 & 1573 & 20.36\% & 1717 & 1377 & 24.69\% \\ \hline
	pareto-$0.5$ & (2,2,2) & 210.255 & 227 & -7.38\% & 313.74 & 346 & -9.32\% & 1164.18 & 1137 & 2.39\% & 1260.7 & 1146 & 10.01\% \\
	pareto-$1.0$ & (2,2,2) & 242.82 & 287 & -15.39\% & 382.77 & 539 & -28.99\% & 1190.7 & 1235 & -3.59\% & 1304.9 & 1335 & -2.25\% \\
	pareto-$1.5$ & (2,2,2) & 282.63 & 342 & -17.36\% & 461.91 & 645 & -28.39\% & 1235.55 & 1149 & 7.53\% & 1379.65 & 1412 & -2.29\% \\
	pareto-$2.0$ & (2,2,2) & 435.27 & 483 & -9.88\% & 677.49 & 811 & -16.46\% & 1372.05 & 1369 & 0.22\% & 2387.15 & 2417 & -1.24\% \\ \hline
    pareto-$1.5$ 200/282/15 & (2,2,2) & 249.48 & 305 & -18.20\% & 375 & 460 & -18.48\% & 882.45 & 779 & 13.28\% & 1327.65 & 1381 & -3.86\% \\
	pareto-$1.5$ 400/1120/30 & (2,2,2) & 282.63 & 342 & -17.36\% & 461.91 & 645 & -28.39\% & 1235.55 & 1149 & 7.53\% & 1379.65 & 1412 & -2.29\% \\
	pareto-$1.5$ 800/4460/60 & (2,2,2) & 469.98 & 434 & 8.29\% & 928.2 & 920 & 0.89\% & 2706.9 & 1731 & 56.38\% & - & - & - \\ \hline
	ebird \& cloud 222/134/15 & (2,2,2) & 198.66 & 204 & -2.62\% & 275.94 & 271 & 1.82\% & 693.27 & 547 & 26.74\% & 1216.96 & 812 & 49.87\% \\
	ebird \& cloud 445/530/30 & (2,2,2) & 273.06 & 190 & 43.72\% & 457.41 & 331 & 38.19\% & 855.9 & 688 & 24.40\% & 778.77 & 771 & 1.01\% \\
	ebird \& cloud 890/2000/60 & (2,2,2) & 235.746 & 213 & 10.68\% & 474.81 & 440 & 7.91\% & 1330.38 & 1117 & 19.10\% & 902.16 & 793 & 13.77\% \\ \hline
	\end{tabular}}
\end{table*}

\begin{table*}[tb]
\setlength{\tabcolsep}{1mm}
\caption{\Cref{sec:exp-model}: Varying $\beta_2$ in the running-time model ``$\beta_1 \cdot I + \beta_2 \cdot (4 \cdot I_m + O_m)$'', where $\beta_1$ is fixed to be 1.0. The relative join time is calculated by dividing each join time by \method's join time in the same row.}
\label{tab:varyingWeightsBig}
	\vspace{5px}
	\begin{subtable}{\linewidth}
	\centering
	\caption{\texttt{Pareto-$1.5$}, band width is $(2,2,2)$, on 30 workers.}
	\scalebox{0.8}{
	\begin{tabular}{|r||r|r|r|r|r|r|r|r||rrr|}\hline
	\multirow{2}{*}{$\beta_2$} & \multicolumn{2}{c|}{\method} & \multicolumn{2}{c|}{\methodVit} & \multicolumn{2}{c|}{\methodOneB} & \multicolumn{2}{c||}{\methodGrid} & \multicolumn{3}{c|}{Relative Join Time vs \method} \\
	\cline{2-12}
	 & $I$ & 4$\cdot I_m + O_m$ & $I$ & 4$\cdot I_m + O_m$ & $I$ & 4$\cdot I_m + O_m$ & I & 4$\cdot I_m + O_m$ & \methodVit & \methodOneB & \methodGrid \\
	\hline
	0.0001 & 401 & 2510 & \multirow{9}{*}{652} & \multirow{9}{*}{145} & \multirow{9}{*}{2200} & \multirow{9}{*}{329} & \multirow{9}{*}{\colorbox{Red1}{5600}} & \multirow{9}{*}{781} & 1.63 & 5.49 & \colorbox{Red1}{13.97} \\
	0.001 & 401 & 433 &  &  &  &  &  &  & 1.62 & 5.47 & \colorbox{Red1}{13.93} \\
	0.01 & 402 & 95 &  &  &  &  &  &  & 1.62 & 5.47 & \colorbox{Red1}{13.92} \\
	0.1 & 402 & 95 &  &  &  &  &  &  & 1.62 & 5.42 & \colorbox{Red1}{13.78} \\
	1 & 403 & 91 &  &  &  &  &  &  & 1.61 & 5.12 & \colorbox{Red1}{12.92} \\
	10 & 404 & 91 &  &  &  &  &  &  & 1.60 & 4.18 & \colorbox{Red1}{10.21} \\
	100 & 404 & 91 &  &  &  &  &  &  & 1.60 & 3.70 & \colorbox{Red1}{8.81} \\
	1000 & 404 & 91 &  &  &  &  &  &  & 1.59 & 3.63 & \colorbox{Red1}{8.61} \\
	10000 & 404 & 91 &  &  &  &  &  &  & 1.59 & 3.62 & \colorbox{Red1}{8.59} \\\hline
	\end{tabular}}
	\end{subtable}
	\vspace{5px}
    \newline
	\begin{subtable}{\linewidth}
	\centering
	\caption{\texttt{ebird} joins \texttt{cloud}, band width is $(2,2,2)$, on 30 workers.}
	\scalebox{0.8}{
	\begin{tabular}{|r||r|r|r|r|r|r|r|r||rrr|}\hline
	\multirow{2}{*}{$\beta_2$} & \multicolumn{2}{c|}{\method} & \multicolumn{2}{c|}{\methodVit} & \multicolumn{2}{c|}{\methodOneB} & \multicolumn{2}{c||}{\methodGrid} & \multicolumn{3}{c|}{Relative Join Time vs \method} \\
	\cline{2-12}
	 & $I$ & 4$\cdot I_m + O_m$ & $I$ & 4$\cdot I_m + O_m$ & $I$ & 4$\cdot I_m + O_m$ & I & 4$\cdot I_m + O_m$ & \methodVit & \methodOneB & \methodGrid \\
	\hline
	0.0001 & 890.34 & 289 & \multirow{9}{*}{1830} & \multirow{9}{*}{502} & \multirow{9}{*}{4832} & \multirow{9}{*}{711} & \multirow{9}{*}{\colorbox{Red1}{10800}} & \multirow{9}{*}{1518} & 2.06 & 5.43 & \colorbox{Red1}{12.14} \\
	0.001 & 890.36 & 223 &  &  &  &  &  &  & 2.05 & 5.42 & \colorbox{Red1}{12.12} \\
	0.01 & 890.42 & 195 &  &  &  &  &  &  & 2.06 & 5.43 & \colorbox{Red1}{12.12} \\
	0.1 & 890.52 & 191 &  &  &  &  &  &  & 2.07 & 5.39 & \colorbox{Red1}{12.03} \\
	1 & 890.8 & 189 &  &  &  &  &  &  & 2.16 & 5.13 & \colorbox{Red1}{11.41} \\
	10 & 890.8 & 189 &  &  &  &  &  &  & 2.46 & 4.29 & \colorbox{Red1}{9.33} \\
	100 & 890.8 & 189 &  &  &  &  &  &  & 2.62 & 3.83 & \colorbox{Red1}{8.20} \\
	1000 & 890.8 & 189 &  &  &  &  &  &  & 2.65 & 3.76 & \colorbox{Red1}{8.04} \\
	10000 & 890.8 & 189 &  &  &  &  &  &  & 2.65 & 3.76 & \colorbox{Red1}{8.02} \\\hline
	\end{tabular}}
	\end{subtable}
\end{table*}

\subsection{Benefits of Symmetric Partitioning}
\label{sec:exp-twoAdap}

\resultbox{\Cref{tab:dupBothOrDupTBig} shows that allowing to choose which of the inputs to duplicate across
a split boundary (which is the benefit of \method over \methodS) can significantly
reduce input duplication as well as the load on the most loaded worker machine at the same time,
resulting in lower join time. 
}

While the gap is small when high-density regions of both inputs coincide (\texttt{pareto-z}) or are sufficiently correlated (real data), the anti-correlated reverse Pareto distributions reveal the benefits of allowing more flexibility in the split decisions. In particular, \method is much better able to achieve near-perfect load balance with
low input duplication, significantly lowering join time.

\subsection{Additional results for \Cref{sec:expMultAttr}: ``Multiple Join Attributes''}
\label{sec:app_expMultAttr}

\resultbox{\Cref{tab:app_time-3d} confirms the superiority of \method as dimensionality of the join
condition increases.}

Compared to the results for 8-dimensional band conditions in \Cref{sec:expMultAttr}, this series of experiments
uses a smaller band width of 5 in each dimension, which results in significantly lower output size compared to band width 20 in each dimension. This helps \methodVit because smaller output implies a less dense join matrix and
hence lower optimization cost, but it still is significantly worse then \method.

\begin{table*}[ptb]
\setlength{\tabcolsep}{1mm}
\centering
\caption{\Cref{sec:exp-twoAdap}: \methodS vs. \method (Model-estimated join times in [sec], I/O sizes in [millions]).}
\label{tab:dupBothOrDupTBig}
\scalebox{0.8}
{
	\begin{tabular}{|r|r||rrrrr|rrrrr||c|}
	\hline
	\multirow{2}{*}{Data Sets}
        & \multirow{2}{*}{Band width}
	    & \multicolumn{5}{c|}{\methodS} 
		& \multicolumn{5}{c||}{\method} 
		& Ratio of \method join time \\
	\cline{3-12}
		& & $I$\qquad & $I_m$ & $O_m$ & \Imbal & Join time
		& $I$\quad & $I_m$ & $O_m$ & \Imbal & Join time
		& to \methodS join time \\
	\hline
    \texttt{pareto-$1.0$} & (2,2,2)
		& 401 & 13 & 17 & \colorbox{Blue1}{1.02} & 243  
        & 401 & 12 & 21 & \colorbox{Blue1}{1.02} & 242 
		& \colorbox{Blue1}{1.00} \\
    \texttt{ebird} and \texttt{cloud} & (0,0,0)
		& 890 & 30 & 0  & \colorbox{Blue1}{1.01} & 332  
        & 890 & 30 & 0  & \colorbox{Blue1}{1.01} & 332 
		& \colorbox{Blue1}{1.00} \\		
    \texttt{ebird} and \texttt{cloud} & (2,2,2)
		& 899 & 32 & 66 & \colorbox{Blue1}{1.03} & 434  
        & 891 & 31 & 67 & \colorbox{Blue1}{1.02} & 429 
		& \colorbox{Blue1}{0.99} \\		
    \texttt{ebird} and \texttt{cloud} & (4,4,4)
		& 918 & 31 & 567 & \colorbox{Blue1}{1.09} & 1220  
        & 894 & 30 & 515 & \colorbox{Blue1}{1.01} & 1038 
		& \colorbox{Blue1}{0.93} \\
	\hline		
    \texttt{reverse-pareto-$2.0$} & (1000,1000,1000)
		& 439 & 168 & 0 & \colorbox{Red1}{11.47} & 1429  
        & 425 & 19 & 0  & \colorbox{Blue1}{1.32} & 261 
		& \colorbox{Blue1}{0.18} \\		
    \texttt{reverse-pareto-$2.0$} & (2000,2000,2000)
		& 418 & 189 & 0 & \colorbox{Red1}{13.55} & 1584  
    	& 475 & 19 & 0  & \colorbox{Blue1}{1.18} & 275 
		& \colorbox{Blue1}{0.17} \\		
    \texttt{reverse-pareto-$2.0$} & 2
		& 458 & 34 & 0 & \colorbox{Red1}{2.21} & 387  
    	& 413 & 19 & 0  & \colorbox{Blue1}{1.36} & 258 
		& \colorbox{Blue1}{0.25} \\		
    \texttt{reverse-pareto-$2.0$} & 1000
		& 400 & 200 & 0 & \colorbox{Red1}{14.97} & 1666 
    	& 550 & 33 & 0  & \colorbox{Blue1}{1.82} & 409 
		& \colorbox{Blue1}{0.67} \\		
	\hline
	\end{tabular}
}
\end{table*}

\begin{table*}[tb]
\setlength{\tabcolsep}{1mm}
\caption{
\Cref{sec:expMultAttr}
and
\Cref{sec:app_expMultAttr}:
Multidimensional joins on \texttt{Pareto-$1.5$}, varying dimension from 1 to 8 with band width 5 on each dimension. (Join time is estimated using the running-time model.)}
\label{tab:app_time-3d}
	\centering
	\scalebox{0.8}
	{
        \def\arraystretch{1.5}%
	    \begin{tabular}{|r|r||r|r|r|r||rrr|rrr|rrr|rrr|}
		\hline
		\multirow{2}{*}{$d$} & \multirow{2}{*}{\begin{tabular}[r]{@{}r@{}}Join Result\\ Size [millions]\end{tabular}} &
		  \multicolumn{4}{c||}{Runtime (\parTime+join time) in [sec]} & \multicolumn{12}{c|}{I/O sizes in [millions]: $I,I_m,O_m$} \\ 
		\cline{3-18}
		 & & \method & \methodVit & \methodOneB & \methodGrid
			& \multicolumn{3}{c|}{\method} & \multicolumn{3}{c|}{\methodVit} 
			& \multicolumn{3}{c|}{\methodOneB} & \multicolumn{3}{c|}{\methodGrid}   \\
		\hline
		1 & $1.12\cdot 10^{8}$ & 
		  \def\arraystretch{1}%
          \begin{tabular}[r]{@{}r@{}} $6.77\cdot 10^6$ \\(1+$6.77\cdot 10^6$) \end{tabular} &
		  \def\arraystretch{1}%
          \begin{tabular}[r]{@{}r@{}} $9.4\cdot 10^6$ \\ (113+$9.4\cdot 10^6$) \end{tabular} &
          $7.27\cdot 10^6$ & $7.27\cdot 10^6$ &
          531 & 18 & 3470000 & 544 & 12 & 4820000 & 2200 & 73 & 3730000 & 785 & 27 & 3730000 \\\hline
		2 & $313,000$ &
		  \def\arraystretch{1}%
		  \begin{tabular}[r]{@{}r@{}} $20291$ \\(1+20290) \end{tabular} &
		  \def\arraystretch{1}%
		  \begin{tabular}[r]{@{}r@{}} $26488$ \\(113+26375) \end{tabular} &
		  21446 & 21340 & 409 & 12 & 10300 & 548 & 13 & 13400 & 2200 & 73 & 10400 & 1956 & 67 & 10400 \\\hline
		4 & $860$ & 266 (3+263) & 519 (120+399) & \colorbox{Red1}{1222} & \colorbox{Red1}{8751} & 406 & 11 & 34 & 573 & 27 & 19 & 2200 & 73 & 29 & \colorbox{Red1}{16004} & \colorbox{Red1}{547} & 29 \\\hline
		8 & 0 & 217 (3+214) & 458 (151+307) & \colorbox{Red1}{1166} & \colorbox{Red1}{694560} & 404 & 14 & 0.0 & 560 & 20 & 0.00 & 2200 & 73 & 0 & \colorbox{Red1}{1280326} & \colorbox{Red1}{43747} & 0 \\\hline
		\end{tabular}
	}
\end{table*}

\subsection{Results of \method using the theoretical termination condition}
\label{sec:app_expMinMaxOverhead}

\resultbox{\Cref{tab:theoretical_condition} shows that \method finds superior partitionings also
when using the theoretical stopping condition that looks for the minimal overhead over the lower
bounds on total input and max worker load. This does not rely on any cost model, except for an estimate
of the relative cost per input vs per output tuple for the local join cost.}

For this series of experiments, we used the Palomar Transient Factory (PTF) data set,
which records observations of stars and galaxies.
We explored join queries with an input of 1.198 billion records on two attributes: object right ascension and declination, with band width of 1 arc second ($2.78\cdot 10^{-4}$) and 3 arc seconds ($8.33\cdot 10^{-4}$), respectively.
This type of queries finds pairs of celestial objects that are close to each other, helping the astronomers
determine repeat observations of the same object. \Cref{tab:theoretical_condition} shows total
input including duplicates ($I$), as well as max input and output
size on the most loaded worker ($I_m, O_m$).
\method beats all competitors on all three measures, meaning it achieves both lower input duplication
and max worker load.

\begin{table*}[ptb]
\setlength{\tabcolsep}{1mm}
\centering
\caption{\Cref{sec:app_expMinMaxOverhead}: \method using theoretical termination condition (I/O sizes in [millions]).}
\label{tab:theoretical_condition}
\scalebox{0.8}
{
	\begin{tabular}{|r|r|r||rrr|rrr|rrr|rrr|}
	\hline
	\multirow{2}{*}{Data Sets}
        & \multirow{2}{*}{Band Width}
        & Join Result
		& \multicolumn{3}{c|}{\method} 
	    & \multicolumn{3}{c|}{\methodVit} 
	    & \multicolumn{3}{c|}{\methodOneB} 
	    & \multicolumn{3}{c|}{\methodGrid} \\
	\cline{4-15}
		& & Size [millions] & $I$\qquad & $I_m$ & $O_m$
		& $I$\quad & $I_m$ & $O_m$
		& $I$\quad & $I_m$ & $O_m$ 
		& $I$\quad & $I_m$ & $O_m$ \\
	\hline
    \texttt{ptf\_objects} & ($2.78\cdot 10^{-4}$,$2.78\cdot 10^{-4}$) & 876
		& 1198 & 39.98 & 29.08 
        & 1488 & 60.02 & 32.13
        & \colorbox{Red1}{6589} & \colorbox{Red1}{220.00} & 29.20
        & \colorbox{Red1}{5990} & \colorbox{Red1}{199.67} & 29.20 \\
    \texttt{ptf\_objects} & ($8.33\cdot 10^{-4}$,$8.33\cdot 10^{-4}$) & 1125
        & 1198 & 40.25 & 36.39
        & 1508 & 60.02 & 40.77
        & \colorbox{Red1}{6589} & \colorbox{Red1}{220.99} & 37.50
        & \colorbox{Red1}{5990} & \colorbox{Red1}{199.67} & 37.50 \\
	\hline
	\end{tabular}
}
\end{table*}

\subsection{Near-Optimality of Our Approach}

\Cref{fig:comp-overhead_appendix} is a variant of \Cref{fig:comp-overhead} 
summarizes the results from all tables including the ones from the appendix.
We again 
see how \method \emph{achieves significantly lower max worker load ($y$-axis)
with less input duplication ($x$-axis).
\method is always within 10\% of the lower bound
on both measures, beating the competition by a wide margin.}

\begin{figure}[t]
\centering
\includegraphics[width=.9\linewidth]{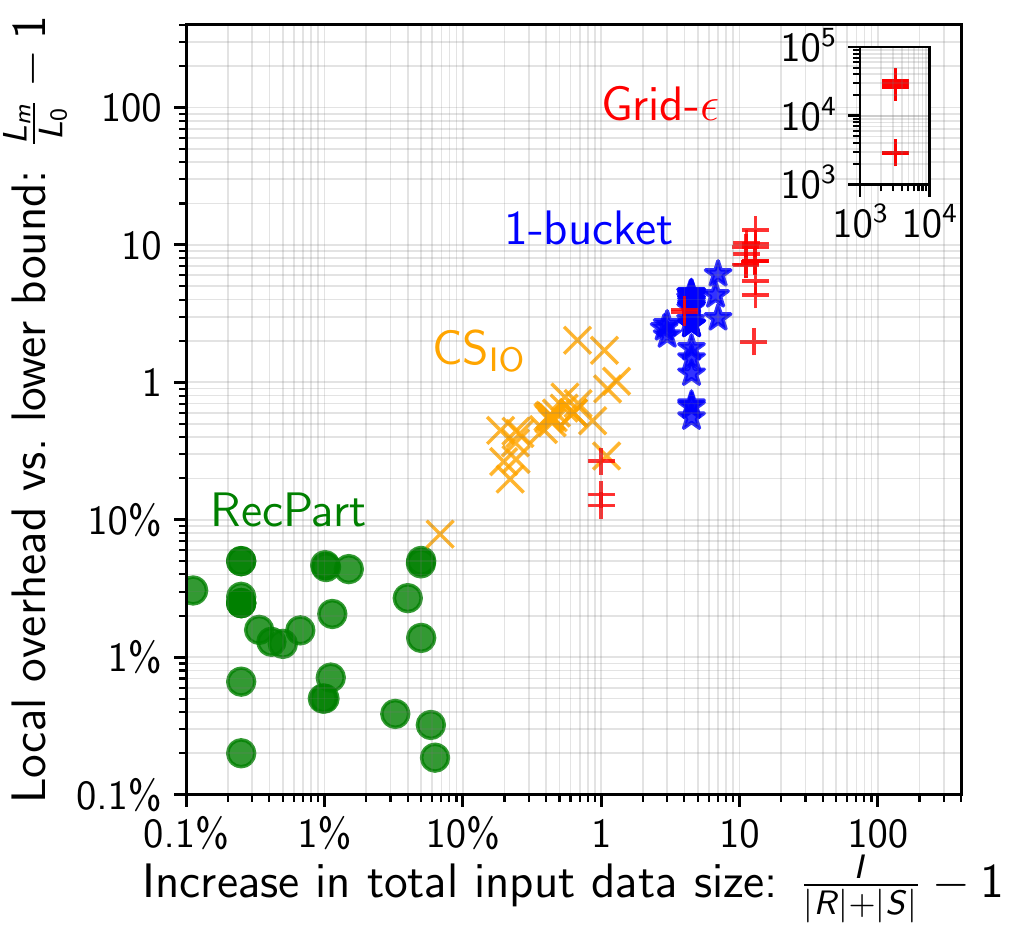}
\caption{Total input duplication (x-axis) and maximum overhead across workers (y-axis) 
for our method \method vs.\ 3 competitors.
In contrast to \cref{fig:comp-overhead}, this figure additionally includes the data points from the appendix.
As before, \method is always within 10\% of the lower bounds 
(0\% duplication and 0\% overhead).}
\label{fig:comp-overhead_appendix}
\end{figure}

\end{appendix}

\end{document}